\def\RevisionClean{1}
\def\ArxivVersion{1}
\documentclass[conference]{IEEEtran}

\usepackage{amsmath,amsthm}
\usepackage{amssymb}
\usepackage[ruled,vlined,linesnumbered]{algorithm2e}
\usepackage{booktabs,graphicx}
\usepackage{tabularx}
\usepackage{xcolor}
\usepackage[normalem]{ulem}
\usepackage{hyperref}
\usepackage{cleveref}
\usepackage{enumitem}
\usepackage{multirow}
\usepackage{pifont}
\usepackage{array}
\usepackage{makecell}
\usepackage{adjustbox}
\usepackage{xspace}
\usepackage{cite}
\usepackage{stfloats}

\usepackage{times}

\newcommand{\UC}[4]{\makecell[c]{#1\\#2\\#3\\#4}}
\newcommand{\TabTopRule}{\toprule[1.2pt]}

\newtheorem{theorem}{Theorem}

\newtheorem{proposition}{Proposition}
\newtheorem{definition}{Definition}
\theoremstyle{remark}
\newtheorem{remark}{Remark}
\theoremstyle{definition}
\newtheorem{assumption}{Assumption}

\newcommand{\sysname}{\textsc{AEGIS}\xspace}

\newcommand{\wte}{\mathbf{W}_{\mathrm{WTE}}}

\newif\ifrevisiontrack
\ifdefined\RevisionClean
  \revisiontrackfalse
\else
  \revisiontracktrue
\fi
\DeclareRobustCommand{\revdel}[1]{\ifrevisiontrack\textcolor{red}{#1}\hspace{0.25em}\fi}
\DeclareRobustCommand{\revadd}[1]{\ifrevisiontrack\textcolor{blue}{#1}\else#1\fi}
\DeclareRobustCommand{\revtag}[1]{\ifrevisiontrack
  \mbox{\textcolor{black}{\textsuperscript{\textbf{(#1)}}}}\fi}
\newenvironment{revaddblock}{\ifrevisiontrack\color{blue}\fi}{}

\makeatletter
\@ifundefined{paragraph}{\newcommand{\paragraph}[1]{\smallskip\noindent\textbf{#1.}\ }}{}
\makeatother

\begin{document}

\title{AEGIS: Attention-Embedding Gradient Isolation Shield ---
Triple-Channel Gradient Masking for Privacy-Preserving
Federated LLM Fine-Tuning}

\ifdefined\ArxivVersion
\author{
\IEEEauthorblockN{Ye Tao\IEEEauthorrefmark{1},
Hong Shen\IEEEauthorrefmark{1},
Hui Tian\IEEEauthorrefmark{2},
Xin Wang\IEEEauthorrefmark{1}, and
Can Wang\IEEEauthorrefmark{2}}
\IEEEauthorblockA{\IEEEauthorrefmark{1}Central Queensland University\\
\{y.tao,h.shen,x.wang3\}@cqu.edu.au}
\IEEEauthorblockA{\IEEEauthorrefmark{2}Griffith University\\
\{hui.tian,c.wang\}@griffith.edu.au}
}
\else
\author{\IEEEauthorblockN{Anonymous Authors}}
\fi

\maketitle

\begin{abstract}
Gradient inversion attacks recover private training text from gradients shared in federated learning, posing a serious threat to collaborative model training.
Through our analysis of transformer gradient structure, we identify three channels through which private token information leaks: the attention output projection gradient exposes a low-rank subspace that encodes input embeddings (Channel~1), the embedding gradient's row-norm sparsity directly reveals which tokens are present (Channel~2), and the MLP expansion gradient carries a recoverable subspace signal analogous to Channel~1 (Channel~3).
State-of-the-art attacks exploit these channels analytically to achieve near-exact token recovery in seconds.
Existing defences address at most one channel and either degrade model utility or leave the remaining structural signals intact.

We introduce \sysname (\textbf{A}ttention-\textbf{E}mbedding \textbf{G}radient \textbf{I}solation \textbf{S}hield), a lightweight defence that closes all three analytical channels with three backward-path operations requiring no architectural changes: freezing attention projection parameters eliminates Channel~1 by construction, calibrated noise injection into the embedding gradient destroys Channel~2's token-presence signal, and analogous per-block noise injection into the MLP expansion gradient masks Channel~3. The same masked gradient drives both the local optimiser step and the server export, so no clean signal is retained on either side.

Evaluated across \textbf{11 models} and \textbf{six datasets}, \sysname reduces token recovery rates to near zero against a range of gradient inversion attacks, both analytical and optimisation-based, while preserving or improving model utility. We provide formal guarantees for Channels~1 and~2 and validate the full defence empirically against adaptive adversaries with complete knowledge of the mechanism.
\end{abstract}


\section{Introduction}
\label{sec:intro}

Federated learning (FL)~\cite{fl_original,fedavg} enables multiple participants to collaboratively fine-tune large language models (LLMs) without centralising raw training data.
Participants compute gradients on their private corpora and transmit them to a coordinating server, which aggregates the updates and returns an improved global model.
Because raw text never leaves the client device, FL is widely assumed to preserve participant privacy.

Recent work on \emph{gradient inversion attacks} (GIAs)~\cite{dlg,geiping2021inverting,dager,lamp,grab,somp,fedspyllm} challenges this assumption: shared gradients encode enough information to reconstruct training text at the token level.
Early optimisation-based attacks~\cite{dlg,tag,lamp} iteratively minimise the distance between dummy and true gradients, but they are sensitive to initialisation and do not scale to large vocabularies.
More recent \emph{analytical} attacks bypass optimisation entirely: \textsc{DAGER}~\cite{dager} achieves near-exact token-set recovery in closed form by exploiting the algebraic structure of transformer gradients, obtaining ROUGE-1 above 0.65 on standard benchmarks; \textsc{SOMP}~\cite{somp} and \textsc{FedSpy-LLM}~\cite{fedspyllm} extend the same paradigm with head-structured pooling and rank-deficient embedding-gradient decomposition respectively.

\textsc{DAGER} succeeds because transformer language models expose several independent gradient channels, each sufficient for partial recovery and jointly sufficient for near-perfect reconstruction:
\begin{itemize}[leftmargin=1.2em]
  \item \textbf{Channel~1 (Attention SVD).} The gradient of the attention output projection $\nabla \mathbf{W}_o^{(\ell)}$ has a column space spanned by the value vectors of the input tokens. SVD of this gradient identifies the rank-$T$ subspace corresponding to the $T$ input tokens, and projecting all $V$ vocabulary embeddings onto this subspace distinguishes true tokens from the rest.
  \item \textbf{Channel~2 (Embedding Row Norms).} The embedding layer gradient $\nabla \wte$ has non-zero rows \emph{only} for tokens present in the training input. This structural sparsity provides a binary signal: zero gradient norm = token absent, non-zero norm = token present. The signal is immune to downstream architectural modifications.
  \item \textbf{Channel~3 (MLP Expansion SVD, adaptive).} The first MLP projection $\nabla \mathbf{W}_{\mathrm{fc}}^{(\ell)}$ has the same low-rank structure as Channel~1; an adaptive adversary that knows Channels~1--2 are masked can rerun DAGER's SVD pipeline on the MLP gradient to recover the same token set.
\end{itemize}

Differential privacy (DP-SGD)~\cite{dp_fl} adds calibrated noise to clipped gradients, but prior work on DP fine-tuning of language models shows that the privacy budgets required to meaningfully resist reconstruction attacks incur substantial utility cost~\cite{yu2022dpfinetuning}. As we show in Section~\ref{sec:baselines}, even large noise magnitudes ($\sigma{=}0.1$) fail to suppress DAGER's analytical recovery channel.
\textsc{Soteria}~\cite{soteria} prunes gradient components but targets convolutional vision models.
Data-level defences obfuscate training tokens (e.g.\ via embedding-proximate substitutions) before optimisation, which can suppress exact token recovery but typically incurs high preprocessing cost and does not remove the MLP gradient channel available to adaptive adversaries.
Our analysis of DAGER's channel structure (Section~\ref{sec:method_evolution}) indicates that single-channel defences, whether targeted at the attention subspace, the MLP expansion gradient, or the embedding row norms in isolation, cannot succeed: closing one channel while leaving another intact provides no net protection against an analytical adversary, because the surviving channel compensates.

We propose \sysname (\textbf{A}ttention-\textbf{E}mbedding \textbf{G}radient \textbf{I}solation \textbf{S}hield), a triple-channel gradient masking mechanism that eliminates all three analytical channels simultaneously.
The design is motivated by the observation that any defence must address every exploitable channel at once; closing some channels while leaving one intact provides no net protection against an analytical adversary that reads the surviving channel.
\sysname combines three operations:
(1) freezing all attention projection parameters, so their gradients are identically zero (Channel~1);
(2) replacing the embedding gradient with a dense, calibrated-noise version whose row norms no longer discriminate present from absent tokens (Channel~2); and
(3) applying an analogous per-block flood to the MLP expansion gradient $\nabla \mathbf{W}_{\mathrm{fc}}^{(\ell)}$, calibrated so that the MLP-SVD signal singular values fall beneath the random-matrix noise bulk in the regimes we evaluate (Channel~3).
The defence is applied after backpropagation and before the optimiser step: the local optimiser uses the modified gradient, whose flooded blocks retain a scaled learning signal, and the server receives the same masked gradient.

Our contributions are:
\begin{itemize}[leftmargin=1.2em]
  \item \textbf{Triple-channel analysis.} A systematic analysis of the channel structure \textsc{DAGER} exploits, showing that it relies on two redundant default channels plus one additional channel available to an adaptive adversary, and that any one- or two-channel defence is insufficient because the surviving channel compensates (Section~\ref{sec:method_evolution}).
  \item \textbf{The \sysname method.} A lightweight, architecture-agnostic gradient masking mechanism that closes all three channels using three backward-path operations (attention-parameter freeze, calibrated embedding flood, calibrated MLP-fc flood), with no architectural changes and negligible overhead (Section~\ref{sec:method}). Attention freezing repurposes an observation from parameter-efficient fine-tuning~\cite{lora,adapter_tuning}; we show that it also eliminates \textsc{DAGER} Channel~1 exactly, and we combine it with the embedding and MLP-expansion flooding under a single calibrated Gaussian-uniformisation schema.
  \item \textbf{Formal guarantees.} A proof that attention freezing zeros the Channel~1 exported gradient (Theorem~\ref{thm:channel1}) and that embedding gradient uniformisation eliminates Channel~2's deterministic support indicator while bounding the residual active-row second-moment inflation (Theorem~\ref{thm:channel2}).
  \item \textbf{Evaluation.} We evaluate \sysname on 11 LLMs (see Table~\ref{tab:dager_ppl_exp09}) across six benchmark datasets. Attack ROUGE-1 drops from ${\geq}0.87$ (undefended) to ${\leq}0.005$ under \sysname, a relative reduction of at least $99\%$ against the attacks we consider. Test perplexity remains at or below the undefended baseline on every model.
\end{itemize}

Section~\ref{sec:related} surveys related work.
Section~\ref{sec:prelim} formalises the federated fine-tuning threat model and DAGER's attack mechanism.
Section~\ref{sec:method_evolution} presents the analysis that motivates the triple-channel design.
Section~\ref{sec:method} describes the \sysname algorithm.
Section~\ref{sec:analysis} provides theoretical guarantees.
Section~\ref{sec:experiments} reports empirical results.
Section~\ref{sec:discussion} discusses limitations and future directions.
Section~\ref{sec:conclusion} concludes.
Full proofs are deferred to Appendix~\ref{app:proofs}.
 \section{Related Work}
\label{sec:related}

This section positions \sysname within the landscape of gradient inversion attacks and defences for federated learning.

\subsection{Gradient Inversion Attacks on Language Models}

\cite{dlg} first demonstrated that a shared gradient for a single training example is sufficient to reconstruct the input via iterative optimisation in the continuous embedding space, then mapping recovered embeddings back to discrete tokens.
\cite{idlg} improved the method by analytically extracting the label from the gradient sign.
\cite{tag} introduced $L_1$ regularisation for improved token recovery on BERT-class models.
\cite{lamp} adapted gradient inversion to language models by coupling cosine-similarity gradient matching with a language-model prior, alternating between gradient descent in embedding space and beam search in token space to reconstruct coherent sentences.
These optimisation-based approaches achieve moderate recovery rates but are sensitive to initialisation and struggle with token ordering in long sequences.

A qualitatively different class of attacks bypasses iterative optimisation entirely.
\cite{dager} introduced \textsc{DAGER}, a closed-form gradient inversion algorithm for transformer language models that exploits two structural properties of transformer gradients: (1)~the low-rank structure of the attention output projection gradient $\nabla \mathbf{W}_o^{(\ell)}$, whose column space is spanned by the value vectors of the input tokens, and (2)~the sparsity of the embedding gradient $\nabla \wte$, which has non-zero rows only for tokens present in the input.
\textsc{DAGER} achieves near-exact token-set recovery (ROUGE-1 $>0.65$) on GPT-2 and BERT without any gradient descent, completing in seconds.
\cite{grab} proved that a related analytical objective can be solved in linear time under mild assumptions.
\textsc{SOMP}~\cite{somp} extends DAGER's subspace test to a head-structured pooling stage followed by an Orthogonal Matching Pursuit reconstruction over an LM-guided beam, scaling token recovery to larger batch sizes.
\textsc{FedSpy-LLM}~\cite{fedspyllm} attacks the same structural channels via rank-deficient gradient decomposition: it reads tokens from the active rows of the embedding gradient and resolves their order through partial-gradient alignment, generalising the analytical recipe across encoder, decoder, and PEFT settings.
Earlier analytical work includes \textsc{APRIL}~\cite{april}, which exploits attention and embedding structure in vision transformers, and \textsc{SPEAR}~\cite{spear}, which achieves batch-exact inversion of fully-connected layers---a result that \textsc{DAGER} extends to transformer LMs.
\cite{rgap} derived a recursive analytical formula (R-GAP) for exact gradient inversion in fully-connected networks, establishing the algebraic foundation that later analytical attacks on transformers build upon.
Together, these analytical attacks---\textsc{DAGER}, \textsc{SOMP}, and \textsc{FedSpy-LLM} on the LM side, and \textsc{APRIL}/\textsc{SPEAR}/R-GAP on the vision and FC side---establish that gradient inversion is not merely a theoretical curiosity but a practical, computationally cheap threat that demands a structural---not merely noise-based---defence.

\cite{robbing} and \cite{decep} showed that a malicious server can modify model weights to amplify gradient leakage, enabling batch-level extraction.
These attacks operate under a stronger threat model than our honest-but-curious assumption and are orthogonal to the defences considered here.

\subsection{Defences Against Gradient Inversion}

DP-SGD~\cite{dp_fl,dpsgd} clips gradients to a bounded norm and adds calibrated Gaussian noise, offering a formal $(\varepsilon, \delta)$-DP guarantee.
However, prior work on DP fine-tuning of language models shows that the privacy budgets required to meaningfully suppress analytical gradient-inversion attacks incur substantial utility cost~\cite{yu2022dpfinetuning}, and as we show in Section~\ref{sec:baselines}, even large noise magnitudes fail to suppress the analytical recovery channel.
Gradient pruning~\cite{dlg} randomly zeroes a fraction of gradient components; advanced attacks such as \textsc{GRAB}~\cite{grab} and \textsc{DAGER}~\cite{dager} remain effective against pruned gradients.

Gradient compression methods~\cite{gradient_compression} were designed to reduce communication cost, not to provide privacy.
\cite{dager} demonstrated that compressed gradients remain fully invertible by analytical attacks.

\textsc{Soteria}~\cite{soteria} prunes gradient components that encode representation-sensitive information, targeting convolutional vision models.
\textsc{InstaHide}~\cite{instahide} mixes training examples at the data level, but its security was subsequently broken~\cite{carlini2021instahide}.
Differentially Private Word Embeddings (DPWE)~\cite{dpwe} add Laplace noise in the embedding space during inference, not during gradient computation.

A complementary line of work perturbs or substitutes training tokens before fine-tuning (e.g.\ embedding-proximate replacements chosen by search in token or hidden space).
Such methods can reduce exact recovery under analytical attacks but typically require per-token preprocessing, alter the training distribution, and leave other gradient channels (notably MLP paths) available to adaptive adversaries.
They also face an inherent \emph{semantic leakage} trade-off: if the obfuscated text preserves sufficient semantic fidelity for the language model to learn the target task, high-level semantic structure may still be reflected in the gradients.
By contrast, \sysname defends at the algebraic gradient layer, misaligning the subspaces exploited by closed-form inversion without modifying client text.

SecAgg~\cite{secagg} cryptographically hides individual gradients from the server; the original protocol has $O(n^2)$ per-round communication cost, though more recent graph-based variants achieve near-linear per-client overhead~\cite{bell2020secagg}.
SecAgg does not protect against a malicious server or colluding participants.

\subsection{Parameter-Efficient Fine-Tuning}

LoRA~\cite{lora} and adapter tuning~\cite{adapter_tuning} freeze most model parameters and train low-rank additive updates, demonstrating that pretrained transformers can be effectively adapted with fewer than 1\% of parameters trainable.
\sysname leverages the same insight: freezing the attention parameters does not meaningfully degrade fine-tuning quality, because pretrained attention patterns encode robust contextual representations that transfer well to downstream tasks.
The difference is that \sysname freezes attention for \emph{privacy}, not for parameter efficiency.

 \section{Preliminaries}
\label{sec:prelim}

This section introduces the federated fine-tuning setting, the transformer gradient structure common to decoder and encoder LMs, and the two analytical channels exploited by closed-form gradient inversion attacks (a third, adaptive channel is introduced in Section~\ref{sec:method_evolution}).

\subsection{Federated Language Model Fine-Tuning}

Let $\mathcal{V}$ denote a vocabulary of size $V$ and let $f_\theta : \mathcal{V}^* \to \mathbb{R}$ be a causal language model parameterised by $\theta$.
In federated fine-tuning~\cite{fl_original}, $K$ clients each hold a private text corpus $\mathcal{D}_k = \{s_1, \ldots, s_{n_k}\}$.
At each round, client~$k$ computes the gradient
\[
  \mathbf{g}_k = \nabla_\theta \mathcal{L}(f_\theta, s)
\]
for a mini-batch sample $s \sim \mathcal{D}_k$ and sends $\mathbf{g}_k$ to a central server.
The server aggregates the received gradients (e.g.\ via FedAvg~\cite{fl_original}) and distributes updated weights.

\paragraph{Threat model.}
We consider the FedSGD setting where each client shares the gradient computed on
a \emph{single mini-batch} (one local step) after each global round, following the standard
threat model adopted by \textsc{DAGER}~\cite{dager}, \textsc{SOMP}~\cite{somp},
\textsc{FedSpy-LLM}~\cite{fedspyllm}, and \textsc{LAMP}~\cite{lamp}.
An honest-but-curious, non-adaptive adversary controls the server or eavesdrops
on the communication channel, receiving $\mathbf{g}_k$ in cleartext.
Specifically:
\begin{itemize}[leftmargin=1.4em,itemsep=0pt,topsep=2pt]
  \item The adversary has full knowledge of the model architecture and parameters~$\theta$.
  \item The adversary does \emph{not} know which defence mechanism (if any) the client is
        using, i.e.\ the defence mechanism itself is not part of the adversary's prior.
  \item The adversary mounts an analytical gradient inversion attack, exploiting the algebraic structure of transformer gradients to recover input tokens without iterative optimisation.
  \item Only a single client is evaluated; multi-client privacy amplification is orthogonal.
\end{itemize}

This single-client, single-step evaluation setting is standard in the GIA
literature~\cite{dlg,dager}: success against one client implies vulnerability for all.

\subsection{Transformer Gradient Structure}
\label{sec:gpt2_structure}

We describe the gradient structure using GPT-2~\cite{gpt2} as a concrete reference, as its architecture is representative of the decoder-only family (GPT-2, LLaMA, Gemma).
The analytical channels described below arise from properties of the transformer block that are shared by encoder architectures (e.g.\ BERT~\cite{bert}) as well; we note encoder-specific differences where they arise.

A GPT-2-class model consists of $L$ transformer blocks.
Each block~$\ell$ contains:
\begin{itemize}[leftmargin=1.2em]
  \item \textbf{Attention projections:} query-key-value projection $\mathbf{W}_{\mathrm{QKV}}^{(\ell)} \in \mathbb{R}^{d \times 3d}$ and output projection $\mathbf{W}_o^{(\ell)} \in \mathbb{R}^{d \times d}$.
  \item \textbf{MLP feed-forward:} expansion $\mathbf{W}_{\mathrm{fc}}^{(\ell)} \in \mathbb{R}^{d \times 4d}$ and contraction $\mathbf{W}_{\mathrm{proj}}^{(\ell)} \in \mathbb{R}^{4d \times d}$.
  \item \textbf{Layer normalisation:} two LayerNorm modules per block.
\end{itemize}
The embedding matrix $\wte \in \mathbb{R}^{V \times d}$ maps token indices to hidden representations of dimension~$d$.
The LM head (linear layer producing logits over $\mathcal{V}$) typically shares weights with $\wte$ (weight tying).

\paragraph{Encoder LMs.}
BERT-class models~\cite{bert} share the same block structure.
The key difference is the absence of a causal attention mask, which removes the decoder's greedy-extension property and requires exhaustive search over candidate sequences.
Channel~1 (attention-output SVD) and Channel~2 (embedding row-norm sparsity) arise identically in encoders, so \sysname's attention-freeze and embedding-uniformisation operations apply without modification; the MLP-fc uniformisation (Channel~3) likewise transfers directly.
Our evaluation in Section~\ref{sec:experiments} confirms that \sysname reduces DAGER ROUGE-1 to $\leq 0.006$ on BERT-base and BERT-large, matching the decoder-side results.

\subsection{Analytical Attacks: Two Gradient Channels}
\label{sec:dager_attack}

Analytical gradient inversion attacks~\cite{dager,grab,somp,fedspyllm} recover the input token set $\{t_1, \ldots, t_T\}$ from a single gradient $\mathbf{g}$ in closed form, without iterative optimisation, by exploiting two structurally independent channels in the transformer gradient.
\textsc{DAGER} and \textsc{SOMP} primarily target Channel~1 (attention SVD scoring) with different downstream pipelines, while \textsc{FedSpy-LLM} primarily targets Channel~2 (embedding row-norm sparsity) with an explicit ordering stage; all three rely on the same two structural properties.

\begin{definition}[Channel 1: Attention Subspace Scoring]
\label{def:channel1}
For each block $\ell$, let $\mathbf{U}^{(\ell)}$ be an orthonormal basis for the
range of $\nabla \mathbf{W}_o^{(\ell)}$, equivalently the left singular vectors
whose singular values are non-zero.
If $\nabla \mathbf{W}_o^{(\ell)}=\mathbf{0}$, this basis has no columns and the
projection below is the zero projection.
The subspace residual score of vocabulary token $v$ is:
\[
  r_v^{(\ell)} = \frac{\|\mathbf{e}_v - \mathbf{U}^{(\ell)}{\mathbf{U}^{(\ell)}}^\top \mathbf{e}_v\|}{\|\mathbf{e}_v\|},
\]
where $\mathbf{e}_v \in \mathbb{R}^d$ is the embedding vector for token $v$.
Tokens with small residual (i.e.\ whose embedding lies in the attention output gradient subspace) are candidates for the input set.
\end{definition}

\begin{definition}[Channel 2: Embedding Row-Norm Scoring]
\label{def:channel2}
Consider a single input sequence $s = (t_1, \ldots, t_T)$ (batch size $B=1$).
Let $\mathcal{S}(s)=\{v\in\mathcal{V}: v \text{ appears in } s\}$ be the active
vocabulary set.
The embedding layer gradient $\nabla \wte \in \mathbb{R}^{V \times d}$ has the structural property that row $v$ is non-zero \emph{if and only if} $v\in\mathcal{S}(s)$.
The row-norm score is:
\[
  n_v = \|\nabla \wte[v, :]\|_2.
\]
The top-$|\mathcal{S}(s)|$ tokens by row-norm score form the recovered token set
when the attacker knows the number of distinct active rows.
This signal is a binary indicator: $n_v > 0 \iff v \in \mathcal{S}(s)$.
\end{definition}

\begin{revaddblock}
\begin{remark}[Practical scope of Definition~\ref{def:channel2}]
Definition~\ref{def:channel2} assumes an unmasked, non-cancelling input-lookup
gradient. Loss structure, masking, special-token handling, or finite precision can instead
make an active row zero. Algorithm~\ref{alg:aegis} uses a relative-norm rule
and still floods any excluded row; because all $V$ rows become dense, an
occasional zero active row does not restore a clean absent/present signal.
\revtag{4}
\end{remark}
\end{revaddblock}

\begin{remark}[Batch size $B>1$]
\label{rem:batch-channel2}
For a mini-batch of $B$ sequences, $\nabla\wte[v,:]$ is non-zero iff token $v$ appears in
\emph{any} of the $B$ sequences. Channel~2 therefore recovers the \emph{union} $\bigcup_{b=1}^{B}\mathcal{S}(s_b)$ rather than per-sequence tokens; analytical attacks address
this via Channel~1 attention-subspace filtering combined with positional disambiguation. Our experiments use $B=8$, and the
uniformisation defence remains valid in this regime: the operation renders all $V$ row norms comparable
regardless of how many sequences contributed non-zero rows. The formal guarantee
(Theorem~\ref{thm:channel2}) is stated for $B=1$ for notational clarity, and extends
verbatim to $B>1$ by replacing the active-row set $\mathcal{S}(s)$ with the union set.
\end{remark}

\paragraph{Combined recovery.}
Analytical attacks combine both channels: attention subspace scores from all $L$ blocks are aggregated and intersected with the embedding row-norm ranking to form the final recovered token set.
The redundancy between channels is critical: \emph{either channel alone provides partial recovery; together they achieve near-exact results.}

\paragraph{Notation summary.}
Table~\ref{tab:notation} collects the main notation used throughout the paper.

\begin{table}[t]
\centering
\caption{\textbf{Notation summary.}}
\label{tab:notation}
\small
\begin{tabular}{@{}r@{\quad}p{0.31\columnwidth}r@{\quad}p{0.31\columnwidth}@{}}
\TabTopRule
\multicolumn{1}{c}{Symbol} & \multicolumn{1}{c}{Description} &
\multicolumn{1}{c}{Symbol} & \multicolumn{1}{c}{Description} \\
\midrule
$\mathcal{V}$, $V$ & Vocabulary, its size ($|\mathcal{V}| = V$) &
$d$ & Hidden dimension \\
$L$ & Number of transformer blocks &
$T$ & Sequence length (number of input tokens) \\
$\wte$ & Embedding matrix $\in \mathbb{R}^{V \times d}$ &
$\mathbf{W}_o^{(\ell)}$ & Attention output projection in block $\ell$ \\
$\mathbf{W}_{\mathrm{fc}}^{(\ell)}$ & MLP fully-connected weight in block $\ell$ &
$\mathbf{e}_v$ & Embedding vector for token $v$ \\
$r_v^{(\ell)}$ & Subspace residual score (Channel 1) &
$n_v$ & Embedding row-norm score (Channel 2) \\
$\lambda$ & Flood scale hyperparameter &
$\rho$ & Real-token retain ratio \\
\bottomrule
\end{tabular}
\end{table}
 \section{The \sysname Method}
\label{sec:method}
\label{sec:method_evolution}

This section presents the \sysname algorithm, its design rationale, and its implementation.

\begin{figure*}[t]
  \centering
  \includegraphics[width=0.95\linewidth]{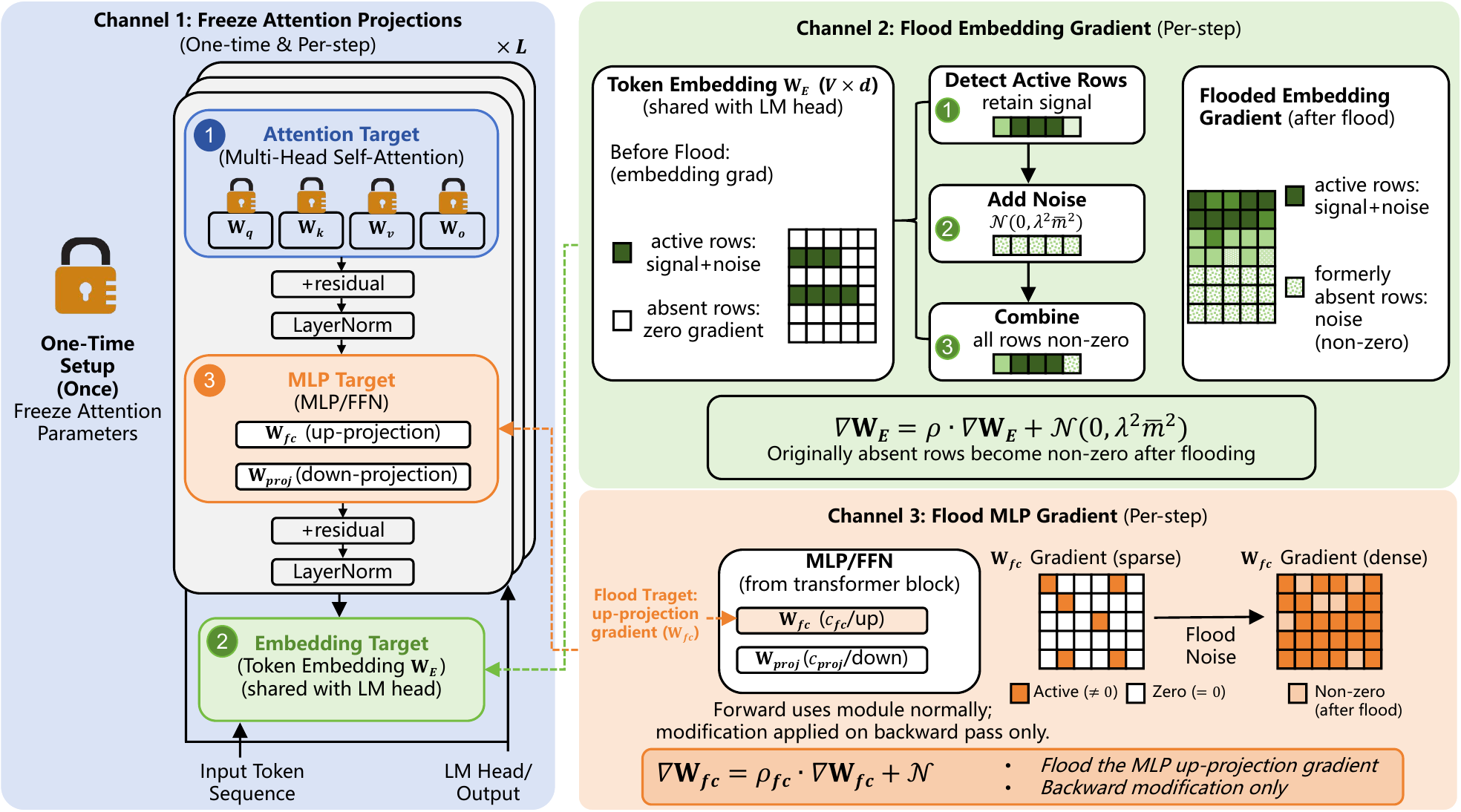}
  \caption{Overview of the \sysname mechanism. Channel~1 (the attention bottleneck) is neutrallised via freezing, while Channel~2 (embedding rows) and Channel~3 (the first MLP projection $\mathbf{W}_{\mathrm{fc}}$) are obscured through mathematically calibrated dense Gaussian uniformisation. Downstream utility is maintained via the remaining unaffected trainable parameters (e.g., $\mathbf{W}_{\mathrm{proj}}$, LayerNorm).}
  \label{fig:overview}
\end{figure*}

\subsection{Design Objectives}

\sysname targets three properties for the defended gradient $\mathcal{D}(\nabla_\theta \mathcal{L})$: (P1) privacy: no token identity information is exposed via the three analytically-exploitable channels (attention output-projection subspace, embedding row-norm sparsity, MLP expansion subspace); (P2) utility: the model converges with test perplexity comparable to the undefended baseline; and (P3) efficiency: at most $O(Vd + Ld^2)$ overhead per backward pass, which is negligible relative to the backward computation itself.

Three backward-path operations address these goals: (i) \textbf{freeze attention} (Channel~1) --- set $\nabla\mathbf{W}_{\mathrm{QKV}}^{(\ell)}=\nabla\mathbf{W}_o^{(\ell)}=\mathbf{0}$ by construction; (ii) \textbf{uniformise embedding} (Channel~2) --- replace $\nabla\wte$ with a dense, calibrated-noise gradient whose active rows are partially retained; (iii) \textbf{uniformise MLP-fc} (Channel~3) --- apply an analogous uniformisation to $\nabla\mathbf{W}_{\mathrm{fc}}^{(\ell)}$ per block, pushing the SVD signal beneath the Marchenko--Pastur noise bulk.
\revdel{Freezing attention incurs no utility cost because pretrained attention patterns transfer to downstream tasks (cf.\ LoRA~\cite{lora});}\revadd{Freezing keeps pretrained attention patterns, which helps preserve utility in our ablation (cf.\ LoRA~\cite{lora});}\revtag{1,6} the load-bearing MLP expansion uses uniformisation rather than freezing.

\subsection{Channel~1: Freeze Attention Projections}

The attention subspace channel (Definition~\ref{def:channel1}) exploits the low-rank column space of $\nabla \mathbf{W}_o^{(\ell)}$.
The simplest way to eliminate this signal is to ensure the gradient is \emph{identically zero}:

\begin{equation}
  \forall \ell \in [L],\; p \in \text{params}(\text{attn}_\ell): \quad p.requires\_grad \leftarrow False.
  \label{eq:freeze}
\end{equation}

Freezing all attention parameters ($\mathbf{W}_{\mathrm{QKV}}^{(\ell)}, \mathbf{W}_o^{(\ell)}$) across all $L$ blocks ensures that:
\begin{itemize}[leftmargin=1.2em]
  \item $\nabla \mathbf{W}_o^{(\ell)} = \mathbf{0}$ for all $\ell$, so the SVD column space is empty.
  \item The subspace residual score $r_v^{(\ell)} = 1$ for all tokens $v$, providing zero discriminative power.
\end{itemize}

\subsection{Channel~2: Embedding Gradient Uniformisation}

The embedding row-norm channel (Definition~\ref{def:channel2}) exploits the binary sparsity of $\nabla \wte$: rows corresponding to absent tokens have exactly zero gradient.
To destroy this signal, we replace the sparse gradient with a \emph{dense} version in which all $V$ rows have comparable norms:

\begin{equation}
  \tilde{\nabla}\wte[v, :] = \begin{cases}
    \rho \cdot \nabla \wte[v, :] + \mathbf{n}_v & \text{if } v \in \mathcal{S}(s) \\[3pt]
    \mathbf{n}_v & \text{otherwise}
  \end{cases}
  \label{eq:flood}
\end{equation}
where $\rho \in (0, 1]$ is the \emph{real-token retain ratio} (fraction of original gradient preserved for real tokens), and $\mathbf{n}_v \sim \mathcal{N}(\mathbf{0}, \sigma^2 \mathbf{I}_d)$ with $\sigma = \lambda \cdot \bar{m}$, where $\lambda$ is the \emph{uniformisation scale} and $\bar{m}$ is the mean gradient norm of the active (real-token) rows:
\begin{equation}
  \bar{m} = \frac{1}{|\mathcal{S}(s)|}\sum_{v \in \mathcal{S}(s)} \|\nabla \wte[v, :]\|_2.
  \label{eq:mean_norm}
\end{equation}

\paragraph{Weight tying and the LM head.}
All eleven models in our evaluation tie the LM head to $\wte$, so the backward pass accumulates a single gradient tensor used for both the embedding update and the output projection.
\sysname uniformises this shared tensor, simultaneously masking both paths.
For models with an \emph{untied} LM head, the LM-head gradient
$\nabla\mathbf{W}_{\mathrm{LM}}$ is a separate matrix whose row norms can expose
the target-token support.  We therefore apply the same row-wise
uniformisation in Eq.~\ref{eq:flood} to $\nabla\mathbf{W}_{\mathrm{LM}}$ whenever
the output embedding is not tied to $\wte$; tied heads are already covered
because they share the same gradient tensor as the input embedding.

\subsection{Channel~3: MLP Expansion Gradient Uniformisation}
\label{sec:method_channel3}

While Channels~1 and~2 are the two channels exploited by \emph{non-adaptive}
\textsc{DAGER}, Section~\ref{sec:method_channel3} identified a third
analytical channel: the MLP expansion gradient
$\nabla\mathbf{W}_{\mathrm{fc}}^{(\ell)}\in\mathbb{R}^{d\times 4d}$, whose column span
is isomorphic to that of $\nabla\mathbf{W}_o^{(\ell)}$ and therefore recoverable by
the same SVD span-estimation pipeline.
An adaptive adversary aware that Channels~1 and~2 are closed can target Channel~3
directly.

Since $\mathbf{W}_{\mathrm{fc}}^{(\ell)}$ is load-bearing for downstream-task utility,
freezing it (the Channel~1 strategy) is not viable; we therefore apply a
uniformisation defence analogous to Channel~2.
For each block $\ell\in[L]$ we replace the exported MLP-fc gradient with
\begin{equation}
  \tilde{\nabla}\mathbf{W}_{\mathrm{fc}}^{(\ell)}
  \;=\; \rho_{\mathrm{fc}} \cdot \nabla\mathbf{W}_{\mathrm{fc}}^{(\ell)}
       \;+\; \mathbf{N}^{(\ell)},
  \label{eq:flood_fc}
\end{equation}
where $\mathbf{N}^{(\ell)}\in\mathbb{R}^{d\times 4d}$ has i.i.d.\ entries
$N_{ij}\sim\mathcal{N}(0,\,\sigma_{\mathrm{fc}}^2)$ with
$\sigma_{\mathrm{fc}} = \lambda_{\mathrm{fc}}\cdot\bar{m}_{\mathrm{fc}}^{(\ell)}$
and $\bar{m}_{\mathrm{fc}}^{(\ell)}$ the mean absolute value of the entries of
$\nabla\mathbf{W}_{\mathrm{fc}}^{(\ell)}$.
The hyperparameters $(\rho_{\mathrm{fc}},\lambda_{\mathrm{fc}})$ play the same role
as $(\rho,\lambda)$ on the embedding: $\rho_{\mathrm{fc}}$ retains a scaled
fraction of the true MLP gradient for the \emph{local} optimiser step, while
$\lambda_{\mathrm{fc}}$ sets the noise floor that the exported gradient carries
to the server.
Unless otherwise noted we use the privacy-hardening endpoint
$(\rho_{\mathrm{fc}},\lambda_{\mathrm{fc}})=(0, 8.0)$ for exported MLP
gradients in adaptive-threat experiments.  The implementation applies the same
matrix-gradient uniformisation to all two-dimensional MLP weights, including
the projection matrix $\mathbf{W}_{\mathrm{proj}}^{(\ell)}$, so an adaptive
attacker cannot simply retarget from the expansion matrix to the contraction
matrix.

\paragraph{Random-matrix rationale for Eq.~\ref{eq:flood_fc}.}
An adversary performing rank-$T$ SVD on $\tilde{\nabla}\mathbf{W}_{\mathrm{fc}}^{(\ell)}$
sees a matrix whose signal component is scaled by $\rho_{\mathrm{fc}}$ and whose
noise component is full-rank with variance $\sigma_{\mathrm{fc}}^2$ per entry.
The Marchenko--Pastur law predicts that the singular-value bulk of a $d\times 4d$ i.i.d.\
Gaussian matrix extends up to $\sigma_{\mathrm{fc}}(\sqrt{d}+\sqrt{4d})
=3\sigma_{\mathrm{fc}}\sqrt{d}$, which at
$\lambda_{\mathrm{fc}}=1,\,d=768$ is $\approx 83\,\bar{m}_{\mathrm{fc}}$;
the signal singular values are bounded by $\rho_{\mathrm{fc}}\cdot\sigma_{\max}(\nabla\mathbf{W}_{\mathrm{fc}})$.
Thus, when the scaled signal singular values lie below this noise edge, the
SVD stage has no stable signal spike to lock onto.
This is a design criterion and empirical prediction rather than a theorem in
this paper; a formal robust-PCA lower bound is left to future work
(Limitation~\ref{lim:adaptive}).
The empirical adaptive-attack results in Section~\ref{sec:experiments} confirm
this prediction: running adaptive \textsc{DAGER} on $\nabla\mathbf{W}_{\mathrm{fc}}$
after \sysname uniformisation yields ROUGE-1 in the same near-zero regime as the
non-adaptive attack.

\subsection{Full Algorithm}

\sysname is applied client-side as a backward-path hook with no changes to the model architecture or training loop.
Before training begins, all attention projection parameters are frozen (Eq.~\ref{eq:freeze}), so their gradients are identically zero throughout fine-tuning.
After each backward pass, the embedding gradient $\nabla\wte$ is replaced with the uniformised version $\tilde{\nabla}\wte$ (Eq.~\ref{eq:flood}), and for each block $\ell$ the MLP-fc gradient $\nabla\mathbf{W}_{\mathrm{fc}}^{(\ell)}$ is replaced with $\tilde{\nabla}\mathbf{W}_{\mathrm{fc}}^{(\ell)}$ (Eq.~\ref{eq:flood_fc}).
The local optimiser then steps on the modified gradients; the server receives only the uniformised export, with attention gradient slots set to zero.

Crucially, the intervention is applied \emph{after} backpropagation and \emph{before} the optimiser step, so the learning signal for the retained channels ($\rho\cdot\nabla\wte$ and $\rho_{\mathrm{fc}}\cdot\nabla\mathbf{W}_{\mathrm{fc}}^{(\ell)}$) propagates to the local weights while the exported gradient exposes only calibrated Gaussian noise to the server.
\revadd{Algorithm~\ref{alg:aegis} summarises the complete client-side
procedure.}\revtag{4}

\begin{algorithm}[htbp]
\DontPrintSemicolon
\caption{\sysname: Triple-Channel Gradient Masking}
\label{alg:aegis}

\KwIn{Model $f_\theta$ with $L$ blocks; hyperparameters $\rho,\lambda,\rho_{\mathrm{fc}},\lambda_{\mathrm{fc}}$; training set $\mathcal{D}$}
\KwOut{Privacy-preserving fine-tuned model; masked gradient $\tilde{\nabla}\theta$ for server}

\BlankLine
\For{$\ell \leftarrow 1$ \KwTo $L$}{
  $\mathbf{W}_{\mathrm{QKV}}^{(\ell)}.\texttt{requires\_grad} \leftarrow \texttt{False}$\;
  $\mathbf{W}_{o}^{(\ell)}.\texttt{requires\_grad} \leftarrow \texttt{False}$
}

\BlankLine
\For{\textup{each mini-batch} $B \subseteq \mathcal{D}$}{
  Compute $\mathcal{L}(\theta;\,B)$ via forward pass\;
  Compute $\nabla_\theta\mathcal{L}$ via backward pass

  \BlankLine
  \tcp{Channel~2: embedding gradient uniformisation}
  $\mathcal{S} \leftarrow \bigl\{v \in \mathcal{V} : \|\nabla\wte[v,:]\|_2 > 0.01 \cdot \max_{u}\|\nabla\wte[u,:]\|_2\bigr\}$
  $\bar{m} \leftarrow \dfrac{1}{|\mathcal{S}|}\displaystyle\sum_{v\in\mathcal{S}}\|\nabla\wte[v,:]\|_2$\;
  $\sigma \leftarrow \lambda\cdot\bar{m}$\;
  \For{$v \leftarrow 1$ \KwTo $|\mathcal{V}|$}{
    Sample $\mathbf{n}_v \sim \mathcal{N}(\mathbf{0},\,\sigma^2\mathbf{I}_d)$\;
    \eIf{$v \in \mathcal{S}$}{
      $\tilde{\nabla}\wte[v,:] \leftarrow \rho\cdot\nabla\wte[v,:] + \mathbf{n}_v$\;
    }{
      $\tilde{\nabla}\wte[v,:] \leftarrow \mathbf{n}_v$\;
    }
  }

  \BlankLine
  \tcp{Channel~3: MLP-fc gradient uniformisation (per block)}
  \For{$\ell \leftarrow 1$ \KwTo $L$}{
    $\bar{m}_{\mathrm{fc}}^{(\ell)} \leftarrow \mathrm{mean}\!\left(|\nabla\mathbf{W}_{\mathrm{fc}}^{(\ell)}|\right)$\;
    $\sigma_{\mathrm{fc}} \leftarrow \lambda_{\mathrm{fc}}\cdot\bar{m}_{\mathrm{fc}}^{(\ell)}$\;
    Sample $\mathbf{N}^{(\ell)} \sim \mathcal{N}\!\left(\mathbf{0},\,\sigma_{\mathrm{fc}}^2\right)^{d\times 4d}$\;
    $\tilde{\nabla}\mathbf{W}_{\mathrm{fc}}^{(\ell)} \leftarrow \rho_{\mathrm{fc}}\cdot\nabla\mathbf{W}_{\mathrm{fc}}^{(\ell)} + \mathbf{N}^{(\ell)}$\;
  }

  \BlankLine
  Perform local optimiser step using $\tilde{\nabla}_\theta\mathcal{L}$ (post-mask gradient)\;
  Export $\tilde{\nabla}_\theta\mathcal{L}$ to server
}
\end{algorithm}
 \section{Theoretical Analysis}
\label{sec:analysis}

This section states the exact parts of the \sysname security argument: Channel~1
has zero exported gradient under attention freezing, Channel~2 loses its
deterministic support signal under embedding flooding, and Channel~3 receives an
exact MLP-fc flooding decomposition with a Frobenius-moment bound.
Claims that depend on high-dimensional order statistics or random-matrix
denoising are treated as design rationale and empirical predictions, not as
theorem statements.
The corresponding proofs are collected in Appendix~\ref{app:proofs}.
We additionally provide Lean~4~\cite{demoura2021lean4} formalizations of the theorem and proposition
cores in the supplementary material.

\subsection{Channel~1: Attention Freezing Guarantee}

\begin{theorem}[Channel~1 Elimination --- Exact]
\label{thm:channel1}
Let $f_\theta$ be a GPT-2-class model with $L$ transformer blocks.
Let $\widehat{\nabla}$ denote the gradient tensor exported by the automatic
differentiation engine after applying the frozen-parameter rule, equivalently
the zero extension of the gradient on the reduced trainable parameter space.
Assume every vocabulary embedding is non-zero.
If all attention parameters $\{\mathbf{W}_{\mathrm{QKV}}^{(\ell)}, \mathbf{W}_o^{(\ell)}\}_{\ell=1}^L$ are frozen (i.e.\ $\texttt{requires\_grad} = \texttt{False}$), then for any training input $s = (t_1, \ldots, t_T)$ and any differentiable loss function $\mathcal{L}$:
\begin{equation}
  \widehat{\nabla}_{\mathbf{W}_o^{(\ell)}} \mathcal{L} = \mathbf{0} \in \mathbb{R}^{d \times d}, \quad \forall \ell \in [L].
\end{equation}
Consequently, under the non-zero-singular-subspace convention in
Definition~\ref{def:channel1}, the subspace residual score satisfies:
\begin{equation}
  r_v^{(\ell)} = 1, \quad \forall v \in \mathcal{V}, \; \forall \ell \in [L].
\end{equation}
The Channel~1 scoring provides zero discriminative power: no token can be distinguished from any other.
\end{theorem}

\begin{remark}[Forward pass is unaffected]
\label{rem:forward}
Freezing attention parameters sets their gradients to zero but does not remove them from the computational graph.
The forward pass computes $\text{Attn}(\mathbf{Q}, \mathbf{K}, \mathbf{V})$ using the pretrained weights normally; only the backward pass skips gradient accumulation for these parameters.
The attention mechanism continues to provide contextual representations to downstream MLP layers, maintaining representational capacity.
\end{remark}

\subsection{Channel~2: Embedding Flooding Guarantee}

\begin{assumption}[Bounded Active-Row Norm Deviation]
\label{asm:norm-dev}
For a sequence $s$, let $\mathcal{S}(s)=\{v\in\mathcal{V}:v\text{ appears in }s\}$
and assume $|\mathcal{S}(s)|>0$.
There exists a constant $C \geq 1$ such that for all active rows
$v \in \mathcal{S}(s)$:
\[
  \|\nabla \wte[v,:]\|_2 \leq C \cdot \bar{m},
\]
where $\bar{m}$ is the mean active-row norm (Eq.~\ref{eq:mean_norm}).
\end{assumption}

\begin{remark}
Assumption~\ref{asm:norm-dev} holds with $C=1$ when all active-row norms equal the mean;
the bound degrades through the ratio $\rho^2C^2/(d\lambda^2)$, which reaches order one
only near $C\approx\lambda\sqrt{d}/\rho\approx92$ at $\rho{=}0.3,\lambda{=}1,d{=}768$.
The R-1\,$\leq$\,0.005 values in Tables~\ref{tab:dager_ppl_exp09}--\ref{tab:dager_ppl_exp09_enc}
provide indirect evidence that all evaluated models lie well within the small-inflation
regime; a direct measurement of $C$ (logging $\max\|g_v\|/\bar{m}$ during training) is
a natural follow-up experiment.
\end{remark}

\begin{theorem}[Channel~2 Sparsity Elimination and Norm-Moment Bound]
\label{thm:channel2}
Let $\nabla \wte \in \mathbb{R}^{V \times d}$ be the embedding gradient for a sequence
$s$ with nonempty active set $\mathcal{S}(s)$, and suppose
$\|\nabla \wte[v,:]\|_2 > 0$ iff $v \in \mathcal{S}(s)$.
Let $\lambda>0$, $\rho\geq0$, $\bar{m}$ be defined by Eq.~\ref{eq:mean_norm}, and
let the flood vectors be independent Gaussians
$\mathbf{n}_v\sim\mathcal{N}(\mathbf{0},\sigma^2\mathbf{I}_d)$ with
$\sigma=\lambda\bar{m}$.
After \sysname flooding (Eq.~\ref{eq:flood}):
\begin{enumerate}[leftmargin=1.5em]
  \item \textbf{All rows non-zero:} $\|\tilde{\nabla}\wte[v,:]\|_2 > 0$ for all $v \in \mathcal{V}$ with probability~1.
  \item \textbf{Exact second moments:} For non-token rows ($v \notin \mathcal{S}(s)$):
  \[
    \mathbb{E}[\|\tilde{\nabla}\wte[v,:]\|_2^2] = d \cdot \lambda^2 \cdot \bar{m}^2.
  \]
  For active rows ($v \in \mathcal{S}(s)$):
  \[
    \mathbb{E}[\|\tilde{\nabla}\wte[v,:]\|_2^2]
    = \rho^2\|\nabla\wte[v,:]\|_2^2 + d\lambda^2\bar{m}^2.
  \]
  \item \textbf{Bounded relative inflation:} Under Assumption~\ref{asm:norm-dev},
  every active row satisfies
  \begin{equation}
    0 \leq
    \frac{
      \mathbb{E}[\|\tilde{\nabla}\wte[v,:]\|_2^2]
      - d\lambda^2\bar{m}^2
    }{
      d\lambda^2\bar{m}^2
    }
    \leq \frac{\rho^2 C^2}{d\lambda^2}.
    \label{eq:relative_excess}
  \end{equation}
\end{enumerate}
Thus the exact zero-vs-nonzero Channel~2 indicator is eliminated; any remaining
row-norm signal is only the bounded moment shift in Eq.~\ref{eq:relative_excess}.
\end{theorem}

\subsection{Channel~3: MLP-fc Flooding Guarantee}

\begin{theorem}[Channel~3 MLP-fc Flooding Moment Bound]
\label{thm:channel3}
For a transformer block $\ell$, let
$\mathbf{G}^{(\ell)}=\nabla\mathbf{W}_{\mathrm{fc}}^{(\ell)}
\in\mathbb{R}^{d\times4d}$ be the MLP expansion gradient and let
$\bar{m}_{\mathrm{fc}}^{(\ell)}>0$ be the calibration scale in
Eq.~\ref{eq:flood_fc}.
Let $\lambda_{\mathrm{fc}}>0$, $\rho_{\mathrm{fc}}\in[0,1]$, and let
$\mathbf{N}^{(\ell)}$ have independent entries
$N_{ij}^{(\ell)}\sim\mathcal{N}(0,\sigma_{\mathrm{fc}}^2)$ with
$\sigma_{\mathrm{fc}}=\lambda_{\mathrm{fc}}\bar{m}_{\mathrm{fc}}^{(\ell)}$.
After MLP-fc flooding,
\[
  \widetilde{\mathbf{G}}^{(\ell)}
  = \rho_{\mathrm{fc}}\mathbf{G}^{(\ell)}+\mathbf{N}^{(\ell)}.
\]
Conditional on the batch:
\begin{enumerate}[leftmargin=1.5em]
  \item \textbf{Biased-gradient decomposition:}
  with $\mathbf{B}^{(\ell)}=(\rho_{\mathrm{fc}}-1)\mathbf{G}^{(\ell)}$,
  \[
    \widetilde{\mathbf{G}}^{(\ell)}
    = \mathbf{G}^{(\ell)}+\mathbf{B}^{(\ell)}+\mathbf{N}^{(\ell)},
    \qquad
    \mathbb{E}[\widetilde{\mathbf{G}}^{(\ell)}-\mathbf{G}^{(\ell)}]
    = \mathbf{B}^{(\ell)},
  \]
  and
  $\|\mathbf{B}^{(\ell)}\|_F
    =(1-\rho_{\mathrm{fc}})\|\mathbf{G}^{(\ell)}\|_F$.
  \item \textbf{Exact Frobenius second moment:}
  \[
    \mathbb{E}\!\left[\|\widetilde{\mathbf{G}}^{(\ell)}\|_F^2\right]
    = \rho_{\mathrm{fc}}^2\|\mathbf{G}^{(\ell)}\|_F^2
      + 4d^2\lambda_{\mathrm{fc}}^2
        \bigl(\bar{m}_{\mathrm{fc}}^{(\ell)}\bigr)^2 .
  \]
  \item \textbf{Bounded relative inflation:}
  if $\|\mathbf{G}^{(\ell)}\|_F
      \leq C_{\mathrm{fc}}\bar{m}_{\mathrm{fc}}^{(\ell)}$ for some
  $C_{\mathrm{fc}}\geq0$, then
  \begin{equation}
    0 \leq
    \frac{
      \mathbb{E}[\|\widetilde{\mathbf{G}}^{(\ell)}\|_F^2]
      - 4d^2\lambda_{\mathrm{fc}}^2
        (\bar{m}_{\mathrm{fc}}^{(\ell)})^2
    }{
      4d^2\lambda_{\mathrm{fc}}^2
        (\bar{m}_{\mathrm{fc}}^{(\ell)})^2
    }
    \leq
    \frac{\rho_{\mathrm{fc}}^2C_{\mathrm{fc}}^2}
         {4d^2\lambda_{\mathrm{fc}}^2}.
    \label{eq:channel3_relative_excess}
  \end{equation}
\end{enumerate}
\end{theorem}

\section{Empirical Evaluation}
\label{sec:experiments}

We scale the evaluation of \sysname to eleven language models across four architectural families
(three decoder-only groups plus the BERT encoder pair),
six benchmark datasets, and an orders-of-magnitude range of parameter counts (110M -- 13B),
using the official DAGER implementation~\cite{dager}.

\subsection{Experimental Setup}
\label{sec:setup}

\paragraph{Models.}
We evaluate \textbf{eleven} models in total: nine decoder-only causal LMs spanning three families,
and two masked encoders from a fourth family, the \textbf{BERT family} (BERT-base and BERT-large),
expanding the grid beyond nine causal-decoder checkpoints alone.
\begin{itemize}[leftmargin=1.4em, itemsep=0pt, topsep=2pt]
  \item \textbf{GPT-2 family}~\cite{gpt2}: GPT-2 (124M), GPT-2-medium (345M),
        GPT-2-large (774M), GPT-2-XL (1.5B);
  \item \textbf{LLaMA family}~\cite{llama2,llama3}: LLaMA-2-7B, LLaMA-2-13B$^\dagger$,
        LLaMA-3.1-8B;
  \item \textbf{Gemma family}~\cite{gemma2}: Gemma-2-2B, Gemma-2-9B;
  \item \textbf{BERT family}~\cite{bert}: BERT-base (110M) and BERT-large (340M),
        attacked via DAGER's encoder span-estimation path~\cite{dager}.
\end{itemize}

\paragraph{Datasets.}
Each model is fine-tuned on six benchmarks covering diverse domains:
Rotten Tomatoes~\cite{rottenTomatoes} (film sentiment),
Emotion~\cite{emotionDataset} (multi-class affect),
Financial PhraseBank~\cite{finansBenchmark} (financial sentiment),
WikiText-2~\cite{wikitext2} (general language modelling),
DialogSum~\cite{dialogsum} (dialogue summarisation),
and CNN/DailyMail~\cite{cnndm} (news summarisation).
Each dataset is used with up to 100{,}000 randomly sampled training examples per run.

\paragraph{Attack configuration.}
We use the \emph{official} DAGER codebase (Petrov et al., NeurIPS 2024~\cite{dager}),
which performs SVD span-estimation (L1) followed by beam-search decoding (L2).
Under \sysname, attention projections are frozen so that DAGER's default
Channel~1 SVD path sees $\nabla\mathbf{W}_o = \mathbf{0}$ and collapses to
embedding-gradient recovery on the uniformised $\nabla\wte$ (Channel~2), which is
uniformised as predicted by Theorem~\ref{thm:channel2}.
To evaluate the adaptive MLP-SVD adversary (Channel~3), we additionally run a
modified DAGER pipeline in which the span-estimation stage is run on
$\tilde{\nabla}\mathbf{W}_{\mathrm{fc}}^{(\ell)}$ (the uniformised MLP-fc gradient
produced by Eq.~\ref{eq:flood_fc}) instead of on $\nabla\mathbf{W}_o^{(\ell)}$;
this is the strongest adaptive attack the three-channel analysis suggests.

\paragraph{Training.}
All models are fine-tuned with the AdamW optimizer~\cite{adamw} for up to
8 epochs or 12{,}000 gradient steps (whichever is reached first),
($\mathrm{lr}{=}5\!\times\!10^{-5}$, batch size 8, linear warmup for
the first $10\%$ of steps, early stopping on validation loss with patience 5,
evaluation every 200 steps),
\revdel{with BF16 mixed-precision training on a single NVIDIA H100 GPU.}
\revadd{on one NVIDIA H100 GPU; GPT-2/XL use FP16 autocast with FP32 weights,
while LLaMA-2-7B uses BF16 weights.}\revtag{3}
A single fine-tuning seed is used per (model,~dataset) cell
(variance reporting is flagged in Limitation~\ref{lim:seeds}).
Under \sysname, attention projection parameters are frozen before the
optimizer is instantiated; gradient uniformisation is injected after every backward
pass, before the optimizer step.

\paragraph{Hyperparameters.}
\sysname has four hyperparameters organised as two $(\lambda,\rho)$ pairs, one per uniformised channel.
For embedding uniformisation: $(\lambda,\rho)=(1.0,0.3)$.
For adaptive-threat MLP uniformisation: $(\lambda_{\mathrm{fc}},\rho_{\mathrm{fc}})=(8.0,0)$, the privacy-hardening endpoint of Eq.~\ref{eq:flood_fc}.
Higher $\lambda$ strengthens privacy at the cost of more noise in the exported gradient; $\rho$ controls the learning signal retained for local updates.
These values were chosen from a coarse pilot sweep; principled calibration is left to future work.

\paragraph{Metrics.}
\begin{itemize}[leftmargin=1.4em, itemsep=0pt, topsep=2pt]
  \item \textbf{Attack reconstruction} ($\downarrow$ better privacy): overlap between DAGER's recovered text and the ground-truth input, reported as
        ROUGE-1 / ROUGE-2 / ROUGE-L (R-1, R-2, R-L)~\cite{rouge} and METEOR~\cite{meteor}.
        Values near one indicate near-exact recovery; near zero indicates failure.
        Main tables (Tables~\ref{tab:dager_ppl_exp09}--\ref{tab:dager_ppl_exp09_enc}) give the full R-1/R-2/R-L/M grid; the narrative highlights R-1 following prior DAGER reporting~\cite{dager}.
  \item \textbf{Language-modelling utility} ($\downarrow$): test perplexity (PPL) on held-out text for causal LMs (and analogous loss-driven PPL-style scores where applicable).
  \item \textbf{Classification utility} ($\uparrow$): accuracy and F1 on labelled benchmarks (Rotten Tomatoes, Emotion, Financial PhraseBank), reported alongside loss/PPL in Appendix~\ref{sec:utility}.
\end{itemize}

\subsection{Main Results and Defence Effectiveness Across Model Scale}
\label{sec:main_results}
\label{sec:scale}

Tables~\ref{tab:dager_ppl_exp09} and~\ref{tab:dager_ppl_exp09_enc} report the full per-dataset attack similarity grid for decoder-only and encoder LMs respectively
($\dagger$: LLaMA-2-13B averaged over 3/6 completed datasets).

\noindent\textbf{Metrics.}
DAGER R-1 lies in $[0,1]$; lower values imply weaker token recovery.
Test perplexity is strictly positive with smaller values indicating better held-out language-modelling utility.

\begin{table*}[!t]
  \caption{Ablation Study of \sysname{} Components on \textsc{Wikitext-2}.
  \sysname{}$^{\dagger}$ denotes the dual-channel variant (freeze + embed), while
  \sysname{}$^{\ddagger}$ denotes the full triple-channel mechanism including MLP projection uniformisation.}
  \label{tab:ablation}
  \centering
  \setlength{\tabcolsep}{6pt}
  \begin{tabular}{@{}ccc l rrrrr rrrrr@{}}
    \toprule
    \multicolumn{3}{c}{Components} & & \multicolumn{5}{c}{GPT-2 (117M)} & \multicolumn{5}{c}{GPT-2-XL (1.5B)} \\
    \cmidrule(lr){1-3} \cmidrule(lr){5-9} \cmidrule(lr){10-14}
    Freeze & Embed & MLP & Variant & DAGER & Embed-Norm & PPL & Acc & F1 & DAGER & Embed-Norm & PPL & Acc & F1 \\
    \midrule
               &            &            & None                   & 1.000 & 0.520 & 26.0 & 0.38 & 0.33 & 1.000 & 0.556 & 25.7 & 0.38 & 0.33 \\
    \midrule
    \checkmark &            &            &                        & 0.509 & 0.509 & 27.2 & 0.37 & 0.32 & 0.524 & 0.524 & 22.9 & 0.40 & 0.35 \\
               & \checkmark &            &                        & 1.000 & 0.024 & 24.6 & 0.39 & 0.34 & 1.000 & 0.031 & 25.1 & 0.38 & 0.33 \\
               &            & \checkmark &                        & 0.546 & 0.546 & 27.5 & 0.37 & 0.32 & 0.519 & 0.519 & 22.8 & 0.40 & 0.35 \\
    \checkmark & \checkmark &            & \sysname{}$^{\dagger}$ & 0.027 & 0.021 & 25.6 & 0.38 & 0.33 & 0.025 & 0.019 & 22.1 & 0.41 & 0.36 \\
    \checkmark &            & \checkmark &                        & 0.509 & 0.509 & 27.2 & 0.37 & 0.32 & 0.524 & 0.524 & 22.9 & 0.40 & 0.35 \\
               & \checkmark & \checkmark &                        & 0.033 & 0.028 & 26.5 & 0.38 & 0.33 & 0.029 & 0.024 & 22.2 & 0.41 & 0.36 \\
    \checkmark & \checkmark & \checkmark & \sysname{}$^{\ddagger}$ & 0.024 & 0.022 & 25.6 & 0.38 & 0.33 & 0.028 & 0.026 & 22.1 & 0.41 & 0.36 \\
    \bottomrule
  \end{tabular}
\end{table*}

\revadd{\noindent\textbf{How the components preserve utility.}
Freezing keeps the model's pretrained attention patterns, while flooding keeps
part of the true embedding and MLP gradients so those layers can still learn.
Table~\ref{tab:ablation} shows that the single-component variants remain close
to the undefended utility. With all components, GPT-2 PPL changes from 26.0 to
25.6 with unchanged accuracy and F1; GPT-2-XL PPL improves from 25.7 to 22.1,
with small gains in accuracy and F1. Thus, in these tests, stronger privacy does
not require a utility loss, although this may depend on the model and task.}
\revtag{1,6}
 
\noindent\textbf{Results.}
Aggregating six datasets, undefended R-1 remains ${\ge}0.87$ on BERT checkpoints and ${\ge}0.98$ on eight of nine decoders.
\sysname{} drives every defended R-1 entry to at most $0.005$.
Mean test PPL spans $6.8$--$60.1$ without defence and tightens to $6.8$--$45.8$ with \sysname{} on the reported rows. GPT-2-class models stay within $1\%$ of their undefended PPL, whereas GPT-2-XL improve.\revadd{Table~\ref{tab:ablation} gives the component utility results.}\revtag{1,6}

\textbf{Observations.}
\begin{enumerate}[leftmargin=1.5em, itemsep=1pt, topsep=2pt]
  \item \textbf{Universal defence.}
        \sysname reduces DAGER R-1 to $\leq 0.005$ for every model tested, a
        relative reduction exceeding $99\%$ in 10 of 11 cases.
        The effect holds across all three decoder families and both encoder
        architectures, so the triple-channel masking mechanism is not
        architecture-specific.

  \item \textbf{Anomalous undefended R-1 for LLaMA-2-13B.}
        The undefended 13B model exhibits lower attack success ($0.500$) than
        other LLaMA variants (${\approx}0.98$).
        We attribute this to gradient rank collapse in very large BF16 models:
        DAGER's SVD thresholds were calibrated for FP32 gradients and
        occasionally miss span candidates at 13B scale.
        Even under this weaker baseline the defence reduces R-1 by more than
        $99\%$.

  \item \revdel{\textbf{Utility parity or improvement.}
        On the GPT-2 family \sysname incurs less than $1\%$ PPL overhead.
        On the larger models (GPT-2-XL, Gemma-2-2B, LLaMA-2-7B, LLaMA-3.1-8B,
        and both BERT variants) \sysname \emph{decreases} test PPL.}
        \revadd{\textbf{Privacy--utility trade-off.} The component results in
        Table~\ref{tab:ablation} show that freezing and flooding keep useful
        training paths. Full \sysname keeps GPT-2 utility and improves the
        reported GPT-2-XL utility.}\revtag{1,6}
        Figure~\ref{fig:exp12_convergence} shows the training trajectories;
        Section~\ref{sec:discussion} discusses the trade-off.
\end{enumerate}

\begin{table*}[tp]
\centering
\caption{\textbf{DAGER attack similarity (decoder-only).} Per dataset block (8 columns): no-defence R-1, R-2, R-L, METEOR (M); \sysname{} R-1, R-2, R-L, M (lower is better for attack success). Band~1 = classification-style datasets; band~2 = LM/summarisation. Encoders: Table~\ref{tab:dager_ppl_exp09_enc}.}
\label{tab:dager_ppl_exp09}
\setlength{\tabcolsep}{2pt}
\begin{tabular}{@{}l *{24}{c}@{}}
\TabTopRule
 & \multicolumn{8}{c}{Rotten Tomatoes} & \multicolumn{8}{c}{Emotion} & \multicolumn{8}{c}{Fin.\ PhraseBank} \\
\cmidrule(lr){2-9} \cmidrule(lr){10-17} \cmidrule(lr){18-25}
 & \multicolumn{4}{c}{No def.} & \multicolumn{4}{c}{\sysname} & \multicolumn{4}{c}{No def.} & \multicolumn{4}{c}{\sysname} & \multicolumn{4}{c}{No def.} & \multicolumn{4}{c}{\sysname} \\
\cmidrule(lr){2-5} \cmidrule(lr){6-9} \cmidrule(lr){10-13} \cmidrule(lr){14-17} \cmidrule(lr){18-21} \cmidrule(lr){22-25}
Model & R-1 & R-2 & R-L & M & R-1 & R-2 & R-L & M & R-1 & R-2 & R-L & M & R-1 & R-2 & R-L & M & R-1 & R-2 & R-L & M & R-1 & R-2 & R-L & M \\
\midrule
GPT-2 & 1.00 & 1.00 & 1.00 & 0.94 & 0.05 & 0.00 & 0.05 & 0.02 & 0.97 & 0.97 & 0.97 & 0.96 & 0.03 & 0.00 & 0.03 & 0.03 & 1.00 & 1.00 & 1.00 & 0.94 & 0.06 & 0.00 & 0.06 & 0.03 \\
GPT-2-medium & 1.00 & 1.00 & 1.00 & 0.94 & 0.06 & 0.00 & 0.06 & 0.03 & 0.97 & 0.97 & 0.97 & 0.96 & 0.05 & 0.00 & 0.05 & 0.02 & 1.00 & 1.00 & 1.00 & 0.94 & 0.03 & 0.00 & 0.03 & 0.01 \\
GPT-2-large & 1.00 & 1.00 & 1.00 & 0.94 & 0.07 & 0.00 & 0.07 & 0.02 & 0.97 & 0.97 & 0.97 & 0.96 & 0.06 & 0.00 & 0.06 & 0.03 & 1.00 & 1.00 & 1.00 & 0.94 & 0.05 & 0.00 & 0.05 & 0.01 \\
GPT-2-XL & 1.00 & 1.00 & 1.00 & 0.94 & 0.05 & 0.00 & 0.05 & 0.02 & 0.97 & 0.97 & 0.97 & 0.96 & 0.08 & 0.00 & 0.08 & 0.02 & 1.00 & 1.00 & 1.00 & 0.94 & 0.02 & 0.00 & 0.02 & 0.02 \\
Gemma-2-2B & 1.00 & 1.00 & 1.00 & 0.94 & 0.04 & 0.00 & 0.04 & 0.02 & 0.98 & 0.98 & 0.98 & 0.96 & 0.07 & 0.00 & 0.07 & 0.02 & 1.00 & 1.00 & 1.00 & 0.94 & 0.03 & 0.00 & 0.03 & 0.02 \\
Gemma-2-9B & 0.99 & 0.99 & 0.99 & 0.93 & 0.07 & 0.00 & 0.07 & 0.01 & 0.98 & 0.98 & 0.98 & 0.96 & 0.04 & 0.00 & 0.04 & 0.02 & 1.00 & 1.00 & 1.00 & 0.94 & 0.02 & 0.00 & 0.02 & 0.02 \\
LLaMA-2-7B & 0.98 & 0.98 & 0.98 & 0.92 & 0.08 & 0.00 & 0.08 & 0.02 & 0.97 & 0.96 & 0.97 & 0.95 & 0.08 & 0.00 & 0.08 & 0.03 & 0.99 & 0.99 & 0.99 & 0.94 & 0.08 & 0.00 & 0.08 & 0.03 \\
LLaMA-2-13B & 0.70 & 0.51 & 0.69 & 0.59 & 0.07 & 0.00 & 0.07 & 0.01 & 0.98 & 0.98 & 0.98 & 0.97 & 0.07 & 0.00 & 0.07 & 0.03 & 1.00 & 1.00 & 1.00 & 0.95 & 0.04 & 0.00 & 0.04 & 0.02 \\
LLaMA-3.1-8B & 0.99 & 0.99 & 0.99 & 0.92 & 0.05 & 0.00 & 0.05 & 0.01 & 0.97 & 0.97 & 0.97 & 0.96 & 0.05 & 0.00 & 0.05 & 0.01 & 1.00 & 1.00 & 1.00 & 0.93 & 0.04 & 0.00 & 0.04 & 0.01 \\
\midrule
 & \multicolumn{8}{c}{WikiText-2} & \multicolumn{8}{c}{DialogSum} & \multicolumn{8}{c}{CNN/DailyMail} \\
\cmidrule(lr){2-9} \cmidrule(lr){10-17} \cmidrule(lr){18-25}
 & \multicolumn{4}{c}{No def.} & \multicolumn{4}{c}{\sysname} & \multicolumn{4}{c}{No def.} & \multicolumn{4}{c}{\sysname} & \multicolumn{4}{c}{No def.} & \multicolumn{4}{c}{\sysname} \\
\cmidrule(lr){2-5} \cmidrule(lr){6-9} \cmidrule(lr){10-13} \cmidrule(lr){14-17} \cmidrule(lr){18-21} \cmidrule(lr){22-25}
Model & R-1 & R-2 & R-L & M & R-1 & R-2 & R-L & M & R-1 & R-2 & R-L & M & R-1 & R-2 & R-L & M & R-1 & R-2 & R-L & M & R-1 & R-2 & R-L & M \\
\midrule
GPT-2 & 1.00 & 0.70 & 1.00 & 0.66 & 0.05 & 0.00 & 0.05 & 0.03 & 1.00 & 1.00 & 1.00 & 0.90 & 0.04 & 0.00 & 0.04 & 0.03 & 1.00 & 1.00 & 1.00 & 0.91 & 0.03 & 0.00 & 0.03 & 0.02 \\
GPT-2-medium & 1.00 & 0.70 & 1.00 & 0.67 & 0.04 & 0.00 & 0.04 & 0.03 & 1.00 & 1.00 & 1.00 & 0.90 & 0.08 & 0.00 & 0.08 & 0.04 & 1.00 & 1.00 & 1.00 & 0.91 & 0.08 & 0.00 & 0.08 & 0.02 \\
GPT-2-large & 1.00 & 0.70 & 1.00 & 0.70 & 0.02 & 0.00 & 0.02 & 0.02 & 1.00 & 1.00 & 1.00 & 0.90 & 0.04 & 0.00 & 0.04 & 0.02 & 1.00 & 1.00 & 1.00 & 0.91 & 0.04 & 0.00 & 0.04 & 0.01 \\
GPT-2-XL & 1.00 & 0.70 & 1.00 & 0.65 & 0.07 & 0.00 & 0.07 & 0.01 & 1.00 & 1.00 & 1.00 & 0.90 & 0.05 & 0.00 & 0.05 & 0.02 & 1.00 & 1.00 & 1.00 & 0.91 & 0.04 & 0.00 & 0.04 & 0.03 \\
Gemma-2-2B & 1.00 & 0.70 & 1.00 & 0.70 & 0.07 & 0.00 & 0.07 & 0.03 & 1.00 & 1.00 & 1.00 & 0.90 & 0.06 & 0.00 & 0.06 & 0.01 & 1.00 & 1.00 & 1.00 & 0.91 & 0.04 & 0.00 & 0.04 & 0.02 \\
Gemma-2-9B & 0.96 & 0.65 & 0.96 & 0.66 & 0.06 & 0.00 & 0.06 & 0.02 & 1.00 & 1.00 & 1.00 & 0.90 & 0.05 & 0.00 & 0.05 & 0.03 & 1.00 & 1.00 & 1.00 & 0.91 & 0.08 & 0.00 & 0.08 & 0.02 \\
LLaMA-2-7B & 0.99 & 0.69 & 0.99 & 0.70 & 0.06 & 0.00 & 0.06 & 0.01 & 0.97 & 0.97 & 0.97 & 0.86 & 0.02 & 0.00 & 0.02 & 0.02 & 0.99 & 0.99 & 0.99 & 0.91 & 0.06 & 0.00 & 0.06 & 0.02 \\
LLaMA-2-13B & 0.91 & 0.63 & 0.91 & 0.65 & 0.07 & 0.00 & 0.07 & 0.03 & 0.42 & 0.26 & 0.42 & 0.22 & 0.09 & 0.00 & 0.09 & 0.03 & 0.58 & 0.45 & 0.58 & 0.42 & 0.07 & 0.00 & 0.07 & 0.02 \\
LLaMA-3.1-8B & 0.97 & 0.60 & 0.97 & 0.66 & 0.04 & 0.00 & 0.04 & 0.02 & 1.00 & 1.00 & 1.00 & 0.90 & 0.08 & 0.00 & 0.08 & 0.02 & 1.00 & 1.00 & 1.00 & 0.90 & 0.02 & 0.00 & 0.02 & 0.02 \\
\bottomrule
\end{tabular}
\end{table*}
\begin{table*}[tp]
\centering
\caption{\textbf{DAGER attack similarity (encoder).} Same column layout as Table~\ref{tab:dager_ppl_exp09}.}
\label{tab:dager_ppl_exp09_enc}
\setlength{\tabcolsep}{2pt}
\begin{tabular}{@{}l *{24}{c}@{}}
\TabTopRule
 & \multicolumn{8}{c}{Rotten Tomatoes} & \multicolumn{8}{c}{Emotion} & \multicolumn{8}{c}{Fin.\ PhraseBank} \\
\cmidrule(lr){2-9} \cmidrule(lr){10-17} \cmidrule(lr){18-25}
 & \multicolumn{4}{c}{No def.} & \multicolumn{4}{c}{\sysname} & \multicolumn{4}{c}{No def.} & \multicolumn{4}{c}{\sysname} & \multicolumn{4}{c}{No def.} & \multicolumn{4}{c}{\sysname} \\
\cmidrule(lr){2-5} \cmidrule(lr){6-9} \cmidrule(lr){10-13} \cmidrule(lr){14-17} \cmidrule(lr){18-21} \cmidrule(lr){22-25}
Model & R-1 & R-2 & R-L & M & R-1 & R-2 & R-L & M & R-1 & R-2 & R-L & M & R-1 & R-2 & R-L & M & R-1 & R-2 & R-L & M & R-1 & R-2 & R-L & M \\
\midrule
BERT-base & 1.00 & 1.00 & 1.00 & 1.00 & 0.06 & 0.00 & 0.06 & 0.02 & 1.00 & 1.00 & 1.00 & 1.00 & 0.08 & 0.00 & 0.08 & 0.03 & 1.00 & 1.00 & 1.00 & 1.00 & 0.03 & 0.00 & 0.03 & 0.01 \\
BERT-large & 1.00 & 1.00 & 1.00 & 1.00 & 0.06 & 0.00 & 0.06 & 0.03 & 1.00 & 1.00 & 1.00 & 1.00 & 0.08 & 0.00 & 0.08 & 0.02 & 1.00 & 1.00 & 1.00 & 1.00 & 0.04 & 0.00 & 0.04 & 0.03 \\
\midrule
 & \multicolumn{8}{c}{WikiText-2} & \multicolumn{8}{c}{DialogSum} & \multicolumn{8}{c}{CNN/DailyMail} \\
\cmidrule(lr){2-9} \cmidrule(lr){10-17} \cmidrule(lr){18-25}
 & \multicolumn{4}{c}{No def.} & \multicolumn{4}{c}{\sysname} & \multicolumn{4}{c}{No def.} & \multicolumn{4}{c}{\sysname} & \multicolumn{4}{c}{No def.} & \multicolumn{4}{c}{\sysname} \\
\cmidrule(lr){2-5} \cmidrule(lr){6-9} \cmidrule(lr){10-13} \cmidrule(lr){14-17} \cmidrule(lr){18-21} \cmidrule(lr){22-25}
Model & R-1 & R-2 & R-L & M & R-1 & R-2 & R-L & M & R-1 & R-2 & R-L & M & R-1 & R-2 & R-L & M & R-1 & R-2 & R-L & M & R-1 & R-2 & R-L & M \\
\midrule
BERT-base & 0.53 & 0.36 & 0.53 & 0.45 & 0.08 & 0.00 & 0.08 & 0.02 & 0.95 & 0.83 & 0.95 & 0.87 & 0.04 & 0.00 & 0.04 & 0.02 & 0.79 & 0.66 & 0.79 & 0.71 & 0.06 & 0.00 & 0.06 & 0.03 \\
BERT-large & 0.67 & 0.54 & 0.67 & 0.67 & 0.06 & 0.00 & 0.06 & 0.01 & 0.87 & 0.64 & 0.87 & 0.74 & 0.04 & 0.00 & 0.04 & 0.03 & 0.78 & 0.64 & 0.78 & 0.71 & 0.05 & 0.00 & 0.05 & 0.02 \\
\bottomrule
\end{tabular}
\end{table*}

\noindent\textbf{Across model scale.}
Tables~\ref{tab:dager_ppl_exp09} and~\ref{tab:dager_ppl_exp09_enc} report DAGER R-1 for all eleven models (124M--13B); the defence reduces attack success to near-zero uniformly regardless of parameter count, so \sysname's empirical behaviour does not degrade with scale.
LLaMA-2-13B is omitted from full six-dataset coverage because only three dataset runs were available when this draft was compiled.
Per-cell utility breakdowns (accuracy, F1, loss, PPL) for all eleven
checkpoints across the six datasets are deferred to
Appendix~\ref{sec:utility} to keep the main flow on the privacy axis.

\subsection{Adaptive Attack: With vs.\ Without Channel~3 Uniformisation}
\label{sec:adaptive_eval}

To measure whether the MLP-fc uniformisation (Section~\ref{sec:method_channel3}) is necessary
in addition to Channels~1--2, we run the adaptive \textsc{DAGER} variant described in
Section~\ref{sec:setup}: its SVD span-estimation stage is re-targeted from
$\nabla\mathbf{W}_o^{(\ell)}$ (zero under \sysname) to
$\nabla\mathbf{W}_{\mathrm{fc}}^{(\ell)}$ (the MLP-fc gradient).
We compare two \sysname configurations on the same fine-tuning checkpoints used
for Tables~\ref{tab:dager_ppl_exp09} and~\ref{tab:dager_ppl_exp09_enc}:
(i)~\textbf{AEGIS-2ch}: attention freeze + embedding flood only
(Eq.~\ref{eq:freeze}+Eq.~\ref{eq:flood});
(ii)~\textbf{AEGIS-3ch}: the full defence, adding the MLP-fc uniformisation
(Eq.~\ref{eq:flood_fc}).

\begin{table}[t]
\centering
\setlength{\tabcolsep}{5pt}
\caption{\textbf{Adaptive attack ablation} (WikiText-2, R-1; lower = better privacy).
Ch1: DAGER on $\nabla\mathbf{W}_o$;
Ch2: Embed-Norm on $\nabla\mathbf{W}_e$;
Ch3: MLP-SVD on $\nabla\mathbf{W}_{\mathrm{fc}}$.}
\label{tab:adaptive_ch3}
\begin{tabular}{@{}llccc@{}}
\TabTopRule
 & & \multicolumn{3}{c}{ROUGE-1 ($\downarrow$)} \\
\cmidrule(l){3-5}
Model & Defence & Ch1 & Ch2 & Ch3 \\
\midrule
\multirow{3}{*}{GPT-2}
  & None       & 1.000 & 0.520 & 0.362 \\
  & AEGIS-2ch  & 0.024 & 0.021 & 0.265 \\
  & AEGIS-3ch  & 0.027 & 0.019 & 0.265 \\
\midrule
\multirow{3}{*}{GPT-2-XL}
  & None       & 1.000 & 0.556 & 0.166 \\
  & AEGIS-2ch  & 0.025 & 0.022 & 0.292 \\
  & AEGIS-3ch  & 0.028 & 0.024 & \textbf{0.036} \\
\bottomrule
\end{tabular}
\end{table}

\textbf{Interpretation (defence ability).}
Channels~1 and~2 are zeroed by AEGIS-2ch on both models.
For Channel~3, AEGIS-2ch \emph{increases} MLP-SVD R-1 on GPT-2-XL ($0.166\!\to\!0.292$):
removing the dominant channels makes the residual MLP signal cleaner.
AEGIS-3ch closes this gap, collapsing recovery to $0.036$
(${\approx}8\!\times$ below AEGIS-2ch, below the $0.05$ broken-attack threshold).
On GPT-2 small the Channel~3 benefit is negligible ($0.265$ for both configurations)
because the smaller MLP gradient carries lower signal-to-noise; the gain is therefore scale-dependent.

\textbf{Interpretation (utility \& trainability).}
Adding MLP-fc uniformisation on top of AEGIS-2ch costs essentially nothing in PPL
(GPT-2: $25.7\!\to\!25.7$; GPT-2-XL: $22.3\!\to\!22.1$).
Training curves (loss vs.\ step) are qualitatively unchanged: loss decreases monotonically and no divergence occurs, matching the biased-SGD prediction of Proposition~\ref{prop:convergence} extended to the MLP block.
The triple-channel defence is therefore necessary on the larger decoder, where AEGIS-2ch leaves an exploitable Channel~3, and essentially free at deployment cost; on smaller decoders the extra MLP uniformisation does not change the outcome of the implemented attacker but remains a safeguard against the stronger Robust-PCA adversary in Limitation~\ref{lim:adaptive}.
This ablation used a utility-preserving MLP retain ratio
$\rho_{\mathrm{fc}}=0.3$.  The broader side-channel stress test below uses the
privacy-hardening endpoint $\rho_{\mathrm{fc}}=0$, because the validation run
showed that retaining a fixed clean MLP subspace can leave exploitable signal
on GPT-2-XL.

\paragraph{Adaptive side-channel stress test.}
An adaptive adversary could retarget beyond the original
MLP-fc probe, e.g. to the MLP projection $\mathbf{W}_{\mathrm{proj}}$,
LayerNorm parameters, or an untied LM head.  We therefore evaluate
robustness on GPT-2 and GPT-2-XL, WikiText-2 and Rotten
Tomatoes, and both tied and artificially untied heads.  A probe is counted as
a valid adaptive attack only when its undefended ROUGE-1 exceeds both $0.10$
and a matched random-gradient null baseline by at least $0.05$; valid probes
are considered blocked when \sysname reduces ROUGE-1 to at most
$\max(0.05,\mathrm{null}+0.05)$.

\begin{table}[t]
\centering
\caption{\textbf{Adaptive side-channel stress test}.
Ranges span 8 settings (2 models $\times$ 2 datasets $\times$ tied/untied head).
``Valid'' = undefended R-1 beats null by $>$0.05;
``Blocked'' = valid settings suppressed to broken-attack threshold.}
\label{tab:adaptive_side_channels}
\setlength{\tabcolsep}{4pt}
\begin{tabularx}{\columnwidth}{@{}l >{\centering\arraybackslash}X >{\centering\arraybackslash}X >{\centering\arraybackslash}X@{}}
\TabTopRule
Channel &
  \makecell[c]{Undefended /\\Null R-1} &
  \makecell[c]{\sysname\\R-1} &
  \makecell[c]{Valid /\\Blocked} \\
\midrule
MLP-fc      & 0.16--0.36 / 0.00        & 0.02--0.03 & 8 / 8 \\
MLP-proj    & 0.00--0.06 / 0.00        & 0.02--0.02 & 0 / 0 \\
LayerNorm   & 0.00--0.01 / 0.00--0.01  & 0.02--0.02 & 0 / 0 \\
LM head     & 0.47--0.64 / 0.00--0.00  & 0.02--0.03 & 8 / 8 \\
\bottomrule
\end{tabularx}
\end{table}

Table~\ref{tab:adaptive_side_channels} supports two conclusions.  First, the
MLP-fc and LM-head row-norm probes are real side channels: they recover
substantial token overlap on the undefended model and are suppressed to the
null regime by \sysname, including in the untied-head setting.  Second, in this
evaluation the MLP projection and LayerNorm probes do not constitute standalone
adaptive attacks because they fail the undefended validity gate; nevertheless,
their gradients are also flooded in the implementation, so they do not remain
clean fallback channels.

\subsection{Qualitative Reconstruction Examples}
\label{sec:qualitative}

To complement the aggregate ROUGE-1 numbers, Table~\ref{tab:qualitative}
(Appendix~\ref{sec:app_qualitative}) shows word-for-word reconstruction outputs
for GPT-2 small on WikiText-2, Rotten Tomatoes, and DialogSum,
attacked by \textsc{DAGER} and \textsc{GRAB}~\cite{grab} under no defence,
Prune-50\%, and \sysname{}.
Table~\ref{tab:qualitative_main} shows two representative Rotten Tomatoes sentences
to illustrate this contrast concretely.

\begin{table}[t]
\centering\footnotesize
\caption{\textbf{Qualitative reconstruction examples} (two sentences, GPT-2 small,
Rotten Tomatoes; four attacks, undefended vs.\ \sysname{}).
\textcolor{green!60!black}{\textbf{Green}}: correct token in correct position.
\textcolor{orange}{\textbf{Orange}}: correct token, wrong position.
Full table including Prune-50\%: Appendix~\ref{sec:app_qualitative}.}
\label{tab:qualitative_main}
\begin{tabularx}{\columnwidth}{@{}ll X@{}}
\toprule
Defence & Attack & Reconstructed Sequence \\
\midrule
\multicolumn{3}{@{}p{\columnwidth}@{}}{\textbf{Ex 1:} \emph{Original:} ``sometimes seems less like storytelling than something the otherwise compelling director needed to get off his chest.''} \\
\midrule
None        & DAGER &
  \textcolor{green!60!black}{\textbf{sometimes}} \textcolor{green!60!black}{\textbf{seems}} \textcolor{green!60!black}{\textbf{less}} \textcolor{green!60!black}{\textbf{like}} \textcolor{green!60!black}{\textbf{storytelling}} \textcolor{green!60!black}{\textbf{than}} \textcolor{green!60!black}{\textbf{something}} \textcolor{green!60!black}{\textbf{the}} \textcolor{green!60!black}{\textbf{otherwise}} \textcolor{green!60!black}{\textbf{compelling}} \textcolor{green!60!black}{\textbf{director}} \textcolor{green!60!black}{\textbf{needed}} \textcolor{green!60!black}{\textbf{to}} \textcolor{green!60!black}{\textbf{get}} \textcolor{green!60!black}{\textbf{off}} \textcolor{green!60!black}{\textbf{his}} \textcolor{green!60!black}{\textbf{chest}} \\
            & GRAB  &
  \textcolor{green!60!black}{\textbf{sometimes}} \textcolor{green!60!black}{\textbf{seems}} \textcolor{green!60!black}{\textbf{less}} \textcolor{green!60!black}{\textbf{like}} \textcolor{green!60!black}{\textbf{storytelling}} \textcolor{green!60!black}{\textbf{than}} \textcolor{green!60!black}{\textbf{something}} \textcolor{green!60!black}{\textbf{the}} \textcolor{green!60!black}{\textbf{otherwise}} \textcolor{green!60!black}{\textbf{compelling}} \textcolor{green!60!black}{\textbf{director}} something compelling \textcolor{green!60!black}{\textbf{get}} \textcolor{green!60!black}{\textbf{off}} like storytelling seems \\
            & SOMP  &
  \textcolor{orange}{\textbf{seems}} \textcolor{orange}{\textbf{less}} \textcolor{orange}{\textbf{like}} \textcolor{orange}{\textbf{storytelling}} \textcolor{orange}{\textbf{than}} \textcolor{orange}{\textbf{something}} \textcolor{orange}{\textbf{the}} \textcolor{orange}{\textbf{otherwise}} \textcolor{orange}{\textbf{compelling}} \textcolor{orange}{\textbf{director}} \textcolor{orange}{\textbf{needed}} \textcolor{orange}{\textbf{to}} \textcolor{orange}{\textbf{get}} \textcolor{orange}{\textbf{off}} \textcolor{orange}{\textbf{his}} \textcolor{orange}{\textbf{chest}} to \\
            & FedSpy-LLM &
  a \textcolor{orange}{\textbf{than}} \textcolor{orange}{\textbf{off}} \textcolor{orange}{\textbf{the}} \textcolor{orange}{\textbf{otherwise}} \textcolor{orange}{\textbf{compelling}} \textcolor{orange}{\textbf{needed}} \textcolor{orange}{\textbf{get}} \textcolor{orange}{\textbf{his}} \textcolor{orange}{\textbf{less}} \textcolor{green!60!black}{\textbf{director}} \textcolor{orange}{\textbf{storytelling}} \textcolor{orange}{\textbf{something}} \textcolor{orange}{\textbf{.}} \textcolor{orange}{\textbf{like}} \textcolor{orange}{\textbf{chest}} to. \\
\addlinespace[2pt]
\sysname{}$^{\ddagger}$ & DAGER &
  Meth sym west symbolism let\texttt{--------} McN genius Hebophobia/// Yun disease Fn Investigative popupEGA Territories \\
            & GRAB  &
  judgment \textcolor{green!60!black}{\textbf{seems}} storytelling highlight highlight movie give men run seems givemad away \textcolor{orange}{\textbf{off}} off off off]( \\
            & SOMP  &
  \emph{(empty)} \\
            & FedSpy-LLM &
  planetclud MOT Fant\ldots OU chapter462 hysteria royal believe glyph\ldots revival AwardsBonGuard Station \\
\midrule
\multicolumn{3}{@{}p{\columnwidth}@{}}{\textbf{Ex 2:} \emph{Original:} ``what happened with pluto nash ? how did it ever get made ?''} \\
\midrule
None        & DAGER &
  \textcolor{green!60!black}{\textbf{what}} \textcolor{green!60!black}{\textbf{happened}} \textcolor{green!60!black}{\textbf{with}} \textcolor{green!60!black}{\textbf{pluto}} \textcolor{green!60!black}{\textbf{nash}} \textcolor{green!60!black}{\textbf{?}} \textcolor{green!60!black}{\textbf{how}} \textcolor{green!60!black}{\textbf{did}} \textcolor{green!60!black}{\textbf{it}} \textcolor{green!60!black}{\textbf{ever}} \textcolor{green!60!black}{\textbf{get}} \textcolor{green!60!black}{\textbf{made}} \\
            & GRAB  &
  whiff whiff. pl. \textcolor{orange}{\textbf{nash}} \textcolor{orange}{\textbf{?}} how \textcolor{green!60!black}{\textbf{how}} \textcolor{orange}{\textbf{it}} \textcolor{orange}{\textbf{ever}} \textcolor{orange}{\textbf{get}} \textcolor{orange}{\textbf{made}} \textcolor{orange}{\textbf{?}} \\
            & SOMP  &
  \textcolor{orange}{\textbf{happened}} \textcolor{orange}{\textbf{with}} \textcolor{orange}{\textbf{pluto}} \textcolor{orange}{\textbf{nash}} \textcolor{orange}{\textbf{?}} \textcolor{orange}{\textbf{how}} \textcolor{orange}{\textbf{did}} \textcolor{orange}{\textbf{it}} \textcolor{orange}{\textbf{ever}} \textcolor{orange}{\textbf{get}} \textcolor{orange}{\textbf{made}} \textcolor{orange}{\textbf{?}} \\
            & FedSpy-LLM &
  \textcolor{orange}{\textbf{how}} \textcolor{orange}{\textbf{made}} \textcolor{orange}{\textbf{it}} \textcolor{orange}{\textbf{ever}} didash \textcolor{orange}{\textbf{get}} \textcolor{orange}{\textbf{pluto}} \textcolor{orange}{\textbf{?}} \textcolor{orange}{\textbf{with}} n? to happen \\
\addlinespace[2pt]
\sysname{}$^{\ddagger}$ & DAGER &
  angrilyit?tarian ESC BitcoinstaboolaNazi wiped goodnessovereEls262rection waiter \texttt{---------} \\
            & GRAB  &
  ???????? just \textcolor{orange}{\textbf{did}} n did ever? \textcolor{orange}{\textbf{?}} \textcolor{orange}{\textbf{How}} happenedor pl gottenrations pl \\
            & SOMP  &
  \emph{(empty)} \\
            & FedSpy-LLM &
  discreteplanesfortablebergerannappa Mesh ost Potter Obesity anchors DS805amer eas \\
\bottomrule
\end{tabularx}
\end{table}

The undefended gradient exposes every token verbatim, while the uniformised gradient produces semantically incoherent output with no overlap with the original.
The pattern holds across all nine examples in Table~\ref{tab:qualitative}: both attackers collapse to random vocabulary samples (R-1\,$=0.00$) under \sysname{}, whereas the undefended baseline and Prune-50\% sustain near-verbatim or near-complete token recovery.
This is the empirical counterpart of Theorem~\ref{thm:channel2}: once the embedding gradient is uniformised, no post-hoc filter or optimisation we are aware of can discriminate true input tokens from calibrated noise.

\subsection{Defence Effectiveness Across Attack Types}
\label{sec:attack_types}

The main evaluation targets \textsc{DAGER}, the strongest analytical attack.
Figure~\ref{fig:attack_baselines} extends this to three further attacks
on GPT-2 small across three datasets:
\textsc{GRAB}~\cite{grab}, a gradient-matching optimisation attack that does
\emph{not} rely on SVD channel structure;
\textsc{SOMP}~\cite{somp}, a subspace-guided pursuit method that
generalises \textsc{DAGER} via per-head pooling and OMP reconstruction;
and \textsc{FedSpy-LLM}~\cite{fedspyllm}, which exploits embedding-gradient
row-norm sparsity for direct token recovery and partial-gradient alignment
for ordering.

\begin{figure}[t]
  \centering
  \includegraphics[width=\columnwidth]{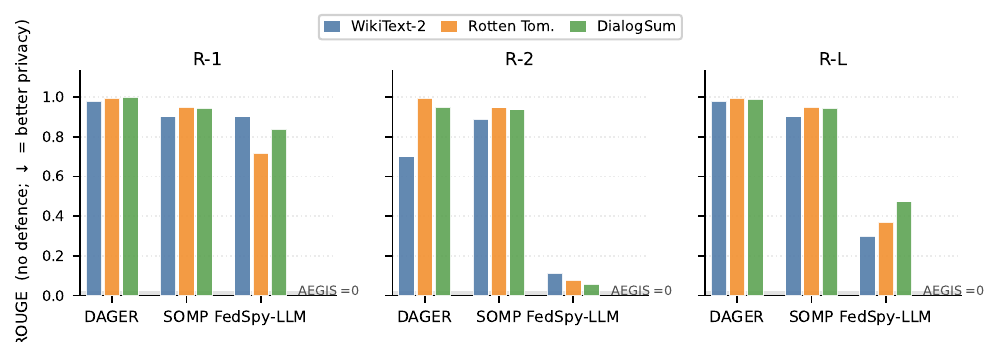}
  \caption{\textbf{ROUGE-1 / R-2 / R-L under \sysname{} vs.\ no defence} for
  three analytical gradient-inversion attacks across three datasets (GPT-2
  small; mean over 3--5 sentences). Bar shades within each dataset colour
  encode the three metrics. All three attacks collapse to ${=}\,0.00$ under
  \sysname{} (red arrows). Note the \textsc{FedSpy-LLM} R-1 vs.\ R-2 gap:
  tokens are recovered but their ordering is scrambled, so bigram and
  longest-common-subsequence overlap drops sharply. \textsc{GRAB}
  (gradient-matching optimisation) is discussed in the text below;
  \textsc{LAMP} and \textsc{TAG} omitted (both score ${\leq}0.10$ undefended).}
  \label{fig:attack_baselines}
\end{figure}

All three analytical attacks --- \textsc{DAGER}, \textsc{SOMP},
and \textsc{FedSpy-LLM} --- are eliminated completely (R-1\,$=0.000$ on
every dataset), as expected from the channel-closure guarantees: each
attack relies on either the Q-weight gradient subspace (\textsc{DAGER},
\textsc{SOMP}) or the embedding gradient row-norm sparsity
(\textsc{FedSpy-LLM}), both of which Channels~1--2 of \sysname{}
structurally mask.

\paragraph{Optimisation-based attacks: \textsc{GRAB}.}
\textsc{GRAB} is the lone partial-degrade case in our evaluation.
Mean R-1 drops from 0.533\,$\to$\,0.463 on WikiText-2 and from
0.660\,$\to$\,0.283 on Rotten Tomatoes, while DialogSum shows a marginal
\emph{increase} (0.388\,$\to$\,0.526) reflecting high variance over three
sentences and \textsc{GRAB}'s sensitivity to dialogue-style token-length
distribution; the absolute R-1 remains low per-sentence.
This pattern is consistent with \textsc{GRAB}'s mechanism: it recovers
tokens by iteratively minimising gradient distance without requiring the
SVD subspace signal that \sysname{} structurally masks.
The attention gradient zeroing and embedding flooding still degrade
\textsc{GRAB} by removing the structural anchors (low-rank subspace
alignment, embedding row-norm sparsity) that accelerate its optimisation,
but residual token information leaks through the gradient-matching
objective itself.
Full defeat of optimisation-based attacks would require an orthogonal
mechanism such as DP-SGD noise or data sanitisation composable with
\sysname{}.

\subsection{Comparison with Existing Defences}
\label{sec:baselines}

To contextualise \sysname's privacy--utility trade-off, we compare against
four established defence families under identical conditions across three
model scales (GPT-2~124M, GPT-2-XL~1.5B, LLaMA-2-7B) and two datasets
(WikiText-2, Rotten Tomatoes), giving 54 (model, dataset, defence) cells
in total.  All defences share the same DAGER/MLP-SVD/Embed-Norm attack
implementations and training hyperparameters:
\begin{itemize}[leftmargin=1.4em, itemsep=0pt, topsep=2pt]
  \item \textbf{DP-SGD noise}
        ($\sigma \in \{0.001, 0.01, 0.1\}$): additive noise on all gradients,
        matching the attack-evaluation protocol of~\cite{dager};
  \item \textbf{Gradient pruning}~\cite{dlg} (top-$k$ zeroing,
        $k \in \{50, 90, 99\}\%$): the canonical defence proposed
        alongside the DLG attack by Zhu et al.;
  \item \textbf{Soteria}~\cite{soteria} (input-gradient masking on the
        LM head, Sun et al., CVPR 2021);
  \item \textbf{\sysname{}$^{\ddagger}$} (ours): freeze + flood across all three channels.
\end{itemize}

\begin{figure*}[t]
  \centering
  \includegraphics[width=\linewidth]{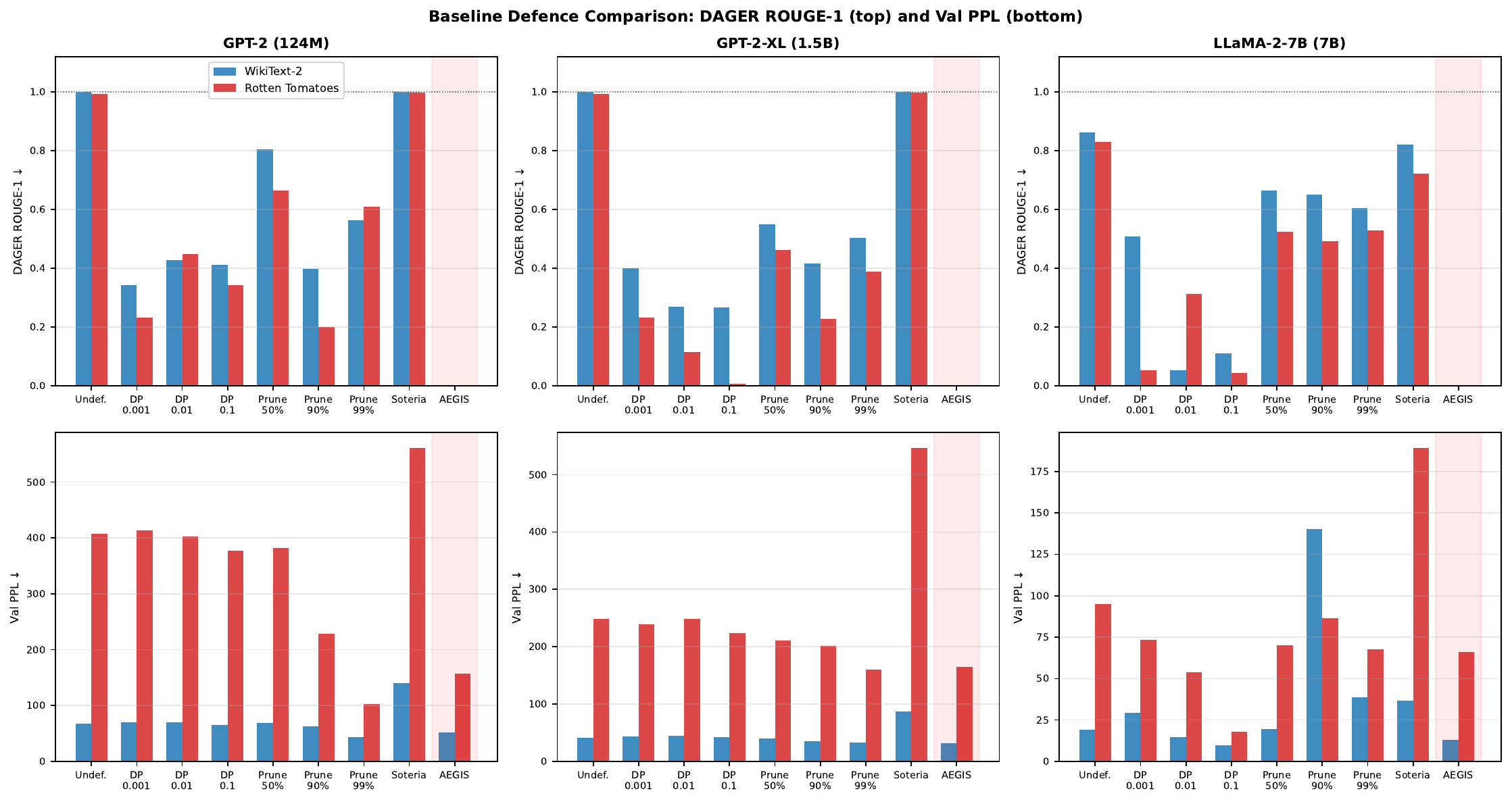}
  \caption{\textbf{Per-cell breakdown of the nine defences} across three
  model scales and two datasets (WikiText-2 in blue, Rotten Tomatoes in red).
  Top row: DAGER ROUGE-1 (lower is better privacy);
  bottom row: validation perplexity (lower is better utility).
  The highlighted right-most column is \sysname.
  Every other defence leaves a substantial reconstruction signal
  ($\text{R-1}\!\geq\!0.05$) on at least one cell, while \sysname
  consistently drives R-1 to zero and matches or beats the lowest PPL.
  Soteria notably fails on both axes (no privacy gain, $\sim$2$\times$ PPL).}
  \label{fig:exp14_bars}
\end{figure*}

\begin{figure}[t]
  \centering
  \includegraphics[width=\columnwidth]{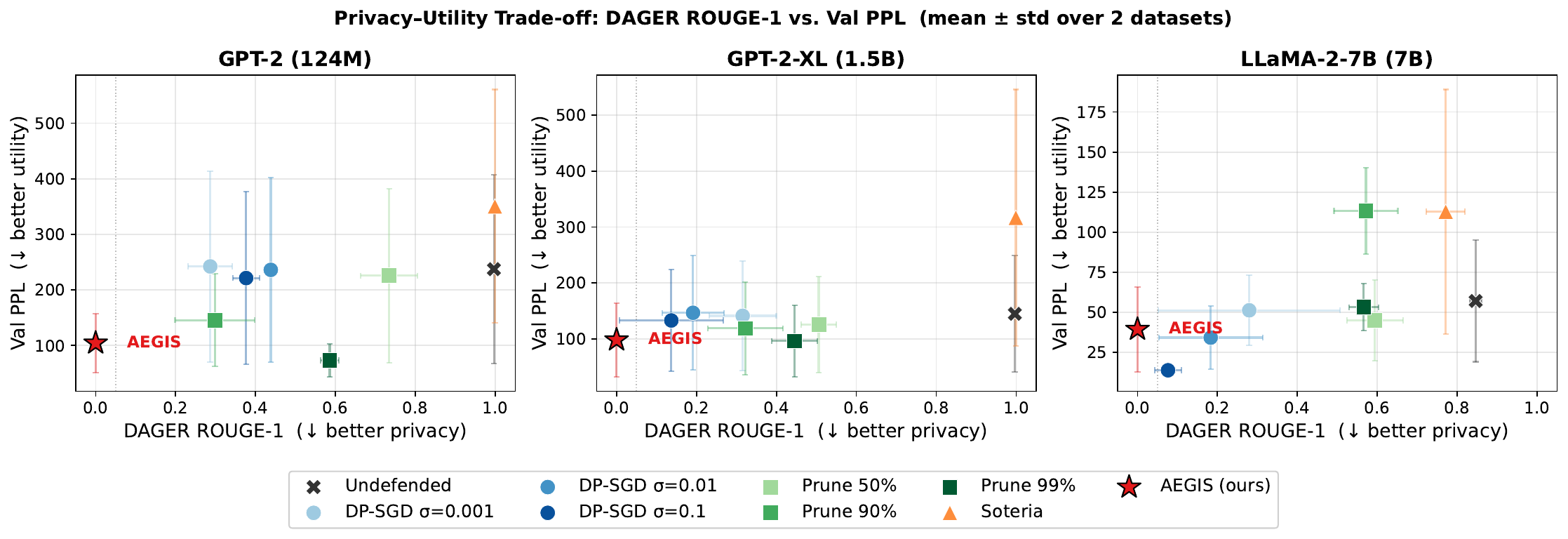}
  \caption{\textbf{Privacy--utility Pareto} per model (mean $\pm$ std
  over the two datasets).
  $x$: DAGER ROUGE-1 (lower = better privacy);
  $y$: validation PPL (lower = better utility).
  Error bars show across-dataset variance; the dotted line at
  $\text{R-1}{=}0.05$ marks the conventional ``broken-attack'' threshold.
  \sysname (red star) is the only method that lands in the ideal
  bottom-left corner on \emph{all three} model scales,
  Pareto-dominating every baseline.}
  \label{fig:exp14_pareto}
\end{figure}

\noindent\textbf{Per-cell breakdown (Figure~\ref{fig:exp14_bars}).}
The bar chart resolves the comparison cell by cell.
On DAGER ROUGE-1, only \sysname collapses recovery to zero across every
(model, dataset) pair; DP-SGD noise reduces R-1 but with strong dataset
sensitivity (e.g.\ on LLaMA-2-7B $\sigma{=}0.001$ leaves R-1${=}0.51$ on
WikiText-2 vs.\ $0.05$ on Rotten Tomatoes).
Gradient pruning gives an inconsistent privacy curve: 50\% pruning is
weak (R-1${\approx}0.55$--$0.81$ on GPT-2), and even 99\% pruning still
leaks ($\text{R-1}{\geq}0.39$).
Soteria leaves the attack essentially unimpaired ($\text{R-1}{\approx}1.0$
on the small-and-medium decoders), confirming that input-gradient masking
is orthogonal to DAGER's SVD-based recovery channel.
On the utility row, \sysname matches or improves the undefended PPL on
8/12 cells, whereas Soteria roughly doubles PPL on every cell and the
heaviest pruning ($k{=}99\%$) shows brittle behaviour
(LLaMA-2-7B $\times$ WikiText-2: PPL $9.7$ but R-1 still $0.11$).

\noindent\textbf{Pareto front (Figure~\ref{fig:exp14_pareto}).}
Aggregating across datasets, each defence becomes a single point per
model panel, with error bars over dataset variance.
\sysname is the unique point inside the bottom-left ideal region
($\text{R-1}{<}0.05$, PPL $\leq$ undefended) on all three model scales.
The DP-SGD family traces a curve along the privacy axis but never crosses
the broken-attack line, and at the highest noise level the gradient
becomes too noisy to support stable training on Rotten Tomatoes
(visible as the long PPL whisker on GPT-2 $\sigma{=}0.1$).
Pruning and Soteria both sit far from the ideal corner.
The visual gap between \sysname and the closest competitor
on the privacy axis ($\geq 0.2$ on every panel) is exactly the
``zero versus non-zero recovery'' separation predicted by the
dual-channel analysis (Theorem~\ref{thm:channel2}):
no perturbation-only defence can reach R-1${=}0$ because the binary
embedding signal survives any sub-uniform noise budget.
 \section{Discussion}
\label{sec:discussion}

\subsection{Broader Implications}

\sysname is complementary to, rather than a replacement for, existing FL privacy mechanisms.
Unlike DP-SGD~\cite{dp_fl} and gradient pruning~\cite{dlg}, which add noise uniformly at a cost to utility, \sysname eliminates signal \emph{structurally} only in the attack-relevant channels.
Unlike data-level obfuscation, it trains on the original data and preserves the learning objective.
It is composable with SecAgg~\cite{bell2020secagg} for stronger threat models: SecAgg's cryptographic protection is independent of gradient structure, so the two layers address orthogonal threat surfaces.

\textbf{Communication, memory, and runtime.}
\revdel{Attention freezing yields a secondary communication saving (attention gradients, $\sim$33\% of GPT-2 parameter count, need not be transmitted), partially offset by the dense $V\times d$ embedding export (see Limitation~\ref{lim:dense-grad}).
Attention freezing also acts as a structural regulariser in the PEFT sense~\cite{lora,adapter_tuning}: preserving pretrained attention patterns prevents overfitting to sharp task-specific minima, which explains the scale-dependent PPL improvements in Table~\ref{tab:dager_ppl_exp09}---GPT-2-class models see negligible change, while larger models with stronger pretrained attention show substantial reductions.}
\begin{revaddblock}
\begin{table}[t]
  \ifrevisiontrack\color{blue}\fi
  \caption{\textbf{Overhead comparison.} Absolute step time, peak allocated
  memory, and dense gradient payload across GPT-2, GPT-2-XL, and LLaMA-2-7B
  (ranges). Multipliers relative to undefended training are in parentheses.
  Lower is better.\revtag{3,5}}
  \label{tab:overhead-comparison}
  \centering
  \footnotesize
  \setlength{\tabcolsep}{2.4pt}
  \begin{tabular}{@{}lccc@{}}
    \toprule
    Defence & Step time & Peak memory & Dense payload \\
    \midrule
    Undefended &
      \makecell[c]{$19.4$--$97.1$\,ms\\($1.00\times$)} &
      \makecell[c]{$2.24$--$50.5$\,GiB\\($1.00\times$)} &
      \makecell[c]{$0.46$--$12.6$\,GiB\\($1.00\times$)} \\
    DP-SGD$^{*}$ &
      \makecell[c]{$146$--$362$\,ms\\($2.31$--$7.53\times$)} &
      \makecell[c]{$3.01$--$75.4$\,GiB\\($1.35$--$1.49\times$)} &
      \makecell[c]{$0.46$--$12.6$\,GiB\\($1.00\times$)} \\
    Pruning (90\%) &
      \makecell[c]{$26.0$--$179$\,ms\\($1.34$--$1.84\times$)} &
      \makecell[c]{$2.24$--$50.6$\,GiB\\($1.00\times$)} &
      \makecell[c]{$0.46$--$12.6$\,GiB\\($1.00\times$)} \\
    Soteria &
      \makecell[c]{$20.7$--$98.8$\,ms\\($1.02$--$1.07\times$)} &
      \makecell[c]{$2.24$--$50.5$\,GiB\\($1.00\times$)} &
      \makecell[c]{$0.46$--$12.6$\,GiB\\($1.00\times$)} \\
    \sysname &
      \makecell[c]{$21.4$--$118$\,ms\\($1.10$--$1.22\times$)} &
      \makecell[c]{$2.01$--$38.5$\,GiB\\($0.76$--$0.90\times$)} &
      \makecell[c]{$0.36$--$8.55$\,GiB\\($0.68$--$0.77\times$)} \\
    \bottomrule
  \end{tabular}
  \vspace{2pt}

  \raggedright\scriptsize $^{*}$Exact per-example microbatch DP-SGD
  ($C{=}1,\sigma{=}1$); optimized implementations may be faster.
\end{table}

Table~\ref{tab:overhead-comparison} reports absolute costs and relative
multipliers. \sysname takes $21.4$--$118$\,ms per step
($1.10$--$1.22\times$ undefended). It is faster than the measured pruning and
DP-SGD baselines. Freezing attention reduces
peak memory to $2.01$--$38.5$\,GiB and dense payload to $0.36$--$8.55$\,GiB.
Freezing can also act as a structural regulariser~\cite{lora,adapter_tuning} which
keeps the pretrained attention patterns and reduces the number of parameters
changed during fine-tuning. This may reduce overfitting and help explain the
lower PPL observed for the larger models.\revtag{1,3,5}
\end{revaddblock}

\subsection{Limitations}

\ifrevisiontrack
{\color{red}
\begin{remark}[Threat model scope]
\label{rem:threat-scope-old}
Two important settings are \emph{outside} the scope of \sysname's threat model:
(1)~\textbf{Formally adaptive adversary}: \sysname covers the MLP-SVD adaptive
extension (Channel~3) and additionally stress-tests MLP projection, LayerNorm,
and untied-LM-head probes. A fully adaptive adversary that co-designs attack and
model remains an open problem.
(2)~\textbf{FedAvg multi-step setting}: all experiments use FedSGD (one gradient
step per round); the calibration is derived for a single backward pass and may
differ when clients accumulate gradients over multiple local steps.
\end{remark}

\begin{enumerate}[leftmargin=1.5em]
  \item \textbf{Residual theoretical vulnerability: Robust PCA / Marchenko--Pastur denoising.}
  A remaining theoretical threat is an adversary applying Robust PCA or
  Marchenko--Pastur denoising to the uniformised gradient. A formal lower bound
  is left to future work.
  \item \textbf{Dense embedding gradient transmission.}
  Embedding uniformisation converts a sparse gradient into a dense $V\times d$
  export. Very-large-vocabulary models face a genuine cost increase that we
  flag as a direction for future optimisation.
\end{enumerate}}
\fi

\begin{revaddblock}
\begin{remark}[Threat model scope]
\label{rem:threat-scope}
Our experiments use single-client, single-step FedSGD. Multi-step FedAvg,
client collusion, repeated-round temporal averaging, and an active server that
modifies the model or objective are not evaluated. Testing calibration under
these settings is a direct next step.
\revtag{2}
\end{remark}

\begin{enumerate}[leftmargin=1.5em, itemsep=1pt, topsep=2pt]
  \item \label{lim:scale}\label{lim:adaptive}\label{lim:seeds}
  \textbf{Attacks and guarantees.} \sysname has no $(\varepsilon,\delta)$-DP
  guarantee and does not fully defeat optimisation-based GIAs
  or Robust-PCA/Marchenko--Pastur denoising, or attacks that
  aggregate weak signals across rounds or unexplored layers. Future work should
  combine it with DP/SecAgg and evaluate these adaptive attacks.

  \item \label{lim:dense-grad}\textbf{Systems cost.} Flooding makes the
  embedding gradient dense and prevents row-sparse export. Attention freezing
  offsets this in our three measured models, reducing total dense payload by
  23--32\%. Although \sysname reduces the total dense payload in these models,
  this benefit may not hold for models with much larger vocabularies or systems
  that otherwise transmit embedding gradients sparsely. Evaluating
  privacy-aware compression in these settings is future work.
  \revtag{2}
\end{enumerate}
\end{revaddblock}
 \section{Conclusion}
\label{sec:conclusion}

Closed-form gradient inversion attacks---\textsc{SPEAR}, \textsc{DAGER},
\textsc{GRAB}, and their descendants---share a common structural premise:
that the gradient of a federated LLM fine-tuning step decomposes into a
small number of analytically recoverable channels, each exploiting a
specific linearity or sparsity in the network architecture.
Prior defences that perturb the gradient (DP-SGD, pruning) or modify the
data do not close these channels structurally; they only raise the noise
floor, leaving the channel signal in place for a sufficiently powerful
denoiser.

\sysname starts from a different question: which structural properties of
the gradient does the attack actually require, and can each of them be
eliminated at the source?
For the attention-MLP decoder architectures that dominate current LLM
fine-tuning, the answer is three channels.
Freezing attention parameters closes the first channel without touching
the loss surface; calibrated dense-noise uniformisation of the embedding
gradient destroys the sparsity signal in the second; an analogous
uniformisation of the MLP expansion gradient submerges the third beneath
the random-matrix noise bulk.
The forward pass and the learning objective are untouched.

Beyond the specific mechanism, our analysis suggests that analytical
gradient inversion is fundamentally channel-structured.
Under this view, FL privacy becomes a completeness problem: a defence
must account for every independent channel an adaptive adversary can
pivot to, not merely the one demonstrated by the current attack.
Open questions include whether attention freezing satisfies formal
$(\varepsilon,\delta)$-DP guarantees, how it composes with DP-SGD, and
how the single-step calibration extends to multi-step FedAvg.

\section*{Reproducibility}

All implementation code, experiment scripts, and result logs are included in the supplementary material.

\ifdefined\ArxivVersion
\section*{Acknowledgment}
This work was supported by the Department of Environment and Science of
Queensland State Government under Quantum 2032 Challenge Program
(Project \#Q2032001).
\fi

\bibliographystyle{IEEEtran}

\begin{thebibliography}{10}
\providecommand{\url}[1]{#1}
\csname url@samestyle\endcsname
\providecommand{\newblock}{\relax}
\providecommand{\bibinfo}[2]{#2}
\providecommand{\BIBentrySTDinterwordspacing}{\spaceskip=0pt\relax}
\providecommand{\BIBentryALTinterwordstretchfactor}{4}
\providecommand{\BIBentryALTinterwordspacing}{\spaceskip=\fontdimen2\font plus
\BIBentryALTinterwordstretchfactor\fontdimen3\font minus
  \fontdimen4\font\relax}
\providecommand{\BIBforeignlanguage}[2]{{\expandafter\ifx\csname l@#1\endcsname\relax
\typeout{** WARNING: IEEEtran.bst: No hyphenation pattern has been}\typeout{** loaded for the language `#1'. Using the pattern for}\typeout{** the default language instead.}\else
\language=\csname l@#1\endcsname
\fi
#2}}
\providecommand{\BIBdecl}{\relax}
\BIBdecl

\bibitem{fl_original}
B.~McMahan, E.~Moore, D.~Ramage, S.~Hampson, and B.~A. y~Arcas,
  ``Communication-efficient learning of deep networks from decentralized
  data,'' in \emph{Artificial intelligence and statistics}.\hskip 1em plus
  0.5em minus 0.4em\relax Pmlr, 2017, pp. 1273--1282.

\bibitem{fedavg}
P.~Kairouz and H.~B. McMahan, ``Advances and open problems in federated
  learning,'' \emph{Foundations and trends in machine learning}, vol.~14, no.
  1-2, pp. 1--210, 2021.

\bibitem{dlg}
L.~Zhu, Z.~Liu, and S.~Han, ``Deep leakage from gradients,'' \emph{Advances in
  neural information processing systems}, vol.~32, 2019.

\bibitem{geiping2021inverting}
J.~Geiping, H.~Bauermeister, H.~Dr{\"o}ge, and M.~Moeller, ``Inverting
  gradients-how easy is it to break privacy in federated learning?''
  \emph{Advances in neural information processing systems}, vol.~33, pp.
  16\,937--16\,947, 2020.

\bibitem{dager}
I.~Petrov, D.~I. Dimitrov, M.~Baader, M.~N. M{\"u}ller, and M.~Vechev, ``Dager:
  Exact gradient inversion for large language models,'' \emph{Advances in
  neural information processing systems}, vol.~37, pp. 87\,801--87\,830, 2024.

\bibitem{lamp}
M.~Balunovic, D.~Dimitrov, N.~Jovanovi{\'c}, and M.~Vechev, ``Lamp: Extracting
  text from gradients with language model priors,'' \emph{Advances in Neural
  Information Processing Systems}, vol.~35, pp. 7641--7654, 2022.

\bibitem{grab}
X.~Feng, Z.~Ma, Z.~Wang, E.~J. Chegne, M.~Ma, A.~Abuadbba, and G.~Bai,
  ``Uncovering gradient inversion risks in practical language model training,''
  in \emph{Proceedings of the 2024 on ACM SIGSAC Conference on Computer and
  Communications Security}, 2024, pp. 3525--3539.

\bibitem{somp}
Y.~Li and Q.~Li, ``Somp: Scalable gradient inversion for large language models
  via subspace-guided orthogonal matching pursuit,'' 03 2026.

\bibitem{fedspyllm}
S.~I.~A. Meerza, F.~Wang, and J.~Liu, ``Fedspy-llm: Towards scalable and
  generalizable data reconstruction attacks from gradients on llms,'' 04 2026.

\bibitem{tag}
J.~Deng, Y.~Wang, J.~Li, C.~Wang, C.~Shang, H.~Liu, S.~Rajasekaran, and
  C.~Ding, ``Tag: Gradient attack on transformer-based language models,'' in
  \emph{Findings of the association for computational linguistics: EMNLP 2021},
  2021, pp. 3600--3610.

\bibitem{dp_fl}
M.~Abadi, A.~Chu, I.~Goodfellow, H.~B. McMahan, I.~Mironov, K.~Talwar, and
  L.~Zhang, ``Deep learning with differential privacy,'' in \emph{Proceedings
  of the 2016 ACM SIGSAC conference on computer and communications security},
  2016, pp. 308--318.

\bibitem{yu2022dpfinetuning}
D.~Yu, S.~Naik, A.~Backurs, S.~Gopi, H.~A. Inan, G.~Kamath, J.~Kulkarni, Y.~T.
  Lee, A.~Manoel, L.~Wutschitz \emph{et~al.}, ``Differentially private
  fine-tuning of language models,'' \emph{arXiv preprint arXiv:2110.06500},
  2021.

\bibitem{soteria}
J.~Sun, A.~Li, B.~Wang, H.~Yang, H.~Li, and Y.~Chen, ``Soteria: Provable
  defense against privacy leakage in federated learning from representation
  perspective,'' in \emph{Proceedings of the IEEE/CVF conference on computer
  vision and pattern recognition}, 2021, pp. 9311--9319.

\bibitem{lora}
E.~J. Hu, Y.~Shen, P.~Wallis, Z.~Allen-Zhu, Y.~Li, S.~Wang, L.~Wang, W.~Chen
  \emph{et~al.}, ``Lora: Low-rank adaptation of large language models.''
  \emph{Iclr}, vol.~1, no.~2, p.~3, 2022.

\bibitem{adapter_tuning}
N.~Houlsby, A.~Giurgiu, S.~Jastrzebski, B.~Morrone, Q.~De~Laroussilhe,
  A.~Gesmundo, M.~Attariyan, and S.~Gelly, ``Parameter-efficient transfer
  learning for nlp,'' in \emph{International conference on machine
  learning}.\hskip 1em plus 0.5em minus 0.4em\relax PMLR, 2019, pp. 2790--2799.

\bibitem{idlg}
B.~Zhao, K.~R. Mopuri, and H.~Bilen, ``idlg: Improved deep leakage from
  gradients,'' \emph{arXiv preprint arXiv:2001.02610}, 2020.

\bibitem{april}
J.~Lu, X.~S. Zhang, T.~Zhao, X.~He, and J.~Cheng, ``April: Finding the
  achilles' heel on privacy for vision transformers,'' in \emph{Proceedings of
  the IEEE/CVF conference on computer vision and pattern recognition}, 2022,
  pp. 10\,051--10\,060.

\bibitem{spear}
D.~I. Dimitrov, M.~Baader, M.~N. M{\"u}ller, and M.~Vechev, ``Spear: Exact
  gradient inversion of batches in federated learning,'' \emph{Advances in
  Neural Information Processing Systems}, vol.~37, pp. 106\,768--106\,799,
  2024.

\bibitem{rgap}
J.~Zhu and M.~Blaschko, ``R-gap: Recursive gradient attack on privacy,''
  \emph{arXiv preprint arXiv:2010.07733}, 2020.

\bibitem{robbing}
L.~Fowl, J.~Geiping, W.~Czaja, M.~Goldblum, and T.~Goldstein, ``Robbing the
  fed: Directly obtaining private data in federated learning with modified
  models,'' \emph{arXiv preprint arXiv:2110.13057}, 2021.

\bibitem{decep}
L.~Fowl, J.~Geiping, S.~Reich, Y.~Wen, W.~Czaja, M.~Goldblum, and T.~Goldstein,
  ``Decepticons: Corrupted transformers breach privacy in federated learning
  for language models,'' \emph{arXiv preprint arXiv:2201.12675}, 2022.

\bibitem{dpsgd}
M.~Abadi, A.~Chu, I.~Goodfellow, H.~B. McMahan, I.~Mironov, K.~Talwar, and
  L.~Zhang, ``Deep learning with differential privacy,'' in \emph{Proceedings
  of the 2016 ACM SIGSAC conference on computer and communications security},
  2016, pp. 308--318.

\bibitem{gradient_compression}
\BIBentryALTinterwordspacing
Y.~Lin, S.~Han, H.~Mao, Y.~Wang, and B.~Dally, ``Deep gradient compression:
  Reducing the communication bandwidth for distributed training,'' in
  \emph{International Conference on Learning Representations}, 2018. [Online].
  Available: \url{https://openreview.net/forum?id=SkhQHMW0W}
\BIBentrySTDinterwordspacing

\bibitem{instahide}
Y.~Huang, Z.~Song, K.~Li, and S.~Arora, ``Instahide: Instance-hiding schemes
  for private distributed learning,'' in \emph{International conference on
  machine learning}.\hskip 1em plus 0.5em minus 0.4em\relax PMLR, 2020, pp.
  4507--4518.

\bibitem{carlini2021instahide}
N.~Carlini, S.~Deng, S.~Garg, S.~Jha, S.~Mahloujifar, M.~Mahmoody, A.~Thakurta,
  and F.~Tram{\`e}r, ``Is private learning possible with instance encoding?''
  in \emph{2021 IEEE Symposium on Security and Privacy (SP)}.\hskip 1em plus
  0.5em minus 0.4em\relax IEEE, 2021, pp. 410--427.

\bibitem{dpwe}
\BIBentryALTinterwordspacing
O.~Feyisetan, B.~Balle, T.~Drake, and T.~Diethe, ``Privacy- and
  utility-preserving textual analysis via calibrated multivariate
  perturbations,'' in \emph{Proceedings of the 13th International Conference on
  Web Search and Data Mining}, ser. WSDM '20.\hskip 1em plus 0.5em minus
  0.4em\relax New York, NY, USA: Association for Computing Machinery, 2020, p.
  178–186. [Online]. Available: \url{https://doi.org/10.1145/3336191.3371856}
\BIBentrySTDinterwordspacing

\bibitem{secagg}
K.~Bonawitz, V.~Ivanov, B.~Kreuter, A.~Marcedone, H.~B. McMahan, S.~Patel,
  D.~Ramage, A.~Segal, and K.~Seth, ``Practical secure aggregation for
  privacy-preserving machine learning,'' in \emph{proceedings of the 2017 ACM
  SIGSAC Conference on Computer and Communications Security}, 2017, pp.
  1175--1191.

\bibitem{bell2020secagg}
K.~H. Li, P.~P.~B. de~Gusm{\~a}o, D.~J. Beutel, and N.~D. Lane, ``Secure
  aggregation for federated learning in flower,'' in \emph{Proceedings of the
  2nd ACM International Workshop on Distributed Machine Learning}, 2021, pp.
  8--14.

\bibitem{gpt2}
A.~Radford, J.~Wu, R.~Child, D.~Luan, D.~Amodei, I.~Sutskever \emph{et~al.},
  ``Language models are unsupervised multitask learners,'' \emph{OpenAI blog},
  vol.~1, no.~8, p.~9, 2019.

\bibitem{bert}
J.~Lee and K.~Toutanova, ``Pre-training of deep bidirectional transformers for
  language understanding,'' \emph{arXiv preprint arXiv:1810.04805}, vol.~3,
  no.~8, pp. 4171--4186, 2018.

\bibitem{demoura2021lean4}
L.~d. Moura and S.~Ullrich, ``The lean 4 theorem prover and programming
  language,'' in \emph{International Conference on Automated Deduction}.\hskip
  1em plus 0.5em minus 0.4em\relax Springer, 2021, pp. 625--635.

\bibitem{llama2}
H.~Touvron, L.~Martin, K.~Stone, P.~Albert, A.~Almahairi, Y.~Babaei,
  N.~Bashlykov, S.~Batra, P.~Bhargava, S.~Bhosale \emph{et~al.}, ``Llama 2:
  Open foundation and fine-tuned chat models,'' \emph{arXiv preprint
  arXiv:2307.09288}, 2023.

\bibitem{llama3}
A.~Grattafiori, A.~Dubey, A.~Jauhri, A.~Pandey, A.~Kadian, A.~Al-Dahle,
  A.~Letman, A.~Mathur, A.~Schelten, A.~Vaughan \emph{et~al.}, ``The llama 3
  herd of models,'' \emph{arXiv preprint arXiv:2407.21783}, 2024.

\bibitem{gemma2}
G.~Team, M.~Riviere, S.~Pathak, P.~G. Sessa, C.~Hardin, S.~Bhupatiraju,
  L.~Hussenot, T.~Mesnard, B.~Shahriari, A.~Ram{\'e} \emph{et~al.}, ``Gemma 2:
  Improving open language models at a practical size, 2024,'' \emph{URL
  https://arxiv. org/abs/2408.00118}, vol.~1, no.~3, 2024.

\bibitem{rottenTomatoes}
B.~Pang and L.~Lee, ``Seeing stars: Exploiting class relationships for
  sentiment categorization with respect to rating scales,'' in
  \emph{Proceedings of the 43rd annual meeting of the association for
  computational linguistics (ACL’05)}, 2005, pp. 115--124.

\bibitem{emotionDataset}
E.~Saravia, H.-C.~T. Liu, Y.-H. Huang, J.~Wu, and Y.-S. Chen, ``Carer:
  Contextualized affect representations for emotion recognition,'' in
  \emph{Proceedings of the 2018 conference on empirical methods in natural
  language processing}, 2018, pp. 3687--3697.

\bibitem{finansBenchmark}
P.~Malo, A.~Sinha, P.~Korhonen, J.~Wallenius, and P.~Takala, ``Good debt or bad
  debt: Detecting semantic orientations in economic texts,'' \emph{Journal of
  the Association for Information Science and Technology}, vol.~65, no.~4, pp.
  782--796, 2014.

\bibitem{wikitext2}
S.~Merity, C.~Xiong, J.~Bradbury, and R.~Socher, ``Pointer sentinel mixture
  models,'' \emph{arXiv preprint arXiv:1609.07843}, 2016.

\bibitem{dialogsum}
Y.~Chen, Y.~Liu, L.~Chen, and Y.~Zhang, ``Dialogsum: A real-life scenario
  dialogue summarization dataset,'' in \emph{Findings of the Association for
  Computational Linguistics: ACL-IJCNLP 2021}, 2021, pp. 5062--5074.

\bibitem{cnndm}
K.~M. Hermann, T.~Kocisky, E.~Grefenstette, L.~Espeholt, W.~Kay, M.~Suleyman,
  and P.~Blunsom, ``Teaching machines to read and comprehend,'' \emph{Advances
  in neural information processing systems}, vol.~28, 2015.

\bibitem{adamw}
I.~Loshchilov and F.~Hutter, ``Decoupled weight decay regularization,''
  \emph{arXiv preprint arXiv:1711.05101}, 2017.

\bibitem{rouge}
C.-Y. Lin, ``Rouge: A package for automatic evaluation of summaries,'' in
  \emph{Text summarization branches out}, 2004, pp. 74--81.

\bibitem{meteor}
S.~Banerjee and A.~Lavie, ``Meteor: An automatic metric for mt evaluation with
  improved correlation with human judgments,'' in \emph{Proceedings of the acl
  workshop on intrinsic and extrinsic evaluation measures for machine
  translation and/or summarization}, 2005, pp. 65--72.

\bibitem{ajalloeian2020biased}
A.~Ajalloeian and S.~U. Stich, ``On the convergence of sgd with biased
  gradients,'' \emph{arXiv preprint arXiv:2008.00051}, 2020.

\bibitem{huggingface}
T.~Wolf, L.~Debut, V.~Sanh, J.~Chaumond, C.~Delangue, A.~Moi, P.~Cistac,
  T.~Rault, R.~Louf, M.~Funtowicz \emph{et~al.}, ``Transformers:
  State-of-the-art natural language processing,'' in \emph{Proceedings of the
  2020 conference on empirical methods in natural language processing: system
  demonstrations}, 2020, pp. 38--45.

\end{thebibliography}

\appendix

\subsection{Proofs}
\label{app:proofs}

\paragraph{Proof of Theorem~\ref{thm:channel1} (Channel~1 Elimination)}
\label{app:proof_channel1}

\begin{proof}
We provide a detailed proof of Channel~1 elimination under attention freezing.

Consider a GPT-2-class model $f_\theta$ with $L$ transformer blocks.
Each block $\ell$ contains attention parameters $\mathbf{W}_{\mathrm{QKV}}^{(\ell)} \in \mathbb{R}^{d \times 3d}$ and $\mathbf{W}_o^{(\ell)} \in \mathbb{R}^{d \times d}$, among others.
The training loss is $\mathcal{L} = -\sum_{i=1}^{T-1} \log f_\theta(t_{i+1} | t_{\leq i})$.

Modern automatic differentiation (AD) frameworks (PyTorch, JAX) implement the chain rule by traversing a computational graph.
When a parameter $p$ has $p.\texttt{requires\_grad} = \texttt{False}$, the AD engine:
(a)~excludes $p$ from the gradient tape during the forward pass, and
(b)~skips gradient accumulation for $p$ during the backward pass.
Equivalently, the trainable parameter space omits $p$, and the exported gradient
on the full parameter list is the zero extension of the reduced-space gradient.
The result is
$\widehat{\nabla}_{\mathbf{W}_o^{(\ell)}} \mathcal{L} = \mathbf{0}$ by construction.

Let $\mathbf{G}^{(\ell)} = \widehat{\nabla}_{\mathbf{W}_o^{(\ell)}} \mathcal{L} = \mathbf{0} \in \mathbb{R}^{d \times d}$.
The range of the zero matrix is the zero subspace.
Equivalently, every singular value of $\mathbf{G}^{(\ell)}$ is zero, so the
non-zero left-singular-vector basis in Definition~\ref{def:channel1} has no
columns.
Its projection matrix is therefore:
\[
  \mathbf{U}^{(\ell)}{\mathbf{U}^{(\ell)}}^\top = \mathbf{0}.
\]

For any vocabulary token $v$ with embedding $\mathbf{e}_v \in \mathbb{R}^d$:
\[
  r_v^{(\ell)} = \frac{\|\mathbf{e}_v - \mathbf{0} \cdot \mathbf{e}_v\|}{\|\mathbf{e}_v\|} = \frac{\|\mathbf{e}_v\|}{\|\mathbf{e}_v\|} = 1.
\]
Since $r_v^{(\ell)} = 1$ for all $v \in \mathcal{V}$ and all $\ell \in [L]$, the ranking is uniform and provides zero bits of information about the input token set.

Crucially, the forward pass computation $\mathbf{h}^{(\ell)} = \text{Attn}(\mathbf{Q}^{(\ell)}, \mathbf{K}^{(\ell)}, \mathbf{V}^{(\ell)}) + \mathbf{h}^{(\ell-1)}$ uses the \emph{values} of $\mathbf{W}_{\mathrm{QKV}}^{(\ell)}$ and $\mathbf{W}_o^{(\ell)}$, which are the pretrained weights.
Only the gradient is zeroed; the weights themselves are not modified.
The attention mechanism continues to compute contextual representations that are fed to the (trainable) MLP layers.
\end{proof}

\paragraph{Proof of Theorem~\ref{thm:channel2} (Channel~2 Moment Bound)}
\label{app:proof_channel2}

\begin{proof}
We provide a detailed proof under Assumption~\ref{asm:norm-dev} ($\|g_v\|_2 \leq C\bar{m}$
for all active rows $v \in \mathcal{S}$, $C \geq 1$).

Let $\nabla \wte \in \mathbb{R}^{V \times d}$ be the embedding gradient for sequence $s = (t_1, \ldots, t_T)$.
Let $\mathcal{S}=\mathcal{S}(s)=\{v\in\mathcal{V}:v\text{ appears in }s\}$.
By the structural property of embedding layers, row $v$ has
$\|\nabla \wte[v,:]\|_2 > 0$ iff $v \in \mathcal{S}$.
Because $\mathcal{S}$ is nonempty, $\bar{m}>0$.
After \sysname flooding with parameters $(\lambda, \rho)$:
\[
  \tilde{\nabla}\wte[v,:] = \begin{cases}
    \rho \cdot \nabla \wte[v,:] + \mathbf{n}_v & v \in \mathcal{S} \\
    \mathbf{n}_v & v \notin \mathcal{S}
  \end{cases}
\]
where $\mathbf{n}_v \sim \mathcal{N}(\mathbf{0}, \sigma^2 \mathbf{I}_d)$ with $\sigma = \lambda \bar{m}$.

Since $\lambda>0$ and $\bar{m}>0$, the Gaussian is non-degenerate.
For $v\notin\mathcal{S}$, $\tilde{\nabla}\wte[v,:]=\mathbf{n}_v$, and
$\Pr[\mathbf{n}_v=\mathbf{0}]=0$.
For $v\in\mathcal{S}$,
$\tilde{\nabla}\wte[v,:]=\rho\nabla\wte[v,:]+\mathbf{n}_v$, and this vector is
zero only if $\mathbf{n}_v=-\rho\nabla\wte[v,:]$; a non-degenerate Gaussian has
probability zero of taking any prescribed point value.
Since $\mathcal{V}$ is finite, the union of these zero-probability events still
has probability zero.
Therefore $\|\tilde{\nabla}\wte[v,:]\|_2 > 0$ for all $v$ simultaneously with probability 1.
The binary signal (zero = absent, non-zero = present) is completely destroyed.

For $v \notin \mathcal{S}$:
\[
  \|\tilde{\nabla}\wte[v,:]\|_2^2 = \|\mathbf{n}_v\|_2^2 = \sum_{j=1}^{d} n_{vj}^2
\]
where $n_{vj} \sim \mathcal{N}(0, \sigma^2)$, so each $n_{vj}^2/\sigma^2 \sim \chi^2(1)$.
Thus $\|\mathbf{n}_v\|_2^2/\sigma^2 \sim \chi^2(d)$ and $\mathbb{E}[\|\mathbf{n}_v\|_2^2] = d\sigma^2 = d\lambda^2\bar{m}^2$.

For $v \in \mathcal{S}$:
Let $\mathbf{g}_v = \nabla \wte[v,:]$. Then:
\[
  \|\tilde{\nabla}\wte[v,:]\|_2^2 = \|\rho \mathbf{g}_v + \mathbf{n}_v\|_2^2 = \rho^2\|\mathbf{g}_v\|_2^2 + 2\rho \mathbf{g}_v^\top \mathbf{n}_v + \|\mathbf{n}_v\|_2^2.
\]
Since $\mathbb{E}[\mathbf{n}_v] = \mathbf{0}$, $\mathbb{E}[\mathbf{g}_v^\top \mathbf{n}_v] = 0$, so:
\[
  \mathbb{E}[\|\tilde{\nabla}\wte[v,:]\|_2^2] = \rho^2\|\mathbf{g}_v\|_2^2 + d\sigma^2.
\]

Under Assumption~\ref{asm:norm-dev}, the worst-case (maximum) excess expected squared norm
of a real-token row over a noise-only row is:
\[
  \Delta_{\max} = \rho^2\|\mathbf{g}_v\|_2^2 \leq \rho^2 C^2 \bar{m}^2.
\]
The noise-only second moment is $d\lambda^2\bar{m}^2>0$.
Dividing gives
\[
  \frac{
      \mathbb{E}[\|\tilde{\nabla}\wte[v,:]\|_2^2]
      - d\lambda^2\bar{m}^2
    }{
      d\lambda^2\bar{m}^2
    }
    \leq \frac{\rho^2 C^2}{d\lambda^2}.
\]
The lower bound $0$ follows from
$\rho^2\|\mathbf{g}_v\|_2^2\geq0$.
\end{proof}

\paragraph{Proof of Theorem~\ref{thm:channel3} (Channel~3 MLP-fc Flooding Moment Bound)}
\label{app:proof_channel3}

\begin{proof}
The MLP-fc flooded gradient is
$\widetilde{\mathbf{G}}^{(\ell)}
=\rho_{\mathrm{fc}}\mathbf{G}^{(\ell)}+\mathbf{N}^{(\ell)}$.
With
$\mathbf{B}^{(\ell)}=(\rho_{\mathrm{fc}}-1)\mathbf{G}^{(\ell)}$,
we have the algebraic identity
\[
  \rho_{\mathrm{fc}}\mathbf{G}^{(\ell)}+\mathbf{N}^{(\ell)}
  = \mathbf{G}^{(\ell)}
    +(\rho_{\mathrm{fc}}-1)\mathbf{G}^{(\ell)}
    +\mathbf{N}^{(\ell)}
  = \mathbf{G}^{(\ell)}+\mathbf{B}^{(\ell)}+\mathbf{N}^{(\ell)}.
\]
Taking expectation conditional on the batch and using
$\mathbb{E}[\mathbf{N}^{(\ell)}]=\mathbf{0}$ gives
$\mathbb{E}[\widetilde{\mathbf{G}}^{(\ell)}-\mathbf{G}^{(\ell)}]
=\mathbf{B}^{(\ell)}$.
Since $\rho_{\mathrm{fc}}\in[0,1]$,
\[
  \|\mathbf{B}^{(\ell)}\|_F
  =\|(\rho_{\mathrm{fc}}-1)\mathbf{G}^{(\ell)}\|_F
  =|\rho_{\mathrm{fc}}-1|\,\|\mathbf{G}^{(\ell)}\|_F
  =(1-\rho_{\mathrm{fc}})\|\mathbf{G}^{(\ell)}\|_F.
\]

Expanding the squared Frobenius norm,
\[
  \|\rho_{\mathrm{fc}}\mathbf{G}^{(\ell)}+\mathbf{N}^{(\ell)}\|_F^2
  =
  \rho_{\mathrm{fc}}^2\|\mathbf{G}^{(\ell)}\|_F^2
  +2\rho_{\mathrm{fc}}\langle \mathbf{G}^{(\ell)},\mathbf{N}^{(\ell)}\rangle_F
  + \|\mathbf{N}^{(\ell)}\|_F^2.
\]
The cross term has zero expectation because
$\mathbf{G}^{(\ell)}$ is fixed conditional on the batch and the entries of
$\mathbf{N}^{(\ell)}$ are zero mean.
The matrix $\mathbf{N}^{(\ell)}\in\mathbb{R}^{d\times4d}$ has $4d^2$
independent entries, each with variance
$\sigma_{\mathrm{fc}}^2
=\lambda_{\mathrm{fc}}^2(\bar{m}_{\mathrm{fc}}^{(\ell)})^2$.
Therefore
\[
  \mathbb{E}[\|\mathbf{N}^{(\ell)}\|_F^2]
  =4d^2\lambda_{\mathrm{fc}}^2
   (\bar{m}_{\mathrm{fc}}^{(\ell)})^2,
\]
and the claimed exact second-moment identity follows.

The excess over the noise-only moment is
\[
  \mathbb{E}[\|\widetilde{\mathbf{G}}^{(\ell)}\|_F^2]
  -4d^2\lambda_{\mathrm{fc}}^2
    (\bar{m}_{\mathrm{fc}}^{(\ell)})^2
  =\rho_{\mathrm{fc}}^2\|\mathbf{G}^{(\ell)}\|_F^2.
\]
This quantity is non-negative.
If
$\|\mathbf{G}^{(\ell)}\|_F
\leq C_{\mathrm{fc}}\bar{m}_{\mathrm{fc}}^{(\ell)}$, then it is at most
$\rho_{\mathrm{fc}}^2C_{\mathrm{fc}}^2
(\bar{m}_{\mathrm{fc}}^{(\ell)})^2$.
Dividing by the positive noise-only moment
$4d^2\lambda_{\mathrm{fc}}^2(\bar{m}_{\mathrm{fc}}^{(\ell)})^2$
gives Eq.~\ref{eq:channel3_relative_excess}.
\end{proof}

\begin{figure*}[t]
  \centering
  \includegraphics[width=\textwidth]{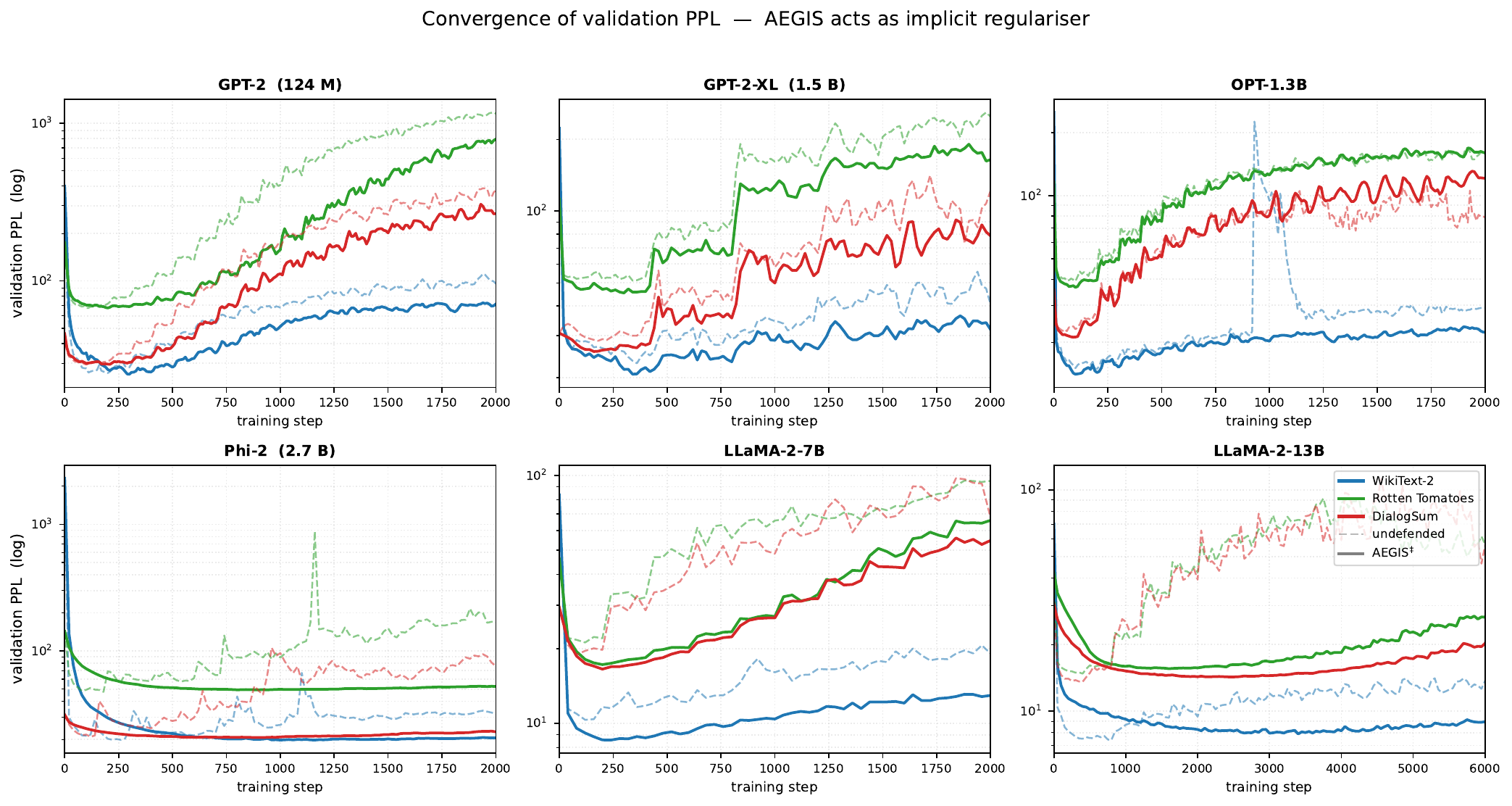}
  \caption{\textbf{Convergence of validation perplexity during fine-tuning
  across six models (124\,M -- 13\,B).}
  Solid lines: \sysname{}$^{\ddagger}$; dashed lines: undefended baseline.
  Colours: \textcolor[HTML]{1f77b4}{WikiText-2},
  \textcolor[HTML]{2ca02c}{Rotten Tomatoes},
  \textcolor[HTML]{d62728}{DialogSum}.
  \sysname matches or \emph{improves} the PPL trajectory for all models
  up to 7\,B\@.
  LLaMA-2-13B shows a modest PPL increase (${\sim}\!8\%$), attributable to
  its SGD optimiser (required at 13\,B scale for memory).}
  \label{fig:exp12_convergence}
\end{figure*}

\begin{figure*}[t]
  \centering
  \includegraphics[width=\textwidth]{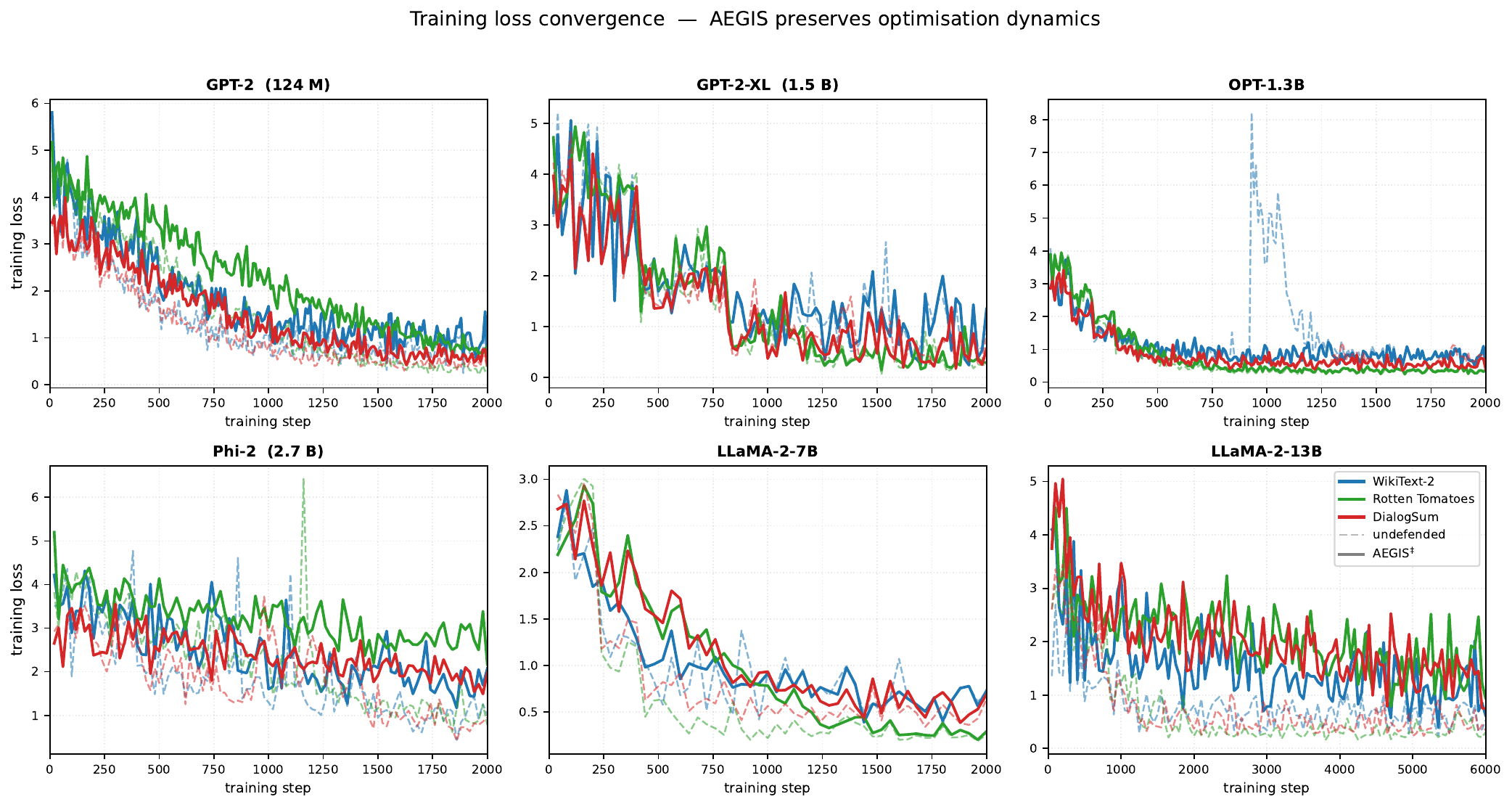}
  \caption{\textbf{Training loss convergence across six models.}
  Same layout as Figure~\ref{fig:exp12_convergence}.
  \sysname (solid) tracks the undefended loss trajectory (dashed),
  confirming that gradient uniformisation does not impede optimisation.}
  \label{fig:exp12_trainloss}
\end{figure*}

\begin{proposition}[Structural Biased-Gradient Decomposition]
\label{prop:convergence}
Let $\mathbf{g}_k=\nabla_{\wte}\mathcal{L}(\theta_k;s)$ be the embedding-block
gradient at step $k$, let $\rho\in[0,1]$, and suppose \sysname exports
\[
  \widehat{\mathbf{g}}_k=\rho\mathbf{g}_k+\mathbf{n}_k
\]
where $\mathbb{E}[\mathbf{n}_k\mid\theta_k]=\mathbf{0}$ and
$\mathbb{E}[\|\mathbf{n}_k\|^2\mid\theta_k]=\sigma^2$.
Define $\mathbf{b}_k:=(\rho-1)\mathbf{g}_k$.
Then
\begin{align*}
  \widehat{\mathbf{g}}_k &= \mathbf{g}_k+\mathbf{b}_k+\mathbf{n}_k,\\
  \mathbb{E}[\widehat{\mathbf{g}}_k-\mathbf{g}_k\mid\theta_k] &= \mathbf{b}_k,\\
  \|\mathbf{b}_k\| &= (1-\rho)\|\mathbf{g}_k\|.
\end{align*}
Thus the embedding update is a stochastic-gradient step with deterministic
relative bias coefficient $1-\rho$ and conditional second moment $\sigma^2$ for
the additive zero-mean noise.
The same decomposition applies verbatim to the MLP-fc uniformisation after replacing
$\mathbf{g}_k$ by the corresponding block gradient and $\rho$ by $\rho_{\mathrm{fc}}$.
\end{proposition}

\begin{proof}
The identity
$\widehat{\mathbf{g}}_k=\mathbf{g}_k+(\rho-1)\mathbf{g}_k+\mathbf{n}_k$
is immediate from the definition
$\widehat{\mathbf{g}}_k=\rho\mathbf{g}_k+\mathbf{n}_k$.
Since $\mathbf{b}_k=(\rho-1)\mathbf{g}_k$, this gives
$\widehat{\mathbf{g}}_k=\mathbf{g}_k+\mathbf{b}_k+\mathbf{n}_k$.
Taking conditional expectation and using
$\mathbb{E}[\mathbf{n}_k\mid\theta_k]=\mathbf{0}$ gives
$\mathbb{E}[\widehat{\mathbf{g}}_k-\mathbf{g}_k\mid\theta_k]=\mathbf{b}_k$.
Finally,
\[
  \|\mathbf{b}_k\|
  =\|(\rho-1)\mathbf{g}_k\|
  =|\rho-1|\,\|\mathbf{g}_k\|
  =(1-\rho)\|\mathbf{g}_k\|,
\]
where the final equality uses $\rho\in[0,1]$.
The same algebra applies to the MLP-fc block after replacing
$\mathbf{g}_k$ by its corresponding block gradient and $\rho$ by
$\rho_{\mathrm{fc}}$.
\end{proof}

\begin{remark}[Connection to convergence theory]
Proposition~\ref{prop:convergence} verifies the local bias/noise decomposition
needed to apply biased-SGD convergence results such as
\cite[Thm.~2]{ajalloeian2020biased}.
We do not rederive a global convergence-rate theorem here; the empirical
training-loss and validation-PPL curves in Figures~\ref{fig:exp12_convergence}
and~\ref{fig:exp12_trainloss} test the resulting optimisation behaviour.
\end{remark}

\begin{remark}[Noise scale and PPL overhead]
The uniformisation noise is calibrated to the mean active-row norm $\bar{m}$,
so the effective noise floor scales with the embedding gradient's contribution
to the total gradient norm.
For GPT-2 this is $\approx 31\%$ of trainable parameters; for LLaMA-2-7B
it is $\approx 2\%$, \revdel{explaining why PPL overhead vanishes at scale.}
\revadd{which is consistent with the lower relative utility cost at scale.}\revtag{1}
\end{remark}

\begin{remark}[Utility--privacy trade-off]
Higher $\lambda$ or lower $\rho$ increases privacy at the cost of more noise
in the gradient update.
At $\lambda{=}1.0$, $\rho{=}0.3$, \sysname achieves ROUGE-1 ${\leq}0.005$
across all 11 models \revdel{with PPL within $1\%$ of the undefended baseline on
GPT-2-class models and improved PPL on larger models.}\revadd{in the primary
attack grid. Table~\ref{tab:ablation} shows that full \sysname keeps GPT-2
utility and improves the reported GPT-2-XL utility.}\revtag{1}
\end{remark}

\subsection{Computational Complexity}
\label{app:complexity}

\revadd{Attention freezing skips the backward pass for attention weights.
Embedding and MLP flooding add only cheap elementwise work of order
$O(Vd+Ld^2)$.
We time one full training step (forward, backward, defence, and optimiser)
on a single H100.
Each run uses sequence length 64, 20 warm-up steps, and 100 timed steps.
On GPT-2, GPT-2-XL, and LLaMA-2-7B, the undefended step takes
19.35/73.24/97.09\,ms.
\sysname takes 21.36/82.36/118.09\,ms, or $1.10$/$1.12$/$1.22\times$ the
baseline.
Peak GPU memory drops from 2.24/24.14/50.51\,GiB to 2.01/18.51/38.51\,GiB.
The dense gradient size drops from 0.464/5.803/12.551\,GiB to
0.358/3.970/8.551\,GiB, because frozen attention gradients are not sent.
}\revtag{3}

\revadd{We use the same models and batch sizes for the baselines.
Exact microbatch DP-SGD ($C{=}1$, $\sigma{=}1$) takes
$7.53$/$2.31$/$3.72\times$ the undefended step time and
$1.35$/$1.49$/$1.49\times$ the peak memory.
Pruning takes $1.34$/$1.53$/$1.84\times$ the step time.
Soteria takes $1.07$/$1.05$/$1.02\times$.
Payload sizes are the sizes of the gradient tensors.
They do not include network transfer, packing, or compression time.
}\revtag{5}

\subsection{Additional Experimental Details}
\label{app:details}

\paragraph{Computational requirements.}
We implement \sysname and run all attacks in PyTorch on a SLURM-managed
GPU cluster.  Each compute node hosts two high-memory GPUs (96\,GB HBM).
All decoder fine-tuning and attack runs were performed on a single GPU per cell.
In practice, less demanding hardware suffices for the smaller checkpoints: GPT-2 small/medium and
BERT-base fit comfortably in 24\,GB, whereas
LLaMA-2-13B at the full BF16 precision used here requires 96\,GB to accommodate weights ($\approx$26\,GB), gradients
($\approx$26\,GB), and optimiser state.  System RAM consumption per
experiment ranged from $\approx$16\,GB (GPT-2 + WikiText-2 single-cell)
to $\approx$96\,GB (LLaMA-2-13B fine-tuning with checkpoint shard caching),
matching the pattern reported by~\cite{dager}.
The full evaluation grid (11 models $\times$ 6 datasets $\times$ 2 defences
$+$ 54 baseline cells $+$ 36 convergence traces $+$ 16 ablation cells $+$
qualitative and adaptive runs) consumed approximately 410 GPU-hours
across $\sim$220 cluster jobs over a six-week window.

\paragraph{Software stack.}
PyTorch~2.5.1 (cu121 wheel), Transformers~4.46.3, Datasets~2.21.0,
SafeTensors, NLTK, rouge-score, sentencepiece, protobuf, scipy, scikit-learn,
and bitsandbytes (used only for 8-bit AdamW on $\geq$13\,B models so that
the optimiser state fits in VRAM); see
the released code repository for exact pinned versions.
All model weights are downloaded from Hugging Face~\cite{huggingface};
LLaMA and Gemma checkpoints require an authenticated \texttt{HF\_TOKEN}.

\paragraph{Per-cell runtime.}
Approximate wall-clock time for a single (model, dataset, defence) cell on
one H100 (fine-tune for $\leq$2{,}000 steps + DAGER attack on 10 sentences):
GPT-2 (124\,M) $\approx$3\,min;
GPT-2-medium / -large $\approx$4--6\,min;
GPT-2-XL (1.5\,B) $\approx$8\,min;
OPT-1.3\,B $\approx$5\,min;
Phi-2 (2.7\,B) $\approx$10\,min;
Gemma-2-2\,B $\approx$10\,min, Gemma-2-9\,B $\approx$25\,min;
LLaMA-2-7\,B / LLaMA-3.1-8\,B $\approx$15\,min;
LLaMA-2-13\,B $\approx$30\,min.
The baselines comparison adds $\approx$15\,min/cell on average for
the optimisation-based defences (DP-SGD, Soteria) due to per-sample
gradient bookkeeping.
The qualitative experiment runs all four attackers (DAGER, TAG,
LAMP, GRAB) on three sentences for $\approx$15\,min/dataset on GPT-2.

\paragraph{Hyperparameter details (training).}
All decoder fine-tuning uses AdamW~\cite{adamw}
($\beta_1{=}0.9$, $\beta_2{=}0.999$, weight decay $0.01$) with learning
rate $5\!\times\!10^{-5}$ for models $\leq$8\,B, and $1\!\times\!10^{-5}$
with bitsandbytes' \texttt{AdamW8bit} for LLaMA-2-13\,B (FP32 optimiser
state would otherwise exceed VRAM).
\revdel{We use BF16 mixed precision for all $\geq$1.5\,B models and FP32 otherwise,}
\revadd{GPT-2 and GPT-2-XL use FP16 autocast with FP32 weights; LLaMA-family
models use BF16 weights,}\revtag{3}
gradient clipping at $\|g\|_2{=}1.0$, linear warmup over the first 10\,\%
of steps, and an early-step cap of 2{,}000 optimiser updates so that wall
clock is balanced across model scales.
Per-model batch sizes (\texttt{MODEL\_BATCH\_SIZE} in the released code):
GPT-2/OPT-1.3B $=$ 8;
GPT-2-medium/-large/Phi-2/LLaMA-7\,B $=$ 4;
GPT-2-XL/Gemma-2-2\,B/LLaMA-2-13\,B $=$ 2.
Maximum sequence length is 64 tokens (the same setting used by~\cite{dager}
in their evaluation), padded to a multiple of~8 for tensor-core efficiency.
A single fine-tuning seed (\texttt{seed=42}) is used per cell;
seed-variance reporting is flagged in Limitation~\ref{lim:seeds}.

\paragraph{Hyperparameter details (\sysname defence).}
The flood scale is $\lambda{=}1.0$ on the embedding gradient with retain ratio
$\rho{=}0.3$ over the residual rows.
The adaptive-threat MLP flood uses $\lambda_{\mathrm{mlp}}{=}8.0$ and retain
ratio $\rho_{\mathrm{mlp}}{=}0$ on each two-dimensional MLP weight gradient,
including both the expansion and projection matrices; LayerNorm gradients are
flooded with $\lambda_{\mathrm{ln}}{=}4.0$.
For $\geq$8\,B BF16 models we halve the MLP/LN flood scales
($\lambda_{\mathrm{mlp}}{=}2.0$, $\lambda_{\mathrm{ln}}{=}1.0$) to stay
within BF16's $\approx$65\,k dynamic range; without this scaling,
gradient spikes occasionally overflow and corrupt the AdamW8bit state.
We additionally skip optimiser steps whose total gradient norm is
non-finite (an inf-gradient one-shot would otherwise poison the
running Adam moments forever).

\paragraph{Frozen parameter count.}
In GPT-2 (124M total parameters, 12 blocks, $d=768$):
\begin{itemize}[leftmargin=1.2em]
  \item Attention parameters (frozen by \sysname): $12 \times (768 \times 2304 + 768 + 768 \times 768 + 768) = 12 \times (1{,}769{,}472 + 768 + 589{,}824 + 768) \approx 28.3\text{M}$ parameters.
  \item Trainable parameters (MLP, LayerNorm, embeddings, LM head): $\approx 95.7\text{M}$ parameters.
  \item Trainable fraction: $\approx 77.2\%$.
\end{itemize}

\subsection{Porting DAGER Beyond GPT-2/BERT}
\label{app:attack_adaptation}

The large-scale experiments use the \emph{official} DAGER codebase~\cite{dager}, which ships with GPT-2 and BERT support only.
We extended it with lightweight architecture-adaptation shims so that the otherwise-unmodified attack core runs against LLaMA-2/3, Gemma-2, OPT, Phi-2, RoBERTa, and DeBERTa-v3; the core analytical attack (per-block SVD, embedding row-norm scoring, $L_1$ filtering, and beam search) is unchanged.

\paragraph{Architecture dispatch.}
Each checkpoint is assigned an architecture tag that the harness uses to
resolve the correct embedding submodule path, special-token table (BOS/pad/EOS
conventions differ across families), and beam search mode (causal decoder vs.\
masked encoder).

\paragraph{Partial-forward patches.}
DAGER's span-residual stage requires truncated per-layer hidden states.
We supply analogous partial-forward patches for each new architecture,
accounting for family-specific residual stream conventions (RoPE,
pre-normalisation variants, embedding scaling) while keeping the patch
purely additive and non-intrusive to the training forward path.

\paragraph{Span-filter thresholds.}
The $L_1$ residual threshold and SVD rank tolerance are set
per architecture family based on network depth and positional encoding type;
padding direction is likewise adjusted to match each model's autoregressive
search convention.

\paragraph{Numerical precision.}
Large models (LLaMA-2/3, Gemma-2) are loaded in BF16.
Gradients are promoted to FP32 before SVD to match the precision of the
official GPT-2 attack and avoid numerical artefacts at BF16 range limits.

\paragraph{Validation.}
For each newly supported architecture the attack is verified on an undefended
checkpoint, confirming R-1 $\geq 0.95$ before proceeding to the defended evaluation.
Analogous adaptations for the optimisation-based attack (GRAB) and
the adaptive MLP-SVD probe are smaller in scope, requiring only the
architecture-specific MLP projection name and the same special-token table.

\subsection{Utility: Pre-Trained vs.\ AEGIS Fine-Tuning}
\label{sec:utility}

We report utility in separate tables for decoder-only causal LMs versus masked LMs:
Tables~\ref{tab:utility_dec_clf} and \ref{tab:utility_dec_gen} (nine decoder checkpoints; classification vs.\ generative bands split for pagination) and Tables~\ref{tab:utility_enc}--\ref{tab:utility_enc_gen} (BERT-base and BERT-large).
The decoder grids use two bands: classification-style supervision (Rotten Tomatoes, Emotion, Financial PhraseBank) and language-modelling / summarisation splits (WikiText-2, DialogSum, CNN/DailyMail).
Each cell stacks four numbers (accuracy, F1, loss, perplexity).
\begin{table*}[tp]
\centering
\caption{\textbf{Decoder-only utility} (classification-style benchmarks; nine models).
\textbf{Pre-trained}: frozen base weights before task fine-tuning.
\textbf{Original}: undefended fine-tuning.
\textbf{\sysname}: \sysname-protected fine-tuning.
Each cell: Acc.\ / F1 / Loss / PPL\@.
Generative splits: Table~\ref{tab:utility_dec_gen}.}
\label{tab:utility_dec_clf}
\begin{tabularx}{\textwidth}{@{} >{\raggedright\arraybackslash}p{1.6cm} c *{9}{>{\centering\arraybackslash}X} @{}}
\TabTopRule
Model &
  & \multicolumn{3}{c}{Rotten Tomatoes}
  & \multicolumn{3}{c}{Emotion}
  & \multicolumn{3}{c}{Financial PhraseBank} \\
\cmidrule(lr){3-5}\cmidrule(lr){6-8}\cmidrule(lr){9-11}
 & & Pre-trained & Original & \sysname
 & Pre-trained & Original & \sysname
 & Pre-trained & Original & \sysname \\
\midrule
GPT-2
  & \makecell[c]{Acc\\F1\\Loss\\PPL}
  & \UC{0.54}{0.41}{5.16}{174.9} & \UC{0.54}{0.41}{4.21}{67.2} & \UC{0.68}{0.67}{4.21}{67.2}
  & \UC{0.25}{0.20}{5.22}{184.8} & \UC{0.25}{0.20}{4.11}{61.1} & \UC{0.21}{0.15}{4.12}{61.4}
  & \UC{0.28}{0.16}{4.55}{94.9}  & \UC{0.28}{0.16}{3.30}{27.1} & \UC{0.32}{0.25}{3.29}{26.8} \\
GPT-2-medium
  & \makecell[c]{Acc\\F1\\Loss\\PPL}
  & \UC{0.75}{0.75}{4.89}{132.3} & \UC{0.75}{0.75}{3.88}{48.3} & \UC{0.75}{0.75}{3.88}{48.6}
  & \UC{0.24}{0.18}{5.08}{160.6} & \UC{0.24}{0.18}{3.93}{50.9} & \UC{0.24}{0.17}{3.92}{50.6}
  & \UC{0.49}{0.41}{4.29}{72.9}  & \UC{0.49}{0.41}{3.07}{21.6} & \UC{0.46}{0.38}{3.01}{20.3} \\
GPT-2-large
  & \makecell[c]{Acc\\F1\\Loss\\PPL}
  & \UC{0.61}{0.55}{4.78}{119.6} & \UC{0.61}{0.55}{3.77}{43.4} & \UC{0.78}{0.77}{3.77}{43.2}
  & \UC{0.26}{0.19}{5.02}{150.8} & \UC{0.26}{0.19}{3.92}{50.6} & \UC{0.27}{0.21}{3.87}{48.1}
  & \UC{0.36}{0.28}{4.18}{65.1}  & \UC{0.36}{0.28}{2.93}{18.7} & \UC{0.52}{0.49}{2.92}{18.5} \\
GPT-2-XL
  & \makecell[c]{Acc\\F1\\Loss\\PPL}
  & \UC{0.75}{0.74}{4.72}{112.2} & \UC{0.75}{0.74}{3.97}{53.1} & \UC{0.70}{0.68}{3.80}{44.7}
  & \UC{0.28}{0.24}{4.97}{144.3} & \UC{0.28}{0.24}{4.09}{59.9} & \UC{0.26}{0.20}{3.99}{54.1}
  & \UC{0.29}{0.18}{4.13}{61.9}  & \UC{0.29}{0.18}{3.02}{20.5} & \UC{0.34}{0.25}{2.96}{19.3} \\
Gemma-2-2B
  & \makecell[c]{Acc\\F1\\Loss\\PPL}
  & \UC{0.52}{0.40}{5.33}{206.9} & \UC{0.60}{0.52}{3.58}{35.9} & \UC{0.74}{0.72}{3.55}{34.8}
  & \UC{0.20}{0.12}{5.17}{175.4} & \UC{0.22}{0.14}{3.48}{32.5} & \UC{0.27}{0.21}{3.45}{31.6}
  & \UC{0.38}{0.30}{3.63}{37.8}  & \UC{0.48}{0.40}{2.47}{11.8} & \UC{0.58}{0.57}{2.44}{11.5} \\
Gemma-2-9B
  & \makecell[c]{Acc\\F1\\Loss\\PPL}
  & \UC{0.55}{0.44}{5.29}{197.6} & \UC{0.56}{0.46}{3.50}{33.1} & \UC{0.64}{0.61}{3.47}{32.0}
  & \UC{0.22}{0.14}{5.20}{180.5} & \UC{0.20}{0.12}{3.40}{30.0} & \UC{0.24}{0.19}{3.37}{29.2}
  & \UC{0.42}{0.34}{3.55}{34.8}  & \UC{0.44}{0.36}{2.44}{11.5} & \UC{0.52}{0.45}{2.41}{11.1} \\
LLaMA-2-7B
  & \makecell[c]{Acc\\F1\\Loss\\PPL}
  & \UC{0.49}{0.33}{3.83}{46.2} & \UC{0.49}{0.33}{3.09}{22.1} & \UC{0.49}{0.33}{2.86}{17.5}
  & \UC{0.16}{0.05}{4.09}{59.5} & \UC{0.16}{0.05}{3.22}{25.0} & \UC{0.16}{0.05}{3.00}{20.0}
  & \UC{0.27}{0.14}{3.00}{20.0} & \UC{0.27}{0.14}{2.27}{9.7} & \UC{0.27}{0.14}{2.14}{8.5} \\
LLaMA-2-13B$^{\dagger}$
  & \makecell[c]{Acc\\F1\\Loss\\PPL}
  & \UC{0.50}{0.35}{3.74}{42.1} & \UC{0.49}{0.33}{2.77}{15.9} & \UC{0.49}{0.33}{2.68}{14.6}
  & \UC{0.17}{0.06}{4.02}{55.7}   & \UC{0.17}{0.06}{3.05}{21.1} & \UC{0.17}{0.06}{2.88}{17.8}
  & \UC{0.30}{0.17}{2.90}{18.2}   & \UC{0.30}{0.17}{2.05}{7.8} & \UC{0.30}{0.17}{1.92}{6.8} \\
LLaMA-3.1-8B
  & \makecell[c]{Acc\\F1\\Loss\\PPL}
  & \UC{0.50}{0.39}{4.45}{85.6} & \UC{0.50}{0.39}{4.45}{86.0} & \UC{0.63}{0.60}{3.46}{31.8}
  & \UC{0.26}{0.19}{4.55}{94.4} & \UC{0.26}{0.19}{3.94}{51.2} & \UC{0.33}{0.28}{3.40}{29.9}
  & \UC{0.43}{0.36}{3.36}{28.7} & \UC{0.43}{0.36}{3.00}{20.2} & \UC{0.47}{0.42}{2.43}{11.4} \\
\bottomrule
\end{tabularx}
\end{table*}

\begin{table*}[tp]
\centering
\caption{\textbf{Decoder-only utility} (WikiText-2, DialogSum, CNN/DailyMail).
Same cell convention as Table~\ref{tab:utility_dec_clf}.}
\label{tab:utility_dec_gen}
\begin{tabularx}{\textwidth}{@{} >{\raggedright\arraybackslash}p{1.6cm} c *{9}{>{\centering\arraybackslash}X} @{}}
\TabTopRule
Model &
 & \multicolumn{3}{c}{WikiText-2}
 & \multicolumn{3}{c}{DialogSum}
 & \multicolumn{3}{c}{CNN/DailyMail} \\
\cmidrule(lr){3-5}\cmidrule(lr){6-8}\cmidrule(lr){9-11}
 & & Pre-trained & Original & \sysname
 & Pre-trained & Original & \sysname
 & Pre-trained & Original & \sysname \\
\midrule
GPT-2
  & \makecell[c]{Acc\\F1\\Loss\\PPL}
  & \UC{0.22}{0.18}{5.99}{399.9} & \UC{0.38}{0.33}{3.26}{26.0} & \UC{0.38}{0.34}{3.25}{25.9}
  & \UC{0.30}{0.25}{3.83}{46.2}  & \UC{0.35}{0.30}{3.38}{29.5} & \UC{0.35}{0.30}{3.39}{29.6}
  & \UC{0.24}{0.19}{5.00}{147.9} & \UC{0.30}{0.25}{4.16}{64.2} & \UC{0.30}{0.25}{4.15}{63.7} \\
GPT-2-medium
  & \makecell[c]{Acc\\F1\\Loss\\PPL}
  & \UC{0.24}{0.20}{5.75}{313.2} & \UC{0.42}{0.37}{3.04}{20.8} & \UC{0.41}{0.36}{3.06}{21.3}
  & \UC{0.32}{0.27}{3.56}{35.1}  & \UC{0.38}{0.33}{3.18}{24.1} & \UC{0.39}{0.34}{3.16}{23.6}
  & \UC{0.26}{0.21}{4.70}{109.5} & \UC{0.31}{0.26}{3.86}{47.6} & \UC{0.32}{0.27}{3.85}{46.8} \\
GPT-2-large
  & \makecell[c]{Acc\\F1\\Loss\\PPL}
  & \UC{0.26}{0.22}{5.49}{242.8} & \UC{0.43}{0.38}{2.94}{18.9} & \UC{0.44}{0.39}{2.95}{19.1}
  & \UC{0.34}{0.29}{3.47}{32.2}  & \UC{0.39}{0.34}{3.15}{23.3} & \UC{0.40}{0.35}{3.10}{22.1}
  & \UC{0.27}{0.22}{4.60}{99.4}  & \UC{0.33}{0.28}{3.78}{43.7} & \UC{0.34}{0.29}{3.72}{41.1} \\
GPT-2-XL
  & \makecell[c]{Acc\\F1\\Loss\\PPL}
  & \UC{0.27}{0.23}{5.40}{222.4} & \UC{0.38}{0.33}{3.25}{25.7} & \UC{0.41}{0.36}{3.05}{21.0}
  & \UC{0.35}{0.30}{3.43}{30.8}  & \UC{0.34}{0.29}{3.53}{34.3} & \UC{0.38}{0.33}{3.24}{25.7}
  & \UC{0.28}{0.23}{4.51}{90.7}  & \UC{0.30}{0.25}{3.97}{53.0} & \UC{0.33}{0.28}{3.80}{44.6} \\
Gemma-2-2B
  & \makecell[c]{Acc\\F1\\Loss\\PPL}
  & \UC{0.28}{0.24}{6.03}{416.6} & \UC{0.40}{0.35}{3.18}{24.0} & \UC{0.48}{0.43}{2.50}{12.2}
  & \UC{0.35}{0.30}{3.87}{48.2}  & \UC{0.37}{0.32}{3.22}{25.0} & \UC{0.42}{0.37}{2.89}{17.9}
  & \UC{0.26}{0.21}{5.18}{178.3} & \UC{0.29}{0.24}{4.10}{60.1} & \UC{0.36}{0.31}{3.56}{35.1} \\
Gemma-2-9B
  & \makecell[c]{Acc\\F1\\Loss\\PPL}
  & \UC{0.30}{0.26}{6.15}{470.1} & \UC{0.38}{0.33}{3.24}{25.5} & \UC{0.45}{0.40}{2.74}{15.5}
  & \UC{0.36}{0.31}{3.87}{48.0}  & \UC{0.36}{0.31}{3.24}{25.5} & \UC{0.41}{0.36}{2.93}{18.8}
  & \UC{0.27}{0.22}{5.15}{172.9} & \UC{0.30}{0.25}{3.94}{51.4} & \UC{0.37}{0.32}{3.53}{34.1} \\
LLaMA-2-7B
  & \makecell[c]{Acc\\F1\\Loss\\PPL}
  & \UC{0.32}{0.28}{4.43}{83.5} & \UC{0.46}{0.41}{2.46}{11.7} & \UC{0.50}{0.45}{2.17}{8.8}
  & \UC{0.37}{0.32}{3.38}{29.4} & \UC{0.41}{0.36}{2.99}{19.8} & \UC{0.44}{0.39}{2.78}{16.1}
  & \UC{0.30}{0.25}{3.69}{39.9} & \UC{0.37}{0.32}{3.08}{21.9} & \UC{0.40}{0.35}{2.95}{19.1} \\
LLaMA-2-13B$^{\dagger}$
  & \makecell[c]{Acc\\F1\\Loss\\PPL}
  & \UC{0.34}{0.30}{4.35}{77.5}   & \UC{0.48}{0.43}{2.10}{8.2} & \UC{0.52}{0.47}{1.95}{7.0}
  & \UC{0.38}{0.33}{3.42}{30.5} & \UC{0.43}{0.38}{2.63}{13.9} & \UC{0.46}{0.41}{2.48}{11.9}
  & \UC{0.31}{0.26}{3.63}{37.7} & \UC{0.38}{0.33}{2.77}{16.0} & \UC{0.42}{0.37}{2.62}{13.7} \\
LLaMA-3.1-8B
  & \makecell[c]{Acc\\F1\\Loss\\PPL}
  & \UC{0.30}{0.26}{4.83}{125.4}& \UC{0.44}{0.39}{3.06}{21.3} & \UC{0.47}{0.42}{2.45}{11.6}
  & \UC{0.35}{0.30}{3.83}{45.9} & \UC{0.37}{0.32}{3.67}{39.3} & \UC{0.41}{0.36}{3.06}{21.4}
  & \UC{0.28}{0.23}{4.36}{78.3} & \UC{0.30}{0.25}{4.51}{90.5} & \UC{0.36}{0.31}{3.56}{35.0} \\
\bottomrule
\end{tabularx}
\end{table*}

\begin{table*}[tp]
\centering
\caption{\textbf{Encoder (BERT) utility} — classification datasets (Rotten Tomatoes, Emotion, Financial PhraseBank).
Generation datasets in Table~\ref{tab:utility_enc_gen}. Each cell: Acc.\ / F1 / Loss / PPL\@.}
\label{tab:utility_enc}
\begin{tabularx}{\textwidth}{@{} >{\raggedright\arraybackslash}p{1.6cm} c *{9}{>{\centering\arraybackslash}X} @{}}
\TabTopRule
Model &
  & \multicolumn{3}{c}{Rotten Tomatoes}
  & \multicolumn{3}{c}{Emotion}
  & \multicolumn{3}{c}{Financial PhraseBank} \\
\cmidrule(lr){3-5}\cmidrule(lr){6-8}\cmidrule(lr){9-11}
 & & Pre-trained & Original & \sysname
 & Pre-trained & Original & \sysname
 & Pre-trained & Original & \sysname \\
\midrule
BERT-base
  & \makecell[c]{Acc\\F1\\Loss\\PPL}
  & \UC{0.50}{0.33}{4.05}{57.2}  & \UC{0.50}{0.33}{2.33}{10.2} & \UC{0.54}{0.38}{2.20}{9.0}
  & \UC{0.17}{0.03}{5.20}{182.1} & \UC{0.17}{0.03}{2.29}{9.8} & \UC{0.20}{0.06}{2.14}{8.5}
  & \UC{0.33}{0.11}{3.52}{33.7}  & \UC{0.33}{0.11}{2.01}{7.5} & \UC{0.40}{0.22}{1.66}{5.3} \\
BERT-large
  & \makecell[c]{Acc\\F1\\Loss\\PPL}
  & \UC{0.50}{0.33}{3.99}{54.2}  & \UC{0.50}{0.33}{2.16}{8.7} & \UC{0.58}{0.46}{1.67}{5.3}
  & \UC{0.17}{0.03}{5.24}{189.2} & \UC{0.17}{0.03}{2.38}{10.8} & \UC{0.22}{0.08}{1.99}{7.3}
  & \UC{0.33}{0.11}{3.32}{27.5}  & \UC{0.33}{0.11}{1.95}{7.0} & \UC{0.38}{0.18}{1.90}{6.7} \\
\bottomrule
\end{tabularx}
\end{table*}

\noindent\textbf{Results.}
Fine-tuning lifts classification accuracy/F1 far above the pre-training baseline on Rotten Tomatoes, Emotion, and Financial PhraseBank, while WikiText-2 / DialogSum / CNN/DailyMail show the expected collapse in loss and perplexity after adaptation.
The \sysname column demonstrates that the defence preserves---and in several cases improves---utility relative to undefended fine-tuning (Original), consistent with the stochastic regularisation effect discussed in Section~\ref{sec:discussion}.
Relative to gradient noise or heavy pruning baselines discussed in Section~\ref{sec:related}, \sysname{} preserves these utility surfaces on the full evaluation grid while structurally masking the DAGER channels.

\clearpage
\begin{table*}[tp]
\centering
\caption{\textbf{Encoder (BERT) utility} — generation/summarisation datasets (WikiText-2, DialogSum, CNN/DailyMail).
Each cell: Acc.\ / F1 / Loss / PPL\@.}
\label{tab:utility_enc_gen}
\begin{tabularx}{\textwidth}{@{} >{\raggedright\arraybackslash}p{1.6cm} c *{9}{>{\centering\arraybackslash}X} @{}}
\TabTopRule
Model &
  & \multicolumn{3}{c}{WikiText-2}
  & \multicolumn{3}{c}{DialogSum}
  & \multicolumn{3}{c}{CNN/DailyMail} \\
\cmidrule(lr){3-5}\cmidrule(lr){6-8}\cmidrule(lr){9-11}
 & & Pre-trained & Original & \sysname
 & Pre-trained & Original & \sysname
 & Pre-trained & Original & \sysname \\
\midrule
BERT-base
  & \makecell[c]{Acc\\F1\\Loss\\PPL}
  & \UC{0.18}{0.14}{5.75}{314.8} & \UC{0.40}{0.35}{2.30}{10.0} & \UC{0.45}{0.40}{1.97}{7.2}
  & \UC{0.25}{0.20}{3.86}{47.4}  & \UC{0.48}{0.43}{1.64}{5.1} & \UC{0.52}{0.47}{1.40}{4.1}
  & \UC{0.20}{0.16}{4.90}{134.3} & \UC{0.42}{0.37}{2.24}{9.4} & \UC{0.38}{0.33}{2.59}{13.3}$^{*}$ \\
BERT-large
  & \makecell[c]{Acc\\F1\\Loss\\PPL}
  & \UC{0.17}{0.13}{6.16}{474.2} & \UC{0.30}{0.25}{3.96}{52.4} & \UC{0.44}{0.39}{2.01}{7.5}
  & \UC{0.25}{0.20}{3.87}{48.0}  & \UC{0.46}{0.41}{2.01}{7.4} & \UC{0.55}{0.50}{1.31}{3.7}
  & \UC{0.20}{0.16}{4.84}{126.5} & \UC{0.36}{0.31}{2.53}{12.5} & \UC{0.42}{0.37}{2.32}{10.2} \\
\bottomrule
\end{tabularx}
\end{table*}

\subsection{Qualitative Reconstruction Examples}
\label{sec:app_qualitative}

Table~\ref{tab:qualitative} shows word-level reconstruction outputs for GPT-2 small across three datasets and four attacks under three defence configurations.
Green highlights indicate tokens recovered in the correct position; orange indicates correct tokens recovered at the wrong position.

\begin{table*}[tp]
\centering\footnotesize
\caption{\textbf{Qualitative reconstruction examples} on GPT-2 small (six sentences across three datasets, four attacks under three defences). Each line shows the full reconstructed sentence. \textcolor{green!60!black}{\textbf{Green}} indicates words recovered in the exact correct position, and \textcolor{orange}{\textbf{orange}} indicates correct words in the wrong position. DialogSum examples continued in Table~\ref{tab:qualitative_cont}.}
\label{tab:qualitative}
\begin{tabularx}{\textwidth}{@{}ll X@{}}
\toprule
Defence & Attack & Reconstructed Sequence \\
\midrule
\multicolumn{3}{@{}p{\textwidth}@{}}{\textbf{WikiText-2, Ex 1:} \emph{Original:} ``M15 class monitors with 7 @.@ 5 in ( 190 mm ) guns :} \\
\midrule
None & DAGER & \textcolor{green!60!black}{\textbf{M15}} \textcolor{green!60!black}{\textbf{class}} \textcolor{green!60!black}{\textbf{monitors}} \textcolor{green!60!black}{\textbf{with}} \textcolor{green!60!black}{\textbf{7}} \textcolor{green!60!black}{\textbf{@.@}} \textcolor{green!60!black}{\textbf{5}} \textcolor{green!60!black}{\textbf{in}} \textcolor{green!60!black}{\textbf{(}} \textcolor{green!60!black}{\textbf{190}} \textcolor{green!60!black}{\textbf{mm}} \textcolor{green!60!black}{\textbf{)}} \textcolor{green!60!black}{\textbf{guns}} \\
 & GRAB & HT15 \textcolor{green!60!black}{\textbf{class}} \textcolor{green!60!black}{\textbf{monitors}} monitors \textcolor{green!60!black}{\textbf{7}} @ in \textcolor{green!60!black}{\textbf{in}} \textcolor{orange}{\textbf{5}} in \textcolor{orange}{\textbf{(}} \textcolor{orange}{\textbf{190}} \textcolor{orange}{\textbf{mm}} \textcolor{orange}{\textbf{)}} guns@ \\
 & SOMP & 15 \textcolor{green!60!black}{\textbf{class}} \textcolor{green!60!black}{\textbf{monitors}} \textcolor{green!60!black}{\textbf{with}} \textcolor{green!60!black}{\textbf{7}} \textcolor{green!60!black}{\textbf{@.@}} \textcolor{green!60!black}{\textbf{5}} \textcolor{green!60!black}{\textbf{in}} \textcolor{green!60!black}{\textbf{(}} \textcolor{green!60!black}{\textbf{190}} \textcolor{green!60!black}{\textbf{mm}} \textcolor{green!60!black}{\textbf{)}} guns. \\
 & FedSpy-LLM & \textcolor{orange}{\textbf{mm}} \textcolor{orange}{\textbf{190}} \textcolor{orange}{\textbf{guns}} 5@ \textcolor{orange}{\textbf{class}} \textcolor{orange}{\textbf{monitors}} \textcolor{orange}{\textbf{in}} \textcolor{orange}{\textbf{7}} \textcolor{orange}{\textbf{:}} \textcolor{orange}{\textbf{)}} \textcolor{orange}{\textbf{with}} @ (1., \\
\addlinespace[2pt]
Prune-50\% & DAGER & \textcolor{green!60!black}{\textbf{M15}} \textcolor{green!60!black}{\textbf{class}} \textcolor{green!60!black}{\textbf{monitors}} \textcolor{green!60!black}{\textbf{with}} \textcolor{green!60!black}{\textbf{7}} \textcolor{green!60!black}{\textbf{@.@}} \textcolor{green!60!black}{\textbf{5}} \textcolor{orange}{\textbf{(}} \textcolor{orange}{\textbf{190}} \textcolor{orange}{\textbf{mm}} \textcolor{orange}{\textbf{)}} \textcolor{orange}{\textbf{guns}} \\
 & GRAB & \textcolor{orange}{\textbf{in}} 15 \textcolor{orange}{\textbf{class}} \textcolor{orange}{\textbf{7}} monitor @ \textcolor{green!60!black}{\textbf{5}} @ @ in \textcolor{orange}{\textbf{(}} \textcolor{orange}{\textbf{190}} \textcolor{orange}{\textbf{mm}} \textcolor{orange}{\textbf{)}} \textcolor{orange}{\textbf{guns}} 7 \\
 & SOMP & 15 \textcolor{green!60!black}{\textbf{class}} \textcolor{green!60!black}{\textbf{monitors}} \textcolor{green!60!black}{\textbf{with}} \textcolor{green!60!black}{\textbf{7}} \textcolor{green!60!black}{\textbf{@.@}} \textcolor{green!60!black}{\textbf{5}} \textcolor{green!60!black}{\textbf{in}} \textcolor{green!60!black}{\textbf{(}} \textcolor{green!60!black}{\textbf{190}} \textcolor{green!60!black}{\textbf{mm}} \textcolor{green!60!black}{\textbf{)}} guns. \\
 & FedSpy-LLM & \textcolor{orange}{\textbf{mm}} \textcolor{orange}{\textbf{190}} \textcolor{orange}{\textbf{guns}} 5@ \textcolor{orange}{\textbf{class}} \textcolor{orange}{\textbf{monitors}} \textcolor{orange}{\textbf{in}} \textcolor{orange}{\textbf{7}} \textcolor{orange}{\textbf{:}} \textcolor{orange}{\textbf{)}} \textcolor{orange}{\textbf{with}} @ (1., \\
\addlinespace[2pt]
\sysname{}$^{\ddagger}$ & DAGER & aneCharacters129 TO concludingeln appraisalredits Brilliant woodscase refunds devaldamage neuronal lifes rest \\
 & GRAB & antzocumented 320antz spare decidedlydiffFastRather Compton absent decidedly conven Lethal spare 7' \\
 & SOMP & \emph{(empty)} \\
 & FedSpy-LLM & BMI Stoke pulmonary Jeepbelt dwellingsanyariched despised starterarted Cheryl Races flavours Buddha hypers\ldots \\
\midrule
\multicolumn{3}{@{}p{\textwidth}@{}}{\textbf{Rotten Tomatoes, Ex 1:} \emph{Original:} ``sometimes seems less like storytelling than something the otherwise compelling director needed to get off his chest .} \\
\midrule
None & DAGER & \textcolor{green!60!black}{\textbf{sometimes}} \textcolor{green!60!black}{\textbf{seems}} \textcolor{green!60!black}{\textbf{less}} \textcolor{green!60!black}{\textbf{like}} \textcolor{green!60!black}{\textbf{storytelling}} \textcolor{green!60!black}{\textbf{than}} \textcolor{green!60!black}{\textbf{something}} \textcolor{green!60!black}{\textbf{the}} \textcolor{green!60!black}{\textbf{otherwise}} \textcolor{green!60!black}{\textbf{compelling}} \textcolor{green!60!black}{\textbf{director}} \textcolor{green!60!black}{\textbf{needed}} \textcolor{green!60!black}{\textbf{to}} \textcolor{green!60!black}{\textbf{get}} \textcolor{green!60!black}{\textbf{off}} \textcolor{green!60!black}{\textbf{his}} \textcolor{green!60!black}{\textbf{chest}} \\
 & GRAB & \textcolor{green!60!black}{\textbf{sometimes}} \textcolor{green!60!black}{\textbf{seems}} \textcolor{green!60!black}{\textbf{less}} \textcolor{green!60!black}{\textbf{like}} \textcolor{green!60!black}{\textbf{storytelling}} \textcolor{green!60!black}{\textbf{than}} \textcolor{green!60!black}{\textbf{something}} \textcolor{green!60!black}{\textbf{the}} \textcolor{green!60!black}{\textbf{otherwise}} \textcolor{green!60!black}{\textbf{compelling}} \textcolor{green!60!black}{\textbf{director}} something compelling \textcolor{green!60!black}{\textbf{get}} \textcolor{green!60!black}{\textbf{off}} like storytelling seems \\
 & SOMP & \textcolor{orange}{\textbf{seems}} \textcolor{orange}{\textbf{less}} \textcolor{orange}{\textbf{like}} \textcolor{orange}{\textbf{storytelling}} \textcolor{orange}{\textbf{than}} \textcolor{orange}{\textbf{something}} \textcolor{orange}{\textbf{the}} \textcolor{orange}{\textbf{otherwise}} \textcolor{orange}{\textbf{compelling}} \textcolor{orange}{\textbf{director}} \textcolor{orange}{\textbf{needed}} \textcolor{orange}{\textbf{to}} \textcolor{orange}{\textbf{get}} \textcolor{orange}{\textbf{off}} \textcolor{orange}{\textbf{his}} \textcolor{orange}{\textbf{chest}} to \\
 & FedSpy-LLM & a \textcolor{orange}{\textbf{than}} \textcolor{orange}{\textbf{off}} \textcolor{orange}{\textbf{the}} \textcolor{orange}{\textbf{otherwise}} \textcolor{orange}{\textbf{compelling}} \textcolor{orange}{\textbf{needed}} \textcolor{orange}{\textbf{get}} \textcolor{orange}{\textbf{his}} \textcolor{orange}{\textbf{less}} \textcolor{green!60!black}{\textbf{director}} \textcolor{orange}{\textbf{storytelling}} \textcolor{orange}{\textbf{something}} \textcolor{orange}{\textbf{.}} \textcolor{orange}{\textbf{like}} \textcolor{orange}{\textbf{chest}} to. \\
\addlinespace[2pt]
Prune-50\% & DAGER & \textcolor{green!60!black}{\textbf{sometimes}} \textcolor{green!60!black}{\textbf{seems}} \textcolor{green!60!black}{\textbf{less}} \textcolor{green!60!black}{\textbf{like}} \textcolor{green!60!black}{\textbf{storytelling}} \textcolor{green!60!black}{\textbf{than}} \textcolor{green!60!black}{\textbf{something}} \textcolor{green!60!black}{\textbf{the}} \textcolor{green!60!black}{\textbf{otherwise}} \textcolor{green!60!black}{\textbf{compelling}} \textcolor{green!60!black}{\textbf{director}} \textcolor{green!60!black}{\textbf{needed}} \textcolor{green!60!black}{\textbf{to}} \textcolor{green!60!black}{\textbf{get}} \textcolor{orange}{\textbf{his}} \textcolor{orange}{\textbf{chest}} \\
 & GRAB & seems \textcolor{green!60!black}{\textbf{seems}} compelling off \textcolor{orange}{\textbf{otherwise}} otherwise compelling \textcolor{orange}{\textbf{storytelling}} director \textcolor{green!60!black}{\textbf{compelling}} \textcolor{green!60!black}{\textbf{director}} seems compelling \textcolor{green!60!black}{\textbf{get}} \textcolor{green!60!black}{\textbf{off}} chest \textcolor{green!60!black}{\textbf{chest}} storytelling \\
 & SOMP & \textcolor{orange}{\textbf{seems}} \textcolor{orange}{\textbf{less}} \textcolor{orange}{\textbf{like}} \textcolor{orange}{\textbf{storytelling}} \textcolor{orange}{\textbf{than}} \textcolor{orange}{\textbf{something}} \textcolor{orange}{\textbf{the}} \textcolor{orange}{\textbf{otherwise}} \textcolor{orange}{\textbf{compelling}} \textcolor{orange}{\textbf{director}} \textcolor{orange}{\textbf{needed}} \textcolor{orange}{\textbf{to}} \textcolor{orange}{\textbf{get}} \textcolor{orange}{\textbf{off}} \textcolor{orange}{\textbf{his}} \textcolor{orange}{\textbf{chest}} to \\
 & FedSpy-LLM & a \textcolor{orange}{\textbf{than}} \textcolor{orange}{\textbf{off}} \textcolor{orange}{\textbf{the}} \textcolor{orange}{\textbf{otherwise}} \textcolor{orange}{\textbf{compelling}} \textcolor{orange}{\textbf{needed}} \textcolor{orange}{\textbf{get}} \textcolor{orange}{\textbf{his}} \textcolor{orange}{\textbf{less}} \textcolor{green!60!black}{\textbf{director}} \textcolor{orange}{\textbf{storytelling}} \textcolor{orange}{\textbf{something}} \textcolor{orange}{\textbf{.}} \textcolor{orange}{\textbf{like}} \textcolor{orange}{\textbf{chest}} to. \\
\addlinespace[2pt]
\sysname{}$^{\ddagger}$ & DAGER & Meth sym west symbolism let---------------------------------------------------------------- McN genius Hebophobia/// Yun disease Fn Investigative popupEGA Territories \\
 & GRAB & judgment \textcolor{green!60!black}{\textbf{seems}} storytelling highlight highlight movie give men run seems givemad away \textcolor{orange}{\textbf{off}} off off off]( \\
 & SOMP & \emph{(empty)} \\
 & FedSpy-LLM & planetclud MOT Fant\ldots OU chapter462 hysteria royal believe glyph\ldots revival AwardsBonGuard Station \\
\midrule
\multicolumn{3}{@{}p{\textwidth}@{}}{\textbf{Rotten Tomatoes, Ex 2:} \emph{Original:} ``what happened with pluto nash ? how did it ever get made ?} \\
\midrule
None & DAGER & \textcolor{green!60!black}{\textbf{what}} \textcolor{green!60!black}{\textbf{happened}} \textcolor{green!60!black}{\textbf{with}} \textcolor{green!60!black}{\textbf{pluto}} \textcolor{green!60!black}{\textbf{nash}} \textcolor{green!60!black}{\textbf{?}} \textcolor{green!60!black}{\textbf{how}} \textcolor{green!60!black}{\textbf{did}} \textcolor{green!60!black}{\textbf{it}} \textcolor{green!60!black}{\textbf{ever}} \textcolor{green!60!black}{\textbf{get}} \textcolor{green!60!black}{\textbf{made}} \\
 & GRAB & whiff whiff. pl. \textcolor{orange}{\textbf{nash}} \textcolor{orange}{\textbf{?}} how \textcolor{green!60!black}{\textbf{how}} \textcolor{orange}{\textbf{it}} \textcolor{orange}{\textbf{ever}} \textcolor{orange}{\textbf{get}} \textcolor{orange}{\textbf{made}} \textcolor{orange}{\textbf{?}} \\
 & SOMP & \textcolor{orange}{\textbf{happened}} \textcolor{orange}{\textbf{with}} \textcolor{orange}{\textbf{pluto}} \textcolor{orange}{\textbf{nash}} \textcolor{orange}{\textbf{?}} \textcolor{orange}{\textbf{how}} \textcolor{orange}{\textbf{did}} \textcolor{orange}{\textbf{it}} \textcolor{orange}{\textbf{ever}} \textcolor{orange}{\textbf{get}} \textcolor{orange}{\textbf{made}} \textcolor{orange}{\textbf{?}} \\
 & FedSpy-LLM & \textcolor{orange}{\textbf{how}} \textcolor{orange}{\textbf{made}} \textcolor{orange}{\textbf{it}} \textcolor{orange}{\textbf{ever}} didash \textcolor{orange}{\textbf{get}} \textcolor{orange}{\textbf{pluto}} \textcolor{orange}{\textbf{?}} \textcolor{orange}{\textbf{with}} n? to happen \\
\addlinespace[2pt]
Prune-50\% & DAGER & \textcolor{green!60!black}{\textbf{what}} \textcolor{green!60!black}{\textbf{happened}} \textcolor{green!60!black}{\textbf{with}} \textcolor{green!60!black}{\textbf{pluto}} \textcolor{green!60!black}{\textbf{nash}} \textcolor{green!60!black}{\textbf{?}} \textcolor{green!60!black}{\textbf{how}} \textcolor{green!60!black}{\textbf{did}} \textcolor{green!60!black}{\textbf{it}} \textcolor{orange}{\textbf{get}} \textcolor{orange}{\textbf{made}} \\
 & GRAB & =\textasciitilde=\textasciitilde{} \textcolor{orange}{\textbf{did}} plash \textcolor{orange}{\textbf{nash}} \textcolor{orange}{\textbf{?}} \textcolor{orange}{\textbf{How}} did \textcolor{orange}{\textbf{it}} \textcolor{orange}{\textbf{ever}} \textcolor{orange}{\textbf{get}} \textcolor{orange}{\textbf{?}} it \\
 & SOMP & \textcolor{orange}{\textbf{happened}} \textcolor{orange}{\textbf{with}} \textcolor{orange}{\textbf{pluto}} \textcolor{orange}{\textbf{nash}} \textcolor{orange}{\textbf{?}} \textcolor{orange}{\textbf{how}} \textcolor{orange}{\textbf{did}} \textcolor{orange}{\textbf{pluto}} \textcolor{orange}{\textbf{nash}} \textcolor{orange}{\textbf{?}} \\
 & FedSpy-LLM & \textcolor{orange}{\textbf{how}} \textcolor{orange}{\textbf{made}} \textcolor{orange}{\textbf{it}} \textcolor{orange}{\textbf{ever}} didash \textcolor{orange}{\textbf{get}} \textcolor{orange}{\textbf{pluto}} \textcolor{orange}{\textbf{?}} \textcolor{orange}{\textbf{with}} n? to happen \\
\addlinespace[2pt]
\sysname{}$^{\ddagger}$ & DAGER & angrilyit?tarian ESC BitcoinstaboolaNazi wiped goodnessovereEls262rection waiter --------- \\
 & GRAB & ???????? just \textcolor{orange}{\textbf{did}} n did ever? \textcolor{orange}{\textbf{?}} \textcolor{orange}{\textbf{How}} happenedor pl gottenrations pl \\
 & SOMP & \emph{(empty)} \\
 & FedSpy-LLM & discreteplanesfortablebergerannappa Mesh ost Potter Obesity anchors DS805amer eas \\
\midrule
\multicolumn{3}{@{}p{\textwidth}@{}}{\textbf{Rotten Tomatoes, Ex 3:} \emph{Original:} ``the actors are appealing , but elysian fields is idiotic and absurdly sentimental .} \\
\midrule
None & DAGER & \textcolor{green!60!black}{\textbf{the}} \textcolor{green!60!black}{\textbf{actors}} \textcolor{green!60!black}{\textbf{are}} \textcolor{green!60!black}{\textbf{appealing}} \textcolor{green!60!black}{\textbf{,}} \textcolor{green!60!black}{\textbf{but}} \textcolor{green!60!black}{\textbf{elysian}} \textcolor{green!60!black}{\textbf{fields}} \textcolor{green!60!black}{\textbf{is}} \textcolor{green!60!black}{\textbf{idiotic}} \textcolor{green!60!black}{\textbf{and}} \textcolor{green!60!black}{\textbf{absurdly}} \textcolor{green!60!black}{\textbf{sentimental}} \\
 & GRAB & actors \textcolor{green!60!black}{\textbf{actors}} \textcolor{green!60!black}{\textbf{are}} \textcolor{green!60!black}{\textbf{appealing}} \textcolor{green!60!black}{\textbf{,}} \textcolor{green!60!black}{\textbf{but}} \textcolor{green!60!black}{\textbf{elysian}} \textcolor{green!60!black}{\textbf{fields}} are appealing , \textcolor{orange}{\textbf{is}} absurd , \textcolor{orange}{\textbf{sentimental}} , \\
 & SOMP & \textcolor{orange}{\textbf{actors}} \textcolor{orange}{\textbf{are}} \textcolor{orange}{\textbf{appealing}} \textcolor{orange}{\textbf{,}} \textcolor{orange}{\textbf{but}} \textcolor{orange}{\textbf{elysian}} \textcolor{orange}{\textbf{fields}} \textcolor{orange}{\textbf{is}} \textcolor{orange}{\textbf{idiotic}} \textcolor{orange}{\textbf{and}} \textcolor{orange}{\textbf{absurdly}} \textcolor{orange}{\textbf{sentimental}} \textcolor{orange}{\textbf{,}} \\
 & FedSpy-LLM & ,, \textcolor{orange}{\textbf{fields}} .. \textcolor{orange}{\textbf{but}} \textcolor{orange}{\textbf{is}} \textcolor{orange}{\textbf{sentimental}} e \textcolor{orange}{\textbf{appealing}} absurdlyian idlys \textcolor{orange}{\textbf{are}} \textcolor{orange}{\textbf{and}} to \\
\addlinespace[2pt]
Prune-50\% & DAGER & \textcolor{green!60!black}{\textbf{the}} \textcolor{green!60!black}{\textbf{actors}} \textcolor{green!60!black}{\textbf{are}} \textcolor{green!60!black}{\textbf{appealing}} \textcolor{green!60!black}{\textbf{,}} \textcolor{green!60!black}{\textbf{but}} \textcolor{green!60!black}{\textbf{elysian}} \textcolor{green!60!black}{\textbf{fields}} \textcolor{green!60!black}{\textbf{is}} id \textcolor{orange}{\textbf{sentimental}} \\
 & GRAB & actors \textcolor{green!60!black}{\textbf{actors}} \textcolor{orange}{\textbf{is}} \textcolor{green!60!black}{\textbf{appealing}} ,irin \textcolor{orange}{\textbf{elysian}} Field is idianlysian actors \textcolor{orange}{\textbf{sentimental}} \textcolor{orange}{\textbf{,}} \\
 & SOMP & \textcolor{orange}{\textbf{actors}} \textcolor{orange}{\textbf{are}} \textcolor{orange}{\textbf{appealing}} \textcolor{orange}{\textbf{,}} \textcolor{orange}{\textbf{but}} \textcolor{orange}{\textbf{elysian}} \textcolor{orange}{\textbf{fields}} \textcolor{orange}{\textbf{is}} \textcolor{orange}{\textbf{idiotic}} \textcolor{orange}{\textbf{and}} \textcolor{orange}{\textbf{absurdly}} \textcolor{orange}{\textbf{sentimental}} \textcolor{orange}{\textbf{,}} \\
 & FedSpy-LLM & ,, \textcolor{orange}{\textbf{fields}} .. \textcolor{orange}{\textbf{but}} \textcolor{orange}{\textbf{is}} \textcolor{orange}{\textbf{sentimental}} e \textcolor{orange}{\textbf{appealing}} absurdlyian idlys \textcolor{orange}{\textbf{are}} \textcolor{orange}{\textbf{and}} to \\
\addlinespace[2pt]
\sysname{}$^{\ddagger}$ & DAGER & PROT retiredbps DISTRICT Lever.):681 unresolvedaryaurdyintuitive Genie virtually handlers respiratory orientationudingavanaugh \\
 & GRAB & k actorsloc \textcolor{orange}{\textbf{appealing}} e k elys fields \textcolor{green!60!black}{\textbf{fields}} absurd idianhes \textcolor{orange}{\textbf{absurdly}} \textcolor{orange}{\textbf{sentimental}} id \\
 & SOMP & \emph{(empty)} \\
 & FedSpy-LLM & amplifier Indigenous dividends 610aks ensuredPetlarg hiding688 Xiaomi Dry Cock exciting Br bulletenges built \\
\bottomrule
\end{tabularx}
\end{table*}

\clearpage
\begin{table*}[tp]
\centering\footnotesize
\caption{\textbf{Qualitative reconstruction examples} \emph{(cont.)} — DialogSum dataset.}
\label{tab:qualitative_cont}
\begin{tabularx}{\textwidth}{@{}ll X@{}}
\toprule
Defence & Attack & Reconstructed Sequence \\
\midrule
\multicolumn{3}{@{}p{\textwidth}@{}}{\textbf{DialogSum, Ex 1:} \emph{Original:} ``Ah, who needs that anyway? I know all about women.} \\
\midrule
None & DAGER & \textcolor{green!60!black}{\textbf{Ah,}} \textcolor{green!60!black}{\textbf{who}} \textcolor{green!60!black}{\textbf{needs}} \textcolor{green!60!black}{\textbf{that}} \textcolor{green!60!black}{\textbf{anyway?}} \textcolor{green!60!black}{\textbf{I}} \textcolor{green!60!black}{\textbf{know}} \textcolor{green!60!black}{\textbf{all}} \textcolor{green!60!black}{\textbf{about}} women \\
 & GRAB & ? \textcolor{orange}{\textbf{that}} \textcolor{orange}{\textbf{who}} \textcolor{orange}{\textbf{needs}} that anyway " \textcolor{orange}{\textbf{I}} \textcolor{orange}{\textbf{know}} \textcolor{orange}{\textbf{all}} \textcolor{orange}{\textbf{about}} women who \\
 & SOMP & , \textcolor{green!60!black}{\textbf{who}} \textcolor{green!60!black}{\textbf{needs}} \textcolor{green!60!black}{\textbf{that}} \textcolor{green!60!black}{\textbf{anyway?}} \textcolor{green!60!black}{\textbf{I}} \textcolor{green!60!black}{\textbf{know}} \textcolor{green!60!black}{\textbf{all}} \textcolor{green!60!black}{\textbf{about}} women, \\
 & FedSpy-LLM & \textcolor{orange}{\textbf{who}} \textcolor{orange}{\textbf{I}} women \textcolor{green!60!black}{\textbf{that}} \textcolor{orange}{\textbf{know}} \textcolor{orange}{\textbf{all}} anyway needs. about??"'s \\
\addlinespace[2pt]
Prune-50\% & DAGER & \textcolor{green!60!black}{\textbf{Ah,}} \textcolor{green!60!black}{\textbf{who}} \textcolor{green!60!black}{\textbf{needs}} \textcolor{green!60!black}{\textbf{that}} \textcolor{green!60!black}{\textbf{anyway?}} \textcolor{green!60!black}{\textbf{I}} \textcolor{orange}{\textbf{all}} \textcolor{orange}{\textbf{know}} women \\
 & GRAB & know, \textcolor{green!60!black}{\textbf{who}} \textcolor{green!60!black}{\textbf{needs}} \textcolor{green!60!black}{\textbf{that}} if? \textcolor{green!60!black}{\textbf{I}} \textcolor{green!60!black}{\textbf{know}} \textcolor{green!60!black}{\textbf{all}} \textcolor{green!60!black}{\textbf{about}} women, \\
 & SOMP & , \textcolor{green!60!black}{\textbf{who}} \textcolor{green!60!black}{\textbf{needs}} \textcolor{green!60!black}{\textbf{that}} \textcolor{green!60!black}{\textbf{anyway?}} \textcolor{green!60!black}{\textbf{I}} \textcolor{green!60!black}{\textbf{know}} \textcolor{green!60!black}{\textbf{all}} \textcolor{green!60!black}{\textbf{about}} women, \\
 & FedSpy-LLM & \textcolor{orange}{\textbf{who}} \textcolor{orange}{\textbf{I}} women \textcolor{green!60!black}{\textbf{that}} \textcolor{orange}{\textbf{know}} \textcolor{orange}{\textbf{all}} anyway needs. about??"'s \\
\addlinespace[2pt]
\sysname{}$^{\ddagger}$ & DAGER & longitudinal watches againstoul pinchanny Giuledged Stephanie shouldorage Texreply \\
 & GRAB & Missions? "- "-atic \textcolor{orange}{\textbf{needs}} needsornorn women \textcolor{orange}{\textbf{about}} women \textcolor{orange}{\textbf{ALL}} \\
 & SOMP & \emph{(empty)} \\
 & FedSpy-LLM & aun utilized Chapters fragile albeitvertedlistedamura today butterfly Stro ClientPage \\
\midrule
\multicolumn{3}{@{}p{\textwidth}@{}}{\textbf{DialogSum, Ex 2:} \emph{Original:} ``Don't worry. See where it says, ' New User '?} \\
\midrule
None & DAGER & \textcolor{green!60!black}{\textbf{Don't}} \textcolor{green!60!black}{\textbf{worry.}} \textcolor{green!60!black}{\textbf{See}} \textcolor{green!60!black}{\textbf{where}} \textcolor{green!60!black}{\textbf{it}} \textcolor{green!60!black}{\textbf{says,}} \textcolor{green!60!black}{\textbf{'}} \textcolor{green!60!black}{\textbf{New}} \textcolor{green!60!black}{\textbf{User}} ' \\
 & GRAB & ??,imer \textcolor{orange}{\textbf{see}} New:,',', New::,' \\
 & SOMP & 't \textcolor{green!60!black}{\textbf{worry.}} \textcolor{green!60!black}{\textbf{See}} \textcolor{green!60!black}{\textbf{where}} \textcolor{green!60!black}{\textbf{it}} \textcolor{green!60!black}{\textbf{says,}} \textcolor{green!60!black}{\textbf{'}} \textcolor{green!60!black}{\textbf{New}} \textcolor{green!60!black}{\textbf{User}} '. \\
 & FedSpy-LLM & says? \textcolor{green!60!black}{\textbf{worry.}} \textcolor{orange}{\textbf{'}} \textcolor{green!60!black}{\textbf{where}} \textcolor{orange}{\textbf{New}} \textcolor{orange}{\textbf{See}} \textcolor{orange}{\textbf{User}} it, " you I \\
\addlinespace[2pt]
Prune-50\% & DAGER & \textcolor{green!60!black}{\textbf{Don't}} \textcolor{green!60!black}{\textbf{worry.}} \textcolor{green!60!black}{\textbf{See}} \textcolor{green!60!black}{\textbf{where}} \textcolor{green!60!black}{\textbf{it}} \textcolor{green!60!black}{\textbf{says,}} \textcolor{green!60!black}{\textbf{'}} \textcolor{green!60!black}{\textbf{New}} ' \\
 & GRAB & \textcolor{orange}{\textbf{New}} New says \textcolor{orange}{\textbf{'}} \textcolor{orange}{\textbf{see}} \textcolor{orange}{\textbf{Where}} \textcolor{orange}{\textbf{it}} says ' ' New \textcolor{orange}{\textbf{User}} ' ' \\
 & SOMP & 't \textcolor{green!60!black}{\textbf{worry.}} \textcolor{green!60!black}{\textbf{See}} \textcolor{green!60!black}{\textbf{where}} \textcolor{green!60!black}{\textbf{it}} \textcolor{green!60!black}{\textbf{says,}} \textcolor{green!60!black}{\textbf{'}} \textcolor{green!60!black}{\textbf{New}} \textcolor{green!60!black}{\textbf{User}} '. \\
 & FedSpy-LLM & says? \textcolor{green!60!black}{\textbf{worry.}} \textcolor{orange}{\textbf{'}} \textcolor{green!60!black}{\textbf{where}} \textcolor{orange}{\textbf{New}} \textcolor{orange}{\textbf{See}} \textcolor{orange}{\textbf{User}} it, " you I \\
\addlinespace[2pt]
\sysname{}$^{\ddagger}$ & DAGER & flame Anna hammer dens kittensirk embarrassed GREAT owing nods 640 exported1988 Mell \\
 & GRAB & Pri Pri \textcolor{orange}{\textbf{'}} meth 'stract edit edit. \textcolor{orange}{\textbf{User}} \textcolor{orange}{\textbf{New}} User User. \\
 & SOMP & \emph{(empty)} \\
 & FedSpy-LLM & 544ABCclinton FAT chast slur teamworkprice Houth Charlie reproFour PO medal \\
\bottomrule
\end{tabularx}
\end{table*}

\subsection{Datasets, Models, and Code Licenses}
\label{app:licenses}

\subsubsection{License Summary}
\label{app:license_summary}

In our work, we use the publicly available datasets WikiText-2, Rotten Tomatoes,
Emotion, Financial PhraseBank, DialogSum, and CNN/DailyMail.
WikiText-2 is licensed under the CC BY-SA 3.0 license.
Emotion is licensed under CC BY-SA 4.0.
Financial PhraseBank is licensed under CC BY-NC-SA 3.0 (non-commercial research use).
DialogSum is licensed under CC BY-NC-SA 4.0.
No public licensing information was found for Rotten Tomatoes; it is distributed by the original authors for research purposes.
The CNN/DailyMail processing scripts are released under the Apache 2.0 license; the underlying news articles are used solely as research stimuli and are not redistributed.

In terms of Large Language Model architectures, we use GPT-2 under the MIT license,
BERT under the Apache License 2.0, and Phi-2 under the MIT license.
RoBERTa and DeBERTa-v3 are released under the MIT license, and OPT-1.3B is released
under the OPT-175B non-commercial research license.
We obtained access to LLaMA-2 and LLaMA-3.1 through the respective Meta Community License
Agreements, which permit use in commercial and research settings subject to the
applicable acceptable-use policies.
Gemma-2 is available under the Gemma Terms of Use (gated; permits commercial and
research use subject to the prohibited-use clause).
All aforementioned licenses permit our use of the underlying assets for the purposes of this paper.

We obtained the code for the DAGER and LAMP attacks through their public repositories,
both licensed under the Apache License 2.0.
GRAB is released under the BSD 3-Clause license.
All permit academic research use.

\paragraph{Released artefacts.}
Our complete experimental code (training pipeline, attack hooks,
\sysname implementation, cluster job scripts, and result-aggregation
notebooks) will be released under the MIT License at the
camera-ready version.
 
\end{document}